%% file: main.tex
\documentclass[11pt]{article}

\usepackage{booktabs}

\usepackage{amsthm,amsmath,amssymb}
\usepackage[bookmarks=true,hypertexnames=false,pagebackref]{hyperref}

\usepackage{thmtools}
\usepackage{thm-restate}

\hypersetup{colorlinks=true, citecolor=blue, linkcolor=red,
  urlcolor=blue}
\usepackage{newpxtext}
\usepackage[libertine,vvarbb]{newtxmath}

\usepackage{fullpage}
\usepackage{hyperref}
\usepackage{cleveref}
\usepackage[utf8]{inputenc}
\usepackage{framed}
\usepackage{enumerate}
\usepackage{graphicx}
\usepackage{float} 
\usepackage{subfigure}

\usepackage[ruled,vlined]{algorithm2e}
\usepackage{algorithmic}

\usepackage{tikz}
\usepackage{complexity}
\usepackage{ulem}

\newtheorem{theorem}{Theorem}[section]
\newtheorem{claim}[theorem]{Claim}
\newtheorem{fact}[theorem]{Fact}
\newtheorem{corollary}[theorem]{Corollary}
\newtheorem{lemma}[theorem]{Lemma}
\newtheorem{definition}[theorem]{Definition}

\newtheorem{problem}[theorem]{Problem}

\newtheorem{remark}[theorem]{Remark}

\newtheorem{proposition}[theorem]{Proposition}

\newcommand{\calD}{\mathcal{D}}

\newcommand{\calT}{\mathcal{T}}

\newcommand{\calX}{\mathcal{X}}

\newcommand{\sfM}{\mathsf{M}}
\newcommand{\bbE}{\mathbb E}
\newcommand{\Var}{\mathbb Var}
\newcommand{\Covar}{\mathbb Cov}
\newcommand{\Ent}{\mathrm H}
\newcommand{\Ber}{\text{Ber}}

\newcommand{\Du}{D_{\mathrm{uniform}}}
\newcommand{\Dp}{D_{\mathrm{planted}}}
\newcommand{\Dn}{D_{\mathrm{null}}}
\newcommand{\pt}{P_{\mathrm{trunc}}}

\newcommand{\X}{X}
\newcommand{\mic}{MIC}
\newcommand{\Dis}{\mathsf{DP}}

\newcommand{\no}{D_{0}}
\newcommand{\yes}{D_{1}}
\newcommand{\allone}{\overrightarrow{1}}

\newcommand{\bth}{\boldsymbol{\theta}}

\newcommand{\bb}{\mathbf{b}}
\renewcommand{\polylog}{\text{polylog}}
\newcommand{\sparsity}{\ell}
\newcommand{\Eb}{\mathbb E}

\newcommand{\zerob}{\mathbf{0}}
\newcommand{\oneb}{\mathbf{1}}
\newcommand{\eb}{\mathbf{e}}
\newcommand{\I}{\mathrm I}
\newcommand{\PiR}{\boldsymbol{\Pi}}
\newcommand{\generalCC}{General Planted Problem}

\renewcommand{\E}{\mathbb{E}}
\renewcommand{\R}{\mathbb{R}}

\newcounter{this-list}

\graphicspath{{img/}}

\title{A Unified Approach to Memory-Sample Tradeoffs for\\ Detecting Planted Structures}

\author{Sumegha Garg\thanks{Rutgers University. Email: \texttt{sumegha.garg@rutgers.edu}}
\and
Jabari Hastings\thanks{Stanford University. Email: \texttt{jabarih@stanford.edu}}
\and
Chirag Pabbaraju\thanks{Stanford University. Email: \texttt{cpabbara@stanford.edu}}
\and
Vatsal Sharan\thanks{University of Southern California. Email: \texttt{vsharan@usc.edu}}
}

\date{}
\allowdisplaybreaks
\begin{document}
\maketitle
\thispagestyle{empty}

\sloppy
\input{abstract}

\newpage
\tableofcontents
\thispagestyle{empty}
\newpage
\setcounter{page}{1}

\input{intro}

\input{proofoverview}
\input{prelims}

\input{framework}

\input{generalized_sdpi}

\input{optbiclique}

\input{semi-random}
\input{gaussian_detection}

\input{sparse_pca}

\normalem

\section*{Acknowledgments}

VS was supported by NSF CAREER Award CCF-$2239265$, an Amazon Research Award, a Google Research Scholar Award and a Okawa Foundation Award. The work was done in part while VS was visiting the Simons Institute for the Theory of Computing. CP was supported by Gregory Valiant's and Moses Charikar's Simons Investigator Awards, and a Google PhD Fellowship. JH is supported by the Simons Foundation Collaboration on the Theory of Algorithmic Fairness and the Simons Foundation investigators award 689988. CP and JH would like to thank Gregory Valiant and Annie Marsden for many insightful discussions on the planted clique problem.

\appendix
\input{appendix-prelims}
\input{appendix-multi-ic}

\input{appendix_planted_clique}
\input{appendix_semi_random}

\input{appendix-gaussian}
\input{appendix-pca}

\bibliographystyle{alpha}
\bibliography{reference}

\end{document}

%% file: abstract.tex
\begin{abstract}

We present a unified framework for proving memory lower bounds for multi‐pass streaming algorithms that detect planted structures. Planted structures --- such as cliques or bicliques in graphs, and sparse signals in high-dimensional data --- arise in numerous applications, and our framework yields multi-pass memory lower bounds for many such fundamental settings. We show memory lower bounds for the planted $k$-biclique detection problem in random bipartite graphs and for detecting  sparse Gaussian means. We also show the first memory-sample tradeoffs for the  sparse principal component analysis (PCA) problem in the spiked covariance model. For all these problems to which we apply our unified framework, we obtain bounds which are nearly tight in the low, $O(\log n)$ memory regime. We also leverage our bounds  to establish new multi-pass streaming lower bounds, in the vertex arrival model, for two well-studied graph streaming problems: approximating the size of the largest biclique  and approximating the maximum density of bounded-size subgraphs.

To show these bounds, we study a general distinguishing problem over matrices, where the goal is to distinguish a null distribution from one that plants an \textit{outlier} distribution over a random submatrix. 
Our analysis builds on a new distributed data processing inequality that provides sufficient conditions for memory hardness in terms of the likelihood ratio between the \emph{averaged} planted and null distributions. This result generalizes the inequality of [Braverman et al., STOC 2016] and may be of independent interest.
The inequality enables us to measure information cost under the null distribution -- a key step for applying subsequent direct-sum-type arguments and incorporating the multi-pass information cost framework of [Braverman et al., STOC 2024].
Finally, to instantiate our framework in concrete settings, we derive bounds on the likelihood ratio between the planted and null distributions using careful truncations.

\end{abstract}

%% file: intro.tex
\section{Introduction}\label{sec:intro}

Many statistical estimation tasks involve discovering certain {hidden structures} in the data distribution. A well-known instance of this is the planted clique problem \cite{jerrum1992large,kuvcera1995expected}, where one is given a random Erd\H{o}s--R\'enyi graph $G(n, 1/2)$ (each edge exists with probability $1/2$) but with a clique added on a uniformly randomly chosen subset of $k$ vertices. The goal is to recover this planted clique. Several  variants of this problem, such as finding the densest subgraph within a graph \cite{chen2016statistical} or finding the presence of certain community structure in the graph \cite{abbe2018community} share a similar ``planted" flavor. Planted structures arise not only in combinatorial problems such as clique detection, but also in classical statistical settings -- for instance, estimating the mean of a high-dimensional Gaussian when the mean vector is known to be 
$k$-sparse \cite{braverman2016communication} or performing dimensionality reduction through sparse principal component analysis (PCA) \cite{zou2006sparse}. Since modern settings often involve large amounts of high-dimensional data with many irrelevant attributes, problems with sparse planted structures capture key challenges in statistical estimation in such settings. 
Other examples include sparse linear regression \cite{steinhardt2015minimax}, sub-matrix detection \cite{ma2015computational} and testing almost $k$-wise independence \cite{alon2007testing}. 

These problems with planted structure also offer a fertile ground to understand average-case complexity, and the interaction of computational and statistical resources. In many of these settings, there is believed to be a gap between what is information-theoretically optimal, and what is possible under computational constraints. The planted clique problem and the sparse PCA problem are among the problems which have been  central objects of study in this line of work. For the planted clique problem, if the clique size $k=\Omega(\sqrt{n})$ then polynomial-time algorithms are known for recovering the clique \cite{alon1998finding,feige2000finding,mcsherry2001spectral}. It is widely conjectured that if $k<O(n^{1/2-\delta})$ for some $\delta>0$, then no polynomial-time algorithms exist for approximately recovering (or detecting) the planted clique. 
Hardness of planted clique implies hardness of a number of problems with planted structure including testing almost $k$-wise independence \cite{alon2007testing}, community detection \cite{brennan2020reducibility}, sub-matrix detection \cite{ma2015computational}, as well as  sparse PCA \cite{berthet2013complexity}---pointing at its fundamental nature for understanding  statistical-computational gaps and average-case complexity. Similarly, the sparse PCA problem which adds a sparsity constraint to the usual PCA problem (discussed further in Section \ref{sec:gaussian_intro}) has played a central role in understanding computational-statistical tradeoffs in statistical settings. 
Its hardness  has been studied from the perspective of sum of squares relaxations \cite{ma2015sum,hopkins2017power}, low-degree likelihood ratio tests \cite{ding2024subexponential}, statistical query algorithms \cite{brennan2021statistical}, robustness to adversarial perturbations \cite{d2020sparse}, failure of approximate message passing  \cite{lesieur2015phase,barbier2020all} and methods from statistical physics  \cite{lesieur2017constrained,arous2023free}.

Our goal in this work is to understand statistical-computational gaps for detecting planted clique, sparse PCA and other problems with planted structures. We consider the streaming model of computation, where the algorithm gets one or more passes over an input drawn from some data-generating distribution. Here, the memory usage of the algorithm is the main metric of computational cost. The streaming model over stochastic inputs is widely studied \cite{guha2007space,andoni2008better,konrad2012maximum,kapralov2014approximating,crouch2016stochastic,raz2018fast,sharan2019memory,braverman2020coin,braverman2024new}, and it captures many modern settings involving massive computation on large graphs or datasets. In addition to its practical relevance,  investigating the role of memory in detecting planted structures offers a complementary vantage point to understand the computational hardness of statistical inference \cite{shamir2014fundamental,steinhardt2016memory,dudeja2024statistical,marsden2024efficient} and, as we  show, also yields new streaming lower bounds for approximation problems on worst-case graphs.

In this work, we develop a general framework for proving memory hardness of detecting planted structures in data, and apply it to several canonical settings ranging from graph problems to learning tasks.  Our first application establishes unconditional statistical-computational tradeoffs for the \emph{planted biclique problem} -- a bipartite generalization of the planted clique problem -- previously studied by \cite{feldman2017statistical} in the context of statistical query hardness for planted clique detection. In this problem,  the goal is to distinguish whether a uniformly random bipartite graph has a $(k \times k)$ biclique planted on a uniformly chosen set of vertices. The problem is at least as hard as the planted clique problem and has been used as a cryptographic primitive \cite{applebaum2010public}. Moreover, most known algorithms and bounds for the planted clique problem naturally extend to the bipartite version~\cite{ames2011nuclear,florescu2016spectral,kumar2022exact,buhai2023algorithms}. In the streaming model, at each time-step the algorithm observes a uniformly random left vertex together with its adjacency list.  \cite{feldman2017statistical} studied the distributional version of the planted biclique problem defined on such adjacency-list vectors.

\begin{problem}\label{prob:dist}
Fix an integer $k$, $1 \le k \le  n$, and a uniformly random subset of $k$ indices $S\subseteq [n]$. The
input distribution $D_S$ on vectors $x\in\{0,1\}^n$ is defined as follows: with probability $1 - (k/n)$, $x$ is
uniform over $\{0,1\}^n$; and with probability $k/n$, $x$ is such that its $k$ coordinates from $S$ are set to $1$,
and the remaining coordinates are uniform in $\{0, 1\}$. Given $m$ independent samples, the distributional planted
$k$-biclique problem is to distinguish between samples drawn from $D_S$ and samples drawn uniformly from $\{0,1\}^n$.
\end{problem}

We show that any $p$-pass streaming algorithm solving the distributional planted $k$-biclique problem with $m$ samples requires  
\begin{equation}\label{eq:distributional}
\Omega\!\left(\frac{n}{m} \cdot \frac{n^2}{p\,k^4\log n}\right)
\end{equation}
bits of memory.  
When $m = \Omega(n^3/k^4)$ -- that is, when $\sqrt{nm} \gg k^2 m / n$ -- a simple edge-counting algorithm using one pass and $O(\log n)$ memory suffices to distinguish the planted distribution from uniform. Hence, our memory-sample tradeoff is tight up to logarithmic factors in the low-memory regime. In the statistically feasible regime -- when $k = \Theta(\log n)$ and $m = \tilde O(n)$ -- any constant-pass streaming algorithm must use $\tilde\Omega(n^2)$ bits of memory.  
Without delving into tedious details, we show the same memory hardness for any multi-pass streaming algorithm that distinguishes between a random $G(m,n,1/2)$ bipartite graph and one with an added planted $(k\times k)$ biclique.  
While this result is significant in its own right and requires new techniques, our main contribution is a \emph{general framework} providing sufficient conditions on the underlying distributions to yield such memory-sample tradeoffs for detecting planted structures. This framework further allows us to generalize our lower bounds to detecting planted $(k\times k)$ bicliques in random $G(m,n,q)$ bipartite graphs for any $0<q<1/2$, which we discuss in more detail in \Cref{introsec:plantedclique}.

\subsection{Our general framework}

Changing notation slightly, consider the planted biclique problem on a bipartite graph with $n$ left vertices and $d$ right vertices. The streaming algorithm observes $n$ adjacency-list vectors in $\{0,1\}^d$, such that at $k$ uniformly chosen time-steps, these vectors contain all $1$s on a fixed subset of coordinates $S \subseteq [d]$. In our general framework for detecting planted structures, we retain the property that a fraction of the rows follow a planted distribution, but we additionally constrain the subset $S$ to lie within a predefined ``partition". This modification allows us to model a broader class of planted distributions, and we formalize this general setup below (see \Cref{fig:planted-distributions} for an illustration). Given a vector $x$, we represent its projection to coordinates in $S$ by $x_S$.

\begin{problem}[General planted structure detection] \label{prob:informal_setup}
Consider some $n,d$, $0<k\le n$ and $0<t\le d$. Let $\{T_{r}\}_{r\in[d/t]}$ be some partition of $[d]$, where $\forall r, |T_r|=t$.
Let $\mu_0,\{\mu_\theta\}$ be distributions on $t$-dimensional vectors, and $P$ be some distribution over $\theta$. The goal is to distinguish between the following joint distributions on $n$ such $d$-dimensional vectors $x^1,\ldots, x^n\in \calX^d$:
\begin{enumerate}
\item $\no$: $\forall i\in[n]$ and $\forall r \in [d/t]$, $x^i_{T_r}$ is drawn from $\mu_0$.

\item $\yes$: Pick $r$ uniformly from $[d/t]$.  Pick set $R$ uniformly from subsets of $[n]$ of size $k$. Pick $\theta\sim P$.%

$$\forall i\in[n] \text{ and }\forall r'\neq r, x^i_{T_{r'}} \sim \mu_0$$ (i.e. except for the chosen partition $T_r$, draw coordinates in all partitions from $\mu_0$, for all datapoints).%

$$\forall i\not \in R, x^i_{T_r} \sim \mu_0$$ (i.e. for datapoints not in chosen set $R$, coordinates in all partitions are drawn similar to $\no$).

$$\forall i\in R, x^i_{T_r} \sim \mu_\theta$$ (i.e. for datapoints in chosen set $R$, the coordinates in chosen partition $T_r$ are drawn from $\mu_{\theta}$).

\end{enumerate}
\end{problem}

\begin{figure}[h]
    \centering
        \includegraphics[scale=0.30]{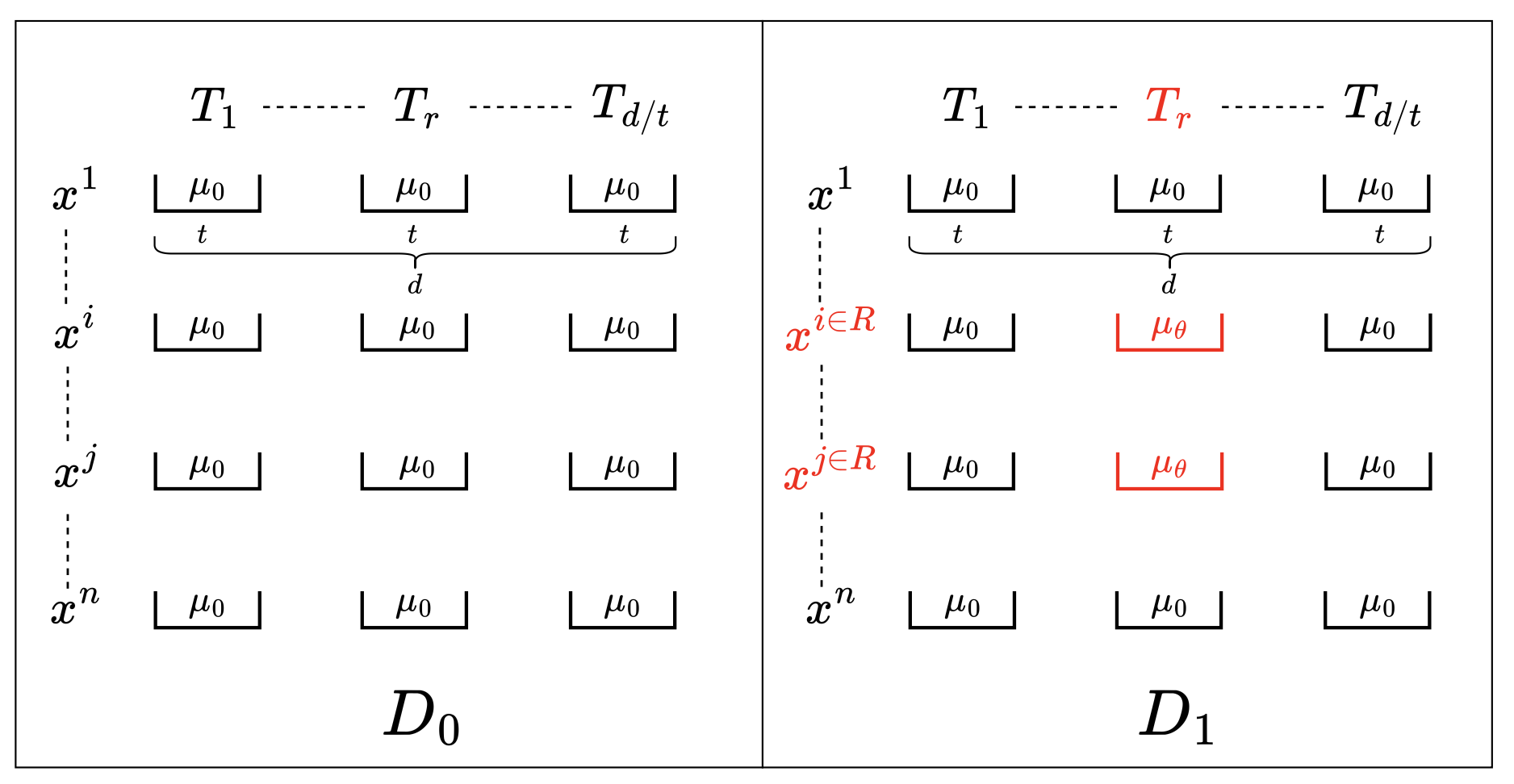}
        \caption{The distributions $\no$ and $\yes$ from \Cref{prob:informal_setup}. The partition $\{T_r\}_{r \in [d/t]}$ is shown to be over contiguous segments here only for convenience. In $\yes$, $r$ is drawn uniformly from $[d/t]$, $R$ is drawn uniformly from subsets of $[n]$ of size $k$, and $\theta$ is drawn from $P$. The planted structure is highlighted in red.}
        \label{fig:planted-distributions}
\end{figure}

The above setup captures settings with sparsely planted structures on certain coordinates of the datapoints (through $T_r$), as well as scenarios where a subset of datapoints are outliers containing planted structure (through $R$). In addition to encompassing the planted biclique detection problem in random bipartite graphs $G(n,d,q)$ -- where each edge across the partition appears independently with probability $q$ -- this framework also models canonical learning problems over Gaussian distributions. Let $\mu_0$ be a product distribution over $t$-dimensional vectors with $i.i.d.$ $\mathcal{N}(0,1)$ entries, let $P$ be the uniform distribution over $k$-sparse subsets $\theta \subseteq [t]$, and let $\mu_\theta$ be the distribution obtained from $\mu_0$ by shifting the mean to $+1$ on coordinates in $\theta$. This yields the problem of detecting a mixture of a standard multivariate Gaussian and a sparse-mean Gaussian. While both the planted biclique and sparse-mean Gaussian problems involve planted distributions that are $i.i.d.$ over the selected coordinates, we use our framework to also model sparse PCA, where the planted distribution $\mu_\theta$ introduces correlations among the selected coordinates -- specifically, the projection of a datapoint onto $\theta$ is drawn from a Gaussian with a shifted covariance. \Cref{table1} summarizes the specific parameters used for these applications.

\begin{table}[h]
\centering
\caption{Different instantiations for \Cref{prob:informal_setup}. We consider $\theta$ to be a uniformly random $\ell$-sized subset of the chosen $t$-sized partition $T_r$.}
\label{table1}

\begin{tabular}{|p{4.2cm}|p{3.2cm}|p{3.7cm}|p{3.5cm}|}
\hline
\textbf{Application} & \textbf{Null Distribution $\mu_0$} & \textbf{Planted Distribution $\mu_\theta$} & \textbf{Parameters $k$ and $t$, where $1\le k\le n$ and $1\le t\le d$} \\
\hline
Planted biclique on random $G(n,d,q)$ graphs & $\Ber(q)^{\otimes t}$ & Set coordinates in $\theta$ to be $1$ & $k=\ell$, $t=\tilde\Theta\left(\frac{\ell^2}{q}\right)$\\
\hline
Planted biclique with monotone adversaries & $\Ber(1/2)^{\otimes t}$ & Set coordinates in $\theta$ to a fixed string in $\{0,1\}^\ell$ & $k=\ell$, $t=\ell$ \\
\hline
$(1-q):q$ mixture of a standard Gaussian and a sparse-mean Gaussian & $\mathcal{N}(0,1)^{\otimes t}$ & Coordinates in $\theta$ are drawn from $\mathcal{N}(1,1)^{\otimes \ell}$ & $k=q n$, $t=\tilde\Theta(d^{o(1)}\ell^2)$\\
\hline
$\ell$-sparse PCA & $\mathcal{N}(0,1)^{\otimes t}$ & Coordinates in $\theta$ are drawn from $\mathcal{N}(0,I_\ell+\alpha  vv^T)$, for some small $\alpha$ and unit vector $v$ & $k=n$, $t=\tilde\Theta(d^{0.01}\ell)$ \\
\hline

\end{tabular}
\end{table}

Next we state our main theorem establishing memory-sample tradeoffs for the general planted structure detection problem. 

\begin{theorem}[Informal version of \Cref{thm:micgeneral}]\label{thm:planted_informal}
Let $\mu_1=\bbE_{\theta\sim P}\left[\mu_\theta\right]$. Suppose
     $\mu_1(x) \le c\cdot \mu_0(x) \; \forall \; x\in \calX^t$.  Then, any $p$-pass streaming algorithm that solves \Cref{prob:informal_setup} requires at least $\Omega\left(\frac{nd}{p\cdot c\cdot  k^2 \cdot t}\right)$ bits of memory. 
\end{theorem}

The above theorem provides a sufficient condition on the null and planted distributions, $\mu_0$ and $\mu_\theta$ respectively, for proving memory-hardness of detecting planted structures. Note that we must at least require the \emph{distance} between $\mu_0$ and the average planted distribution $\mathbb{E}_{\theta \sim P}[\mu_\theta]$ to be small; otherwise, a single sample would suffice to distinguish the two distributions. Our condition is both simple to state -- as it depends only on the average planted distribution -- and broadly applicable. However, ensuring it holds for the distributions used in our applications (listed in \Cref{table1}) requires careful truncation and modification.
To prove the above theorem,  we first establish a new, generalized distributed data processing inequality that is of independent interest (see \Cref{thm:informalgdpi} for a detailed statement). \Cref{thm:planted_informal} then follows through multiple applications of direct-sum-type arguments -- one over the partitions and another over the rows. Each step is nontrivial, as it crucially depends on the specific distributions used in the information complexity notions. We provide a detailed outline of our technical contributions in \Cref{sec:tech}.

\subsection{Applications to planted biclique detection and its variants}\label{introsec:plantedclique}

Firstly, we consider memory requirements for detecting planted bicliques. Let $G(m,n,q)$ denote the distribution over bipartite graphs with $m$ left vertices and $n$ right vertices, where each edge across the partition is present independently with probability $q$. In the generalized planted biclique detection problem, the goal is to distinguish between a random bipartite graph drawn from $G(m,n,q)$ and one containing a planted biclique of size $(k \times k)$ on a uniformly chosen set of vertices. By instantiating \Cref{prob:informal_setup} with the null ($\mu_0$) and planted distributions ($\mu_\theta$) as in \Cref{table1}, and setting $t = \tilde{\Theta}(k^2/q)$, we obtain the following multi-pass streaming lower bound.

\begin{theorem}[Informal version of \Cref{thm:main}]\label{thm:clique_informal}
    For the planted $k$-biclique problem in random bipartite graphs $G(m,n,q)$, any $p$-pass streaming algorithm that observes the adjacency lists of left vertices in random order and achieves a constant distinguishing advantage requires at least $\tilde\Omega\left(\frac{nmq}{p\cdot  k^4 }\right)$ bits of memory. 
\end{theorem}
To apply \Cref{thm:planted_informal}, we require that the likelihood ratio $\mu_1(x)/\mu_0(x)$ be bounded for all $x$. This fails when $t \ll k^2/q$, as planting $k$ ones substantially changes the number of ones under $\mu_0$. Using careful truncation arguments (briefly outlined in \Cref{subsec:graph}), we show the condition holds for $t\gg k^2/q$. While $q = 1/2$ is the most common setting, it is noteworthy that our framework applies to general $q$. In particular, the case $q = \frac{\polylog(n)}{n}$ is crucial for our new memory lower bound on approximating the density of subgraphs, discussed in the next subsection.

Fix $m=n$ and $q=1/2$. In the regime where the algorithm has $\poly\log n$ space (a common notion of space-efficient computation, particularly with regards to planted clique \cite{mardia2021space}), the result says that it is not possible to detect cliques of size $O(n^{1/2-\delta})$, unless the algorithm makes $n^{\Omega(\delta)}$ passes over the data. Since we are usually interested in algorithms with constant or logarithmic number of passes over the data in streaming settings, the bound says that the problem cannot be solved with $n^{o(\delta)}$ space in those settings. The result is tight in the sense that for clique size $k=\Omega(\sqrt{n \log n})$, simply counting edges (which uses $O(\log n)$ space and one pass) suffices \cite{kuvcera1995expected}. 

We next relate our result to prior work on the hardness of planted clique detection in the streaming model. For worst-case graphs, \cite{halldorsson2012streaming} and \cite{braverman2018new} prove a $\tilde{\Omega}(n^2 / r^2)$ memory lower bound for one-pass algorithms that compute an $r$-approximation of the maximum clique size, and \cite{braverman2018new} also provide a matching upper bound (which extends to bicliques). Since the largest (bi)clique in $G(n, n,1/2)$ has size $\Theta(\log n)$ with high probability, a $\tilde{\Theta}(k)$-approximation suffices to detect a planted $k$-clique, yielding a $\tilde{O}(n^2 / k^2)$-space one-pass algorithm for the planted biclique problem -- leaving room to tighten\footnote{Note that since our general framework yields memory–sample tradeoffs, \Cref{thm:clique_informal} is tight for the distributional version of the planted biclique problem (\Cref{prob:dist}).} our lower bound by a factor of $\tilde{O}(k^2)$ when $k < \sqrt{n}$. For the planted clique problem, \cite{rashtchian2021average} also establish a $\Omega(n^2 / (p k^4))$ memory lower bound for $p$-pass algorithms, but in a stronger model where edges arrive in an adversarial order. In contrast, our model is arguably more natural and easier, as it reveals all neighbors of each vertex together while vertices arrive in random order. The only prior work establishing non-trivial lower bounds in a related communication model -- where each player receives the adjacency list of a vertex -- is \cite{chen2019broadcast}, which applies only to cliques of size at most $n^{1/4}$.

\paragraph{Planted biclique under monotone adversaries}
Starting with the work of \cite{feige1998heuristics}, the \emph{monotone adversary model} studies the extent to which algorithms for planted clique depend on the specific distributional assumptions of the problem. The monotone adversary model corresponds to starting with the standard input for planted clique, after which an adversary is allowed to remove any edges not belonging to the planted clique (if the graph has a planted clique). Since the adversary only removes such edges, it is in some sense helpful. \cite{feige2000finding} showed that while simpler algorithms based on edge counting and the spectral method fail at the $k=\tilde{O}(\sqrt{n})$ threshold, a semi-definite programming based method still recovers cliques at the previous $k=\tilde{O}(\sqrt{n})$ threshold in the presence of such adversaries. 
Our framework can capture monotone adversaries (see \Cref{table1} for the parameters), and we get the following result against streaming algorithms that detect whether there is a clique of size greater than $k$.
\begin{theorem}[Informal version of \Cref{thm:semi-random}]\label{thm:monotone_informal}
    For the planted biclique problem in the presence of a monotone adversary, any successful $p$-pass streaming algorithm that detects the presence of planted cliques of size at least $k$, requires $\tilde{\Omega}\left(\frac{n^2}{p\cdot  k^3}\right)$ bits of memory. 
\end{theorem}

The result shows that the threshold for solving the problem in constant passes with $O(\log n)$ memory moves from $k=\tilde{\Omega}(\sqrt{n})$ to $k=\tilde{\Omega}(n^{2/3})$ --- showing that somewhat strong distributional assumptions are needed to solve the problem at the $k=\tilde{\Omega}(\sqrt{n})$ threshold with small memory. Note that the previous algorithm based on counting the number of edges no longer works in this model, though the  $\tilde{O}\left(\frac{n^2}{k^2}\right)$ memory one-pass algorithm  from \cite{braverman2018new}  does work. It is possible that no $O(\log n)$-space, constant-pass algorithm can solve the problem in the presence of a monotone adversary for $k = \tilde{\Omega}(n)$, suggesting that even a monotone adversary may make the planted biclique problem as hard for streaming algorithms as in the worst-case setting.

\subsection{Application to graph streaming under the vertex arrival model}

In this section, we focus on general undirected graphs, not necessarily bipartite, and study the memory requirements for approximating certain graph properties in the \emph{vertex arrival} streaming model. In this model, vertices arrive in a worst-case order, and each new vertex reveals its connectivity to all previously arrived vertices. This model is natural for graph streaming problems and has been fairly studied; the seminal work of \cite{karp1990optimal} on online bipartite matching in the vertex arrival model sparked extensive research in this area. More recently, \cite{kapralov2021space} established a separation between the edge and vertex arrival models for the online bipartite matching problem.

For other graph properties, while interesting upper bounds are known (e.g., \cite{kallaugher2019complexity} for triangle counting), lower bounds in the vertex-arrival streaming model are hard to come by. Among such problems, approximating the maximum clique or independent set size has been the most studied \cite{braverman2018new,cormode2018approximating,cormode2019independent}. While \cite{halldorsson2012streaming} established a tight $\Omega(n^2 / \alpha^2)$ memory lower bound for one-pass algorithms, under the edge-arrival model, that compute an $\alpha$-approximation to the maximum clique size, it is conceivable that the same approximation might be achievable using lesser memory in the vertex-arrival model. In fact, \cite{cormode2019independent} showed that computing a \emph{maximal} independent set is trivial in vertex-arrival streams but requires $\Omega(n^2)$ space in edge-arrival streams. In terms of lower bounds,  previous works \cite{cormode2019independent,sundaresan2025optimal} proved that any $\alpha$-approximation of the maximum clique size in one-pass vertex-arrival streams requires $\Omega(n^2 / \alpha^5)$ space, while \cite{braverman2018new} established an incomparable $\Omega(n / \alpha^2)$ lower bound for one-pass adjacency-list streams.
\Cref{thm:monotone_informal} implies the following stronger memory lower bound for \emph{multi-pass} streaming algorithms that approximate the size of the largest \emph{biclique} in an undirected graph.

\begin{theorem}[Informal version of \Cref{thm:biclique-approx}]
    \label{cor:max-biclique-informal}
    Any $p$-pass streaming algorithm in the vertex-arrival model that approximates the maximum biclique size within a factor of $\alpha\ge 1$, must use $\tilde{\Omega}\left(\frac{n^2}{p \alpha^3}\right)$ memory.
\end{theorem}

\paragraph{Approximating density for \texorpdfstring{$\beta$}{B}-bounded subgraphs}
Next, we turn to the \emph{densest subgraph} problem, a fundamental primitive in graph mining that has been extensively studied since the 1970s (see the excellent survey by \cite{lanciano2024survey} on the problem and its variants). Broadly, the goal is to find a subset of vertices $S$ maximizing the ratio of the number of edges within $S$ to $|S|$, referred to as the \emph{density} of $S$.
\cite{bahmani2012densest} initiated the study of streaming algorithms for this problem, presenting an $O(\log n)$-pass algorithm that achieves a constant-factor approximation using $O(n)$ bits of memory. They also proved that any $p$-pass streaming algorithm that $\alpha$-approximates the maximum density requires $\Omega(n / (p \alpha^2))$ bits of memory under worst-case edge arrival streams.
We establish the following stronger memory lower bound for the harder problem of approximating the maximum density of subgraphs of size at most $\beta$, in the (potentially stronger) vertex-arrival streaming model.

\begin{theorem}[Informal version of \Cref{thm:densest-at-most-d-subgraph-approx}] \label{cor:densest-informal}
For any $\alpha\ge 1$,
    any $p$-pass streaming algorithm in the vertex arrival model, which $\alpha$-approximates the maximum density among all subgraphs of size at most $\beta$ for $\beta=o(n/\alpha^2)$, requires at least $\tilde \Omega\left(\frac{ n^2}{p\cdot \alpha^4 \beta}\right)$ bits of memory. 
\end{theorem}
The multi-pass streaming lower bound of \cite{bahmani2012densest} is based on a reduction from set disjointness, which critically relies on the worst-case edge-arrival order and does not extend to vertex-arrival streams.
To prove the above theorem, we establish a reduction from the planted biclique detection problem (\Cref{thm:clique_informal}) with parameters $k = \alpha \log n$ and $q = \log n / \beta$. Since detecting planted bicliques is believed to be computationally hard even in ``sparse" graphs, we cannot hope to extend our hardness result beyond an approximation factor of $\alpha = n/\beta$; since simple greedy algorithms \cite{asahiro2000greedily} are known to give such approximation to maximum density of $\beta$-bounded subgraphs.

\subsection{Applications to canonical learning problems over Gaussians}\label{sec:gaussian_intro}

We now discuss our results for learning problems over Gaussians. Gaussian distributions pose significantly more challenges in bounding the likelihood ratio $\mu_1/\mu_0$ between the planted and null distributions. In Section~\ref{sec:learning_overview}, we discuss how to suitably truncate the distributions to apply our framework in more detail.

\paragraph{Detecting sparse mean Gaussians}

We now consider the problem of testing whether the data --- or some of the data --- is coming from a Gaussian with a sparse mean. This is a fundamental problem  with a long line of work \cite{ingster1996some,baraud2002non,IngsterSuslina2003,donoho2004higher,jager2007goodness,collier2017minimax,collier2018optimal}. It models various applications where the goal is to do hypothesis testing to determine if there is some sparse signal present in the data. In many applications such as anomaly detection the signal is also `weak' and not all datapoints come from the planted distribution (see \cite{donoho2004higher} and the survey \cite{donoho2015higher}), and there has been significant work on detecting such signals which are both sparse and weak \cite{donoho2004higher,donoho2008higher,hall2010innovated,klaus2013signal}. This aspect can also be captured by our general setting in Problem \ref{prob:informal_setup} (through choice of the set `$R$'). 

We now describe the sparse Gaussian testing setting in more detail. We first draw the planted mean vector $\mathbf{\theta}\in \{0,\alpha\}^d$ uniformly at random from the set $\{0,\alpha\}^d$, but subject to it being $\sparsity$ sparse. Here $\alpha\in [0,1]$ is the signal strength parameter. Let $q>0$ be the probability of getting a planted sample. In the null distribution, we always get samples from $N(0,I)$. In the planted distribution,  at every time step with probability $q$ we get a sample from $N(\mathbf{\theta},I)$, and  with probability $(1-q)$ we get a sample from $N(0,I)$. Using our general framework, we show the following memory-sample tradeoff for algorithms which take as input a lower bound on the sparsity in the planted case, and then work for all sparsity levels above this lower bound.

\begin{theorem}[Informal version of \Cref{thm:gaussian_general}]\label{thm:gaussian_informal}
    For the problem of detecting sparse mean Gaussians where the mean vector has sparsity at least $\ell$, any successful $p$-pass, $s$-bit memory algorithm which uses $n$ samples requires $s \cdot n \ge \Omega\left(\frac{d^{0.99}}{p\cdot  (\alpha \sparsity q)^2 }\right)$.
\end{theorem}  

{We note that our general framework is versatile enough to capture dependence of the tradeoff on the signal strength $\alpha$ here, and the bound also holds for constant values of $\alpha$ where the planted vector has a super-constant norm.} Several other remarks about this lower bound are in order, starting with upper bounds for this problem. By storing the sum of all the co-ordinates of all the vectors, it is possible to distinguish the two distributions with $n=O\left(\frac{d}{(\alpha \sparsity q)^2}\right)$ samples (since, roughly, the means in the planted versus null case differ by $\alpha \sparsity q n$, and the variance is $O(nd)$). Therefore the problem can be solved with a one-pass $O(\log d)$ memory  algorithm, but using $O\left(\frac{d}{(\alpha \sparsity q)^2}\right)$ samples. Our bound shows that this sample-complexity is near-optimal and necessary for  $O(\log d)$ memory constant pass algorithms. This required sample complexity for $O(\log n)$ memory algorithms is significantly worse than the optimal sample complexity without memory constraints. We can solve the problem by storing $\tilde{O}(d/\sparsity)$ randomly chosen co-ordinates of  $\tilde{O}\left(\frac{1}{q\alpha^2}\right)$ datapoints, and by checking the empirical averages of the co-ordinates for every $\tilde{O}\left(\frac{1}{\alpha^2}\right)$  sized subset of the stored points. This requires $\tilde{O}\left(\frac{d}{\sparsity\alpha^2q}\right)$ memory and $\tilde{O}\left(\frac{1}{\alpha^2 q}\right)$ samples. $\tilde{O}\left(\frac{1}{\alpha^2 q}\right)$ is the information-theoretic sample complexity of the problem, and hence our memory lower bound to achieve optimal sample complexity is optimal up to a factor of $\frac{d^{0.01}\sparsity}{\alpha^2}$. We also note that the lower bound shows that memory-limited algorithms need a sample complexity which depends on $1/q^2$, whereas information-theoretically only a $1/q$ dependence is needed --- hence memory-limited algorithms could need many more samples to detect outliers or find weak signals in the data distribution. This is similar to gaps observed for the \emph{needle problem} \cite{andoni2008better,chakrabarti2008robust,lovett2023streaming}, where the goal is to detect if a data stream has one element which appears with a higher than uniform probability.

Even for the case of $q=1$ where all samples are drawn from a Gaussian with a sparse mean, we are unaware of previous memory lower bounds for the detection problem, though memory lower bounds are known for the estimation version of the sparse Gaussian mean problem \cite{zhang2013information,garg2014communication,braverman2016communication}.

\paragraph{Sparse PCA detection problem}

Sparse PCA adds a sparsity constraint to the  PCA problem and has found widespread applications in statistics, ML and data analysis \cite{zou2006sparse,zou2018selective}. As discussed earlier, it is also a prototypical problem for studying  understanding statistical-computational tradeoffs.
From the perspective of memory constraints, streaming algorithms have been developed for sparse PCA  \cite{mairal2010online,yang2015streaming,wang2016online,kumar2024oja} --- building on developments in streaming PCA \cite{mitliagkas2013memory,jain2016streaming}. These algorithms all need at least $\Omega(d)$ memory to find the sparse principal component, but the trivial information-theoretic lower bound only says that $\tilde{\Omega}(k)$ memory is needed for the estimation problem if the principal component  is $k$-sparse. We are unaware of any previous non-trivial memory lower bound for the problem. 

We describe the detection  version of the sparse PCA problem. We first draw the sparse principal component $\mathbf{\theta}\in \{0,1/\sqrt{\ell}\}^d$ uniformly at random from the set $\{0,1/\sqrt{\ell}\}^d$, but subject to it being $\sparsity$-sparse. The goal is to distinguish whether the samples are coming from $N(0,I)$ or from $N(0,\Sigma)$, where $\Sigma=I+\alpha \theta\theta^\top$. This is the widely studied spiked covariance model, also known as the  spiked Wishart model \cite{zou2006sparse,johnstone2009sparse}. Here $\alpha>0$ is the signal strength parameter, and we consider $\alpha$ which is a small enough constant. 
Note that in contrast to previous settings, here all the samples have the sparse, planted structure, as is standard in sparse PCA. We show the following tradeoff for this problem.

\begin{theorem}[Informal version of \Cref{thm:pca_general}]\label{thm:pca_informal}
    For the sparse PCA detection problem, any successful $p$-pass, $s$-bit memory algorithm which uses $n$ samples requires $s \cdot n \ge \Omega\left(\frac{d^{0.99}}{p\cdot  \ell}\right)$.
\end{theorem}  

For the small-memory regime where $s=O(\log n)$ and $p=1$, our result shows that $\Omega(d^{0.99}/\ell)$ samples are necessary. In contrast, note that the problem is information-theoretically solvable with only $\tilde{O}(\ell)$ samples \cite{ma2015sum}. Therefore, small, $O(\log n)$-memory algorithms need significantly more samples than the information-theoretic limit to solve the problem. In the $O(\log n)$-memory regime, it is possible to solve the problem with $\tilde{O}(d)$ samples by thresholding the sum of the squares of all the co-ordinates over all the samples. Therefore, there is a gap of $\ell$ (and some other less significant terms) between our lower bound and the best-known upper bound. However, we show our lower bound for a more structured version of the problem where a consecutive set of co-ordinates of $\theta$  are non-zero (in the technical overview in Section \ref{sec:learning_overview}, we discuss this further). In the presence of this structure, it is possible to solve the problem with $\tilde{\Omega}(d/\ell)$ samples, by thresholding the squares of the {sum of} consecutive co-ordinates instead. Therefore, our bound is nearly tight for this setting that we consider.

To the best of our knowledge, our result represents the first memory-sample tradeoffs for sparse PCA in the standard spiked covariance model, either for the detection or the estimation version of the problem. Note that a reduction is known from the planted clique problem to the sparse PCA \cite{berthet2013complexity,mardia2021logspace}, however this reduction does not work in the streaming model. The closest related setting for which memory-sample tradeoffs are known is for detecting if a pair of co-ordinates in samples drawn from an unknown distribution are correlated \cite{shamir2014fundamental,dagan2018detecting}. This is similar to the sparse PCA problem when $\ell=2$. However, the bound of \cite{shamir2014fundamental,dagan2018detecting} only holds when the correlation (which is analogous to our signal strength parameter $\alpha$) is polynomially small in $d$ (in which regime they prove a stronger bound than Thm \ref{thm:pca_informal}), in contrast our bound holds for constant values of $\alpha$, and importantly, generalizes beyond the case of correlations where $\ell=2$.\footnote{Note that in the sparse PCA problem the co-ordinates of the samples are not independent, and hence we cannot do a direct reduction from $\ell=2$ to larger values of $\ell$.}

\subsection{Other Related Work}
We now discuss some other relevant literature. There has been significant work on understanding learning under information constraints such as limited memory or communication constraints \cite{balcan2012distributed,duchi2013local, shamir2014fundamental, arjevani2015communication, steinhardt2015minimax, raz2018fast, dagan2019space, woodworth2021min}, including implications for privacy and memorization \cite{feldman2025trade}. 
The works closely related to our work are \cite{brown2022strong, jiapeng2025colt, braverman2024new}, where the former two also study information cost for a similar setup to planted biclique (Task B in their paper) en route to showing memory lower bounds for certain \emph{estimation} problems. However they measure information cost with respect to a non-uniform distribution which prevents our subsequent direct sum  application, and hence their analysis is not helpful for us. Particularly, while Task B in these papers plants 0/1 uniformly on a $k$-sized subset, at best their proof can be massaged  to  show that detecting cliques of size $k$ requires $\Omega(n/k^3)$ memory, whereas we show a $\Omega(n^2/k^4)$ bound which requires significantly new techniques and measuring information $w.r.t.$ the null distribution. Secondly, hybrid arguments used by these papers fail to work for our other applications such as sparse PCA where planted coordinates are correlated.

Work on the needle problem \cite{andoni2008better,chakrabarti2008robust,lovett2023streaming,braverman2024new} also shares elements of our analysis. However, bounds for the needle problem do not yield bounds for our general planted structure detection problem since the needle problem is limited in the sense that a needle is chosen uniformly from the null distribution.  Another relevant set of papers is the work of \cite{feldman2017statistical} on variants of statistical query (SQ) dimension for Problem \ref{prob:dist} and the work of \cite{garg2018extractor} which shows memory-sample tradeoffs parameterized by the SQ dimension of the problem using extractor-based bounds. However, this connection is weak to give anything non-trivial.

In the graph streaming literature, there is an extensive literature on both upper bounds and lower bounds  (see survey \cite{mcgregor2014graph} on upper bounds and \cite{assadi2023recent} on lower bounds for an overview). Typical lower bounds here are for worst-case edge arrival graphs, whereas our results hold in the vertex arrival model as well as random order.
Finally, we note that there is a large body of work on the problem of finding outliers in streaming data, such as \cite{tan2011fast,manzoor2016fast,ahmad2017unsupervised}, see \cite{lu2023review} for a detailed survey. There is also work on memory lower bounds for streaming outlier detection \cite{sharan2018efficient}, but for worst-case data.

%% file: proofoverview.tex
\section{Technical Overview}\label{sec:tech}

In this section, we present an overview of our proofs, beginning with the general planted structure detection problem in \Cref{prob:informal_setup}. Recall that in this problem, only one set in the partition contains a planted structure, and only a subset of rows are planted (see also \Cref{fig:planted-distributions}).

We start with a simpler setting involving only a single set in the partition, and where all rows are planted. For this case, we prove a new data processing inequality to establish information cost lower bounds with respect to the \emph{null} distribution (\Cref{subsec:dpi}). We then extend our analysis to the case where only one set in the partition is planted -- matching the structure of \Cref{prob:informal_setup} -- but all rows remain planted. Using direct-sum-type arguments, we establish an information complexity lower bound for this case (again with respect to the null distribution), building on the bounds provided by the data processing inequality (\Cref{subsec:directsum}).
Finally, we address the full problem where only a subset of rows are planted. Here, we apply the recent multi-pass information cost framework of \cite{braverman2024new} to derive memory lower bounds (\Cref{subsec:mic}). This step introduces an additional $\sim n$ factor in the bound and  critically relies on the fact that the previous bounds measure information under the null distribution.
We conclude by describing applications of our general framework to graph problems and statistical detection tasks in \Cref{subsec:graph} and \Cref{sec:learning_overview}, respectively.

As an instrumental warm-up exercise, we first consider the case of single partition; under the \textbf{no} case, at each time-step, we get a sample  drawn from some distribution $\mu_0$ on $t$-dimensional vectors. Whereas in the $\textbf{yes}$ case, a sample is drawn from the \textit{planted} distribution -- $\mu_\theta$ -- with probability $\gamma$ (think of $\gamma$ as $k/n$). As we want to develop a general framework for studying hardness of detecting planted structures, we want to make as few assumptions on $\mu_0$ and $\mu_\theta$, where the parameter $\theta$ takes value in some set $\Omega$.

\cite{braverman2016communication} studies a similar distribution detection question (albeit) under the communication complexity model, where every player independently gets a sample either from $\mu_0$ or from $\mu_1$. \cite[Theorem 1.1]{braverman2016communication} establishes information complexity (IC) lower bounds ($w.r.t.$ $\mu_0$) when $\mu_1$ is point-wise bounded by $O(\mu_0)$. Even though such a restriction might seem stringent, distributions under many natural detection problems such as Gaussian mean estimation can be truncated to satisfy it. However, for the distributions we consider -- for example, in the planted clique problem, $\mu_0$ is the uniform distribution over $t$-dimensional vectors, whereas $\mu_\theta$ has a planted $1$s on a set indexed by $\theta$ --- $\mu_\theta(x)/\mu_0(x)$ can be exponentially large for ``typical" $x$s. In fact, we cannot hope to prove memory lower bounds when we know $\theta$, as even a single time-step can detect the outlier without any prior knowledge. Hence, under the random process $\theta\sim P$, we at least want that $||\bbE_{\theta\sim P}\mu_\theta - \mu_0||_{TV} =o(1)$. One of our main technical contributions is the following generalized distributed data processing inequality, when the expected distribution $\bbE_{\theta\sim P}\mu_\theta$ is point-wise bounded by $O(\mu_0)$.

\begin{theorem}[Informal version of \Cref{thm:gdpi} + \Cref{lem:boundhellinger}]\label{thm:informalgdpi}
   Let $\mu_0, \{\mu_\theta\}_{\theta\in \Omega}$ be a family of distributions over some sample space $\mathcal{X}$ such that $\bbE_{\theta\sim P} \mu_\theta\le O\left(\mu_0\right)$. Consider the distributed detection setting, where if $V=0$, then each party receives $x_i \sim \mu_0$, and if $V=1$ then we first draw $\theta\sim P$, and then each party receives $x_i \sim \mu_\theta$. Then, for any multi-party communication protocol $\Pi$ that learns $V$ with large enough constant probability,
    \begin{align*}
         I(X; \Pi(X))=\Omega(1), \text{ when $\forall i, X_i\sim \mu_0$.}
    \end{align*}
\end{theorem}

One can view the above theorem as a generalization of \cite[Theorem 3.1]{braverman2016communication} to non-product distributions and might be of independent interest.\footnote{While \cite[Theorem 3.1]{braverman2016communication} has been widely used to analyze distributed computing under information constraints, \cite{duchi2019lower} noted that the independence condition in the theorem is often too strong.}  Before diving into the proof overview, let's talk about a bit about the implications of this theorem.  If $V$ is a uniform bit, then it is trivial to show that any communication protocol that detects $V$ requires $\Omega(1)$ information $w.r.t.$ the mixture distribution. However, our key objective is to establish an information complexity lower bound with respect to $\mu_0$, that is, the \textbf{no} distribution. This is crucial for two reasons. Firstly, using a direct sum argument over $d/t$ partitions, we can prove an $\Omega\left(d/t\right)$ multi-party communication lower bound for detecting planted structures when every player gets a sample from the planted distribution (that is, $\gamma=1$). Secondly, leveraging the recently introduced information cost notions of multi-pass streaming algorithms \cite{braverman2024new}, we are able to lift IC bounds for communication protocols to memory lower bounds that grow quadratically with $1/\gamma$. Again, this is possible because we always use information cost measures $w.r.t.$ \textbf{no} distribution, which is a product distribution. We discuss these in more detail in \Cref{subsec:directsum} and \Cref{subsec:mic} respectively.

\subsection{Generalized distributed data processing inequality}\label{subsec:dpi}
We follow and build upon the proof of \cite[Theorem 3.1]{braverman2016communication}. Let $\mu_1=\bbE_{\theta\sim P}\mu_\theta$. Then, the theorem from \cite{braverman2016communication} would show an $\Omega(1)$ bound on IC $w.r.t.$ $\mu_0$ for communication protocols that distinguish between cases where each player either receives \emph{independent} samples from $\mu_0$ or from $\mu_1$. In fact, their results potentially offer a stronger bound involving the strong data processing inequality (SDPI) constant\footnote{We refer the reader to \cite{raginsky2016strong} for a survey on SDPI. In this paper, we will only be using data processing inequalities.} of $\mu_0$ and $\mu_1$. 
In \Cref{thm:informalgdpi}, we are aiming for an $\Omega(1)$ IC bound $w.r.t.$ $\mu_0$ but when under the $V=1$ case, samples are not independent.

When $V=1$, first we draw $\theta$ from a prior distribution $P$, and then every player gets an independent sample drawn from $\mu_\theta$. Note that if the players knew $\theta$, then a \emph{single} player can solve the detection problem, without sending $\Omega(1)$ bits of information under $\mu_0$. For example, for the planted biclique problem on $t$-dimensional vectors, a single player can send the AND on $k$ bits of the plant, which reveals $O(1/2^k)$ bits of information about the input under uniform distribution. This is the first step in the proof -- to show that a single player cannot solve the distinguishing problem without revealing information under $\mu_0$ distribution. In fact, there is no distinction between the two detection problems (the one we study from the one where each player gets a sample from $\mu_1$ under the yes case), when only \textit{one} player receives a sample from the planted distribution ($\mu_\theta,\theta\sim P $) while all others receive samples from the null distribution $\mu_0$. For the \textit{single player} setting, we can adopt the previous proof to show that to detect $V$, the player needs to send $\Omega(1)$ bits of information under $\mu_0$. 

For communication protocols under \textit{product} input distributions, one can use the cut-and-paste property \cite{bar2004information} to connect the distinguishing capacity of the single player setting to when each player gets a sample according to $\mu_0$ or $\mu_1$. However, this property does not hold for non-product input distributions. Our approach capitalizes on the fact that once we condition on $\theta$, the cut-and-paste property still holds. As we are using information-theoretic quantities that easily tensorize and work well with linearity of expectation, we are able to use both 1) that the cut-and-paste property holds for every $\theta$, to lift hardness of single player setting to the when all receive samples from $\mu_0$ or $\mu_\theta$, and 2) that in the single player setting, to solve the detection problem on average over $\theta\sim P$, one needs to reveal information $w.r.t.$ $\mu_0$.

\subsection{Direct sum over the partitions}\label{subsec:directsum}
In \Cref{prob:informal_setup}, under the planted distribution, we embed samples from $\mu_\theta$ on a random partition of $d$-dimensional vectors of size $t$, as well as on a random subset of ``rows". While \Cref{thm:informalgdpi} is a crucial part of the our multi-pass memory lower bounds, under the general framework, it doesn't say anything about memory needed to distinguish. To prove \Cref{thm:planted_informal}, we apply \textit{direct-sum} like  techniques twice; once on the partitions and another on the randomness of rows. In this subsection, we discuss the former. As \Cref{thm:informalgdpi} gives us an IC bound $w.r.t$ $\mu_0$, using a standard direct-sum argument, we get $\Omega(d/t)$ IC bound for any communication protocol that distinguishes the case when all player get samples from $\mu_0^{\otimes d/t}$, from the case when \textit{all} players get samples from $\mu_\theta$ embedded at a fixed partition (that is, $n=k$ case in \Cref{prob:informal_setup}). See \Cref{lem:ccdirectsum} for the formal statement of the result. 

Consider the special case of the planted clique distributed detection problem, where under the planted distribution, each player (total $k$ players) receives an $n$-dimensional vector with 1s planted on a pre-chosen set of coordinates $S$ of size $k$. For $n>> k^2$, we can adjust the distributions over the vectors to satisfy the requirements of \Cref{thm:informalgdpi}, and get an $\Omega(1)$ IC lower bound $w.r.t.$ uniform distribution (we discuss the truncation more in \Cref{subsec:graph}). In fact, using \Cref{lem:ccdirectsum}, we can get an $\tilde \Omega(n/k^2)$ IC lower bound\footnote{This is tight upto $\log n$ factors if we increase the number of players to  $\sim n/k^2$.}, $w.r.t.$ the uniform distribution, for any communication protocol that solves this planted clique distributed detection problem. In \Cref{subsec:mic}, we will crucially use the fact that this is an IC bound $w.r.t.$ uniform distribution ($\mu_0$ generally), which is stronger than a lower bound on the amount of communication needed to solve the detection problem.

The total communication lower bound of $\tilde \Omega(n/k^2)$ is interesting in its own right. Notably, it immediately implies a memory lower bound of $\tilde \Omega(n/p k^3)$ bits for any $p$-pass streaming algorithm solving the planted biclique version of \Cref{prob:informal_setup}. This is because, even if we fix the rows where the plants are made, distinguishing them from the uniform distribution would still require significant communication. In the next subsection, we leverage row-level randomness to establish a stronger memory lower bound of $\Omega(n^2/p k^4)$ bits, which is optimal (upto $\log n$ factors) for detecting planted bicliques of size $\poly\log n$, and provides non-trivial memory lower bounds for biclique sizes up to the critical threshold of $k=\sqrt{n}$. 

This communication complexity lower bound also implies that no $n^{o(1)}$ round protocol in the Broadcast Congested Clique model ($BCAST(\log n)$) can detect planted bicliques in directed graphs, even when the biclique size is as large as $n^{1/3-\epsilon}$ for some constant $\epsilon>0$. In comparison, the previous work of \cite{chen2019broadcast} established similar bounds for the planted clique problem in directed graphs for cliques of size at most $n^{1/4-\epsilon}$. Since our primary focus is on multi-pass streaming algorithms, we do not include further discussion on the implications for the $BCAST(\log n)$ model.

\subsection{Lifting to memory lower bounds}\label{subsec:mic}
Next, our goal is to lift the communication lower bounds established in \Cref{lem:ccdirectsum}, when each row gets a ``plant" in the planted distribution, to memory lower bounds in the streaming setting (beyond the implications discussed in the last subsection), when each row gets a ``plant" with probability $\gamma$. For sake of exposition, let us look at the planted bi-clique problem when $k=\poly\log n$. All the techniques discussed in this subsection readily generalize to \Cref{prob:informal_setup}, when IC notions are measured $w.r.t.$ $\mu_0$. \Cref{lem:ccdirectsum} gives an $\Omega(n/k^2)$ bound on information complexity ($w.r.t.$ $\sim$ uniform distribution on every player's input) for any $k$-party communication protocol that solves the planted bi-clique problem when $n=k$.

In Problem \ref{prob:informal_setup}, the planted distribution only samples from the non-uniform distribution on $k=\gamma n$ number of rows $R$, and this set $R$ is chosen uniformly at random. Thus, we want to use the fact that any multi-pass streaming algorithm that doesn't know $R$ needs to solve, the distributed detection problem for the special case, where all rows are planted, for multiple instances embedded in the stream, \textit{simultaneously}. Such arguments are usually made using direct sum theorems in communication complexity, but it is challenging to use these techniques to prove optimal\footnote{Here, by optimal, we mean that we get non-trivial memory lower bounds even for $k=n^{1/2-\epsilon}$ for constant $\epsilon>0$.} memory bounds that we get. One challenge is that as we want hardness for a distributional problem, we are aiming for multi-pass lower bounds for stochastic streams and hence, we cannot strategically embed multiple instances of the communication problem in the stream, so as to force a \textit{single} time-step to communicate a lot to the next time-step.

Recent work of \cite{braverman2020coin} introduced a new information cost notion for one-pass streaming algorithms, which is amenable to memory lower bounds for stochastic streams. We leverage the multi-pass information cost (MIC) notion, introduced in \cite{braverman2024new}, to prove our result. Briefly, these information cost notions measure the information a time-step needs to retain about the input stream \textit{on average}. As these information cost notions are only meaningful for product distributions over the time-steps, it is crucial that we prove the IC bound $w.r.t.$ a product distribution over the players' inputs (\Cref{lem:ccdirectsum}). Secondly, as the rows -- where the communication problem is embedded in the stream -- is random, to be able to use tensorization properties of the MIC notion, we need the distribution over every sample without a plant to have the same distribution as the no distribution. Hence, the proof heavily relies on the fact that we prove information complexity bounds $w.r.t.$ $\mu_0^{\otimes d/t}$, which is the distribution for every sample under the null distribution. Once we have the right IC bounds, to get a multi-pass information cost lower bound that depends quadratically on $1/\gamma$, we use a similar argument to one made in \cite{braverman2024new} for lifting the hardness of MostlyEQ to the needle problem.

\subsection{Applications to planted biclique and other graph streaming problems}
\label{subsec:graph}

Our general-purpose framework for proving memory lower bounds for detecting planted structures yields multi-pass memory lower bounds for the classical planted bi-clique problem in bipartite graphs, as well as a semi-random version of the planted bi-clique problem. These lower bounds in turn allow us to derive worst-case memory-bounded hardness of approximation for fundamental graph problems like the \textit{Maximum Bi-Clique} and \textit{Densest at-most-$\beta$ Subgraph}, in a natural (potentially stronger) vertex-arrival streaming model. Obtain these bounds require a careful application of our general framework, and we elaborate on some of the technical aspects below.

Consider first the planted bi-clique problem. We first note that to get a lower bound for the planted bi-clique problem  where $k$ out of $n$ coordinates are planted uniformly at random, it is sufficient to show lower bounds for an easier version where we partition the $n$ coordinates into  $n/t$ subsets of size $t \ge k$, and all the planted $k$ co-ordinates belong to one subset in the partition. To see this, note that an algorithm $\mathcal{A}$ that can solve the general non-partition version of the problem can be used to solve the partition version, by simply permuting the $n$ coordinates of each input according to a consistent, uniformly random permutation, and feeding this input to $\mathcal{A}$. Now, to show a lower bound for this partition version of the problem, we first need to fix the size of the partition. Note that if $t = o(k^{2})$, then simply counting the number of ones in each subset of the partition suffices to distinguish between the uniform and planted distributions, since with high probability only $t/2\pm O(\sqrt{t})$ ones are observed in each subset for the uniform distribution. Therefore, we will choose $t = \Omega(k^2)$. 

The next step, which is also the key technical step in all applications of our framework, is to appropriately truncate the distributions $\mu_0$ and $\mu_\theta$ to obtain new distributions $\tilde{\mu}_0$ and $\tilde{\mu}_\theta$ such that: (1) $\tilde{\mu}_0$ and $\tilde{\mu}_\theta$ are  close to $\mu_0$ and $\mu_\theta$ respectively; (2) but they satisfy that $\Eb_\theta\tilde{\mu}_\theta := \tilde{\mu}_1$ is pointwise upper-bounded by $c \cdot \tilde{\mu}_0$ for some constant $c$. The reason this truncation is necessary is that the original distributions $\mu_0$ and $\mu_\theta$ \emph{do not} satisfy that $\Eb_\theta{\mu}_\theta := {\mu}_1$ is pointwise upper-bounded by $c \cdot {\mu}_0$ for some constant $c$. This is because $t$-bit strings which have exactly $k$ ones have roughly $2^k$ more probability under $\Eb_\theta{\mu}_\theta$ than under ${\mu}_0$, since strings under ${\mu}_\theta$ always have at least $k$ ones. To address this, we restrict to typical strings which have $t/2 \pm O(\sqrt{t\log t})$ ones, and define $\tilde{\mu}_0$ and $\tilde{\mu}_\theta$ by restricting $\mu_0$ and $\mu_\theta$ to such strings. This truncation allows us to bound $c$ by a constant, and effectively leverage the lower bound from the general framework.

Our next result shows a stronger memory lower bound for the semi-random version of the planted bi-clique problem, where an adversary is monotone---it can only remove ones from non-planted locations. Our main insight here is to relate this problem with a monotone adversary to a slightly modified version of the planted bi-clique problem itself. While in the standard planted bi-clique problem, we plant the $\allone_k$ vector at some subset of $k$ coordinates, we instead think of the version where we plant an arbitrary random vector $v \in \{0,1\}^k$ at these coordinates. To relate this "pattern" planted bi-clique problem to monotone adversaries, note that the generated samples for a pattern vector $v$ with $k'$ ones are equivalent to samples from a planted bi-clique problem with a bi-clique of size $k'$, but in the presence of a monotone adversary which forces a consistent set of $k-k'$ non-planted coordinates in the planted rows to be 0. The pattern planted bi-clique problem is at least as hard as the standard problem, and in fact we show a stronger lower bound for it. This is because for this problem we can choose the size of the partitions to be as small as $t=k$, since the vector that we plant within the partition is also a uniformly random vector, and therefore the number of ones in the partition is still typical. 
This allows us to improve on the memory lower bound for this problem by a factor of $k$. 

Finally, we outline how our memory lower bounds above for the planted problems allow us to derive hardness results for approximating both, the densest at-most $\beta$ subgraph, as well as the maximum bi-clique in undirected graphs, in the vertex arrival model. Both these results are for undirected graphs, whereas the planted bi-clique lower bounds stated above are for bipartite graphs. In transferring the hardness to these undirected graph applications, we need to transition to a different streaming model, where vertices arrive in a worst-case order, and connectivity is only revealed to vertices that have previously occurred in the stream. For both the applications, the technique is similar: if the bipartite graph is planted, the translated undirected graph has a sizable edge density on some subgraph, and also a sizable bi-clique. However, if there is no plant, using standard concentration arguments, we can show these quantities to be small in the translated graph. Thus, if we had an accurate approximation, we can use it to figure out which case we are in. We note that for the densest subgraph application, we crucially rely on our framework allowing us to instantiate the hardness of the planted bi-clique problem with $q$ (the Bernoulli parameter at the non-planted locations) being as small as $\frac{\log n}{\beta}$, as opposed to $q=1/2$. 

\subsection{Applications to detecting \texorpdfstring{$\sparsity$}{k}-sparse Gaussians and sparse PCA}\label{sec:learning_overview}

We now discuss our proof for the detecting sparse Gaussians and for the sparse PCA detection problem. For these applications, appropriately truncating the distributions is more challenging and subtle than for the planted biclique problem.

We first sketch the proof for the sparse Gaussian detection problem. Recall that in the $d$ dimensional $\sparsity$-sparse Gaussian detection problem, in the planted distribution, with probability $(1-q)$ we get samples from $N(0,I)$ and with probability $q$ we get samples from $N(\theta,I)$, where $\theta$ is $\sparsity$-sparse. In the null distribution, we always get samples from $N(0,I)$. For simplicity, we first consider a simpler version of the problem where $\| \theta \|_2 = 1$. The proof here generally follows a similar outline to the planted biclique problem. As in the planted biclique case, we partition the coordinates into sets of size $t$. In this special case of the sparse Gaussian mean problem, we take $t=\sparsity$, and $\mu_1=N((1/\sqrt{t})\allone_t,I)$. Notice that there is no parameter $\theta$ to choose in the planted case in this partition version of the problem, and the co-ordinates in the chosen set in the partition are simply sampled from $\mu_1$. Our goal is now to show that $\mu_1$ is pointwise upper-bounded by $c \cdot \mu_0$ for some that $c$ that is not too large. In this case for any $x\in \R^t$, $\mu_1(x)/\mu_0(x)=\exp((1/\sqrt{t})\sum_j x_j)$. Note that $\sum_j x_j$ can be unbounded, therefore, as in the planted bipartite case, we need to truncate the distributions $\mu_1$ and $\mu_0$. We can truncate the distributions to the set $\{x \in \R^t: \sum_j x_j\le \sqrt{Ct}\}$ for some $C$. This is satisfied with high-probability, and allows us to bound $\mu_1(x)/\mu_0(x)$. Using our result for the general planted detection setup (informal version in \Cref{thm:planted_informal}) we get a $s \cdot n \ge \Omega\left(\frac{d^{0.99}}{p\cdot  \sparsity q^2 }\right)$ lower bound for $p$-pass, $s$-bit algorithms which use $n$ samples. 

The general Gaussian case with signal strength $\alpha$ (i.e. non-zero co-ordinates of $\theta$ are $\alpha)$ has a similar outline, but requires a much more careful analysis to get the dependence on the signal strength parameter $\alpha$. First, we show via a reduction that to show a bound for the original problem where $\theta$ is exactly $\sparsity$-sparse it suffices to show a lower bound for the case where each co-ordinate of the planted mean vector $\theta$ is non-zero with probability $\sparsity/t$. Independence across co-ordinates of $\theta$ facilitates the analysis, and for this distribution of $\theta$ we can show that for $\mu_1=\E_{\theta}\mu_\theta$, $\mu_1(x)/\mu_0(x)$ is bounded if $\sum_{j=1}^t \exp(\alpha x_j)$ is bounded---analogous to the simple case of $\alpha=1/\sqrt{t}$ sketched above---as long as $t$ is sufficiently larger than $\sparsity^2$. This suggests a truncation: we truncate the distributions $\mu_0$ and $\mu_{\theta}$ to $x$ which satisfy an appropriate, $\alpha$-dependent  bound on $\sum_{j=1}^t \exp(\alpha x_j)$. Finally, we show that the truncated distributions are close to the original ones using certain concentration bounds.

We now sketch the proof for the sparse PCA detection problem. Recall that in the sparse PCA detection problem the goal is to distinguish whether samples are drawn from the standard Gaussian distribution $\mu_0 = N(0,I_d)$ or from $\mu_\theta = N(0,\Sigma)$ for $\Sigma=I_d + \alpha \theta \theta^T$ for a $\sparsity$-sparse unit vector $\theta$. In this case, $\mu_\theta(x)/\mu_0(x)$ depends on $\exp \left (  \frac{\alpha}{2(\alpha+1)}(x^\intercal \theta)^2 \right )$. The fact that there is a quadratic instead of a linear term in the exponent makes truncation significantly more difficult here than  in previous settings, because $\E_{x\sim N(0,I)}[\exp(c x^2)]$ diverges for $c\ge1/2$. To see the challenge this poses, consider the following quantity for some $\theta \in \Omega$ (where $\Omega$ is the domain of the parameters, such as all $k$-sparse unit vectors), 
\begin{align*}
    \E_{x\sim \mu_\theta}\left[\frac{\mu_1(x)}{\mu_0(x)}\right].
\end{align*}
For all our previous applications such as planted biclique and sparse Gaussian mean detection, $
    \E_{x\sim \mu_\theta}\left[\frac{\mu_1(x)}{\mu_0(x)}\right]
$ is a \emph{constant}. Intuitively, concentration bounds then allow us to show that with high probability over $x\sim \mu_\theta$, $\frac{\mu_1(x)}{\mu_0(x)}$ is bounded. This allows us to find a truncation set $T$ such that with high probability $x$ lies in $T$  for all $\mu_\theta $, and moreover $\frac{\mu_1(x)}{\mu_0(x)}$ is bounded for all $x\in T$, which allows us to use our general framework to derive lower bounds. In the case of sparse PCA  detection, $
    \E_{x\sim \mu_\theta}\left[\frac{\mu_1(x)}{\mu_0(x)}\right]$ is \emph{unbounded} if $\alpha\ge1$, due to the Gaussian integral diverging. This is the reason why our bounds for sparse PCA detection only hold for small constants $\alpha$. In addition, we have to consider a more structured version of the problem, where the co-ordinates of $\theta$ are divided into blocks of size $\ell$, and we uniformly select one of the blocks and set all the co-ordinates in that block to $1/\sqrt{\ell}$, with the remaining co-ordinates being 0. With these assumptions, we can then derive a suitable truncation which is satisfied with high probability, and for which $\mu_1(x)/\mu_0(x)$ is bounded.

    \subsection{Organization of the paper}
In \Cref{sec:prelims}, we define preliminaries and notations. \Cref{sec:main_tech} formally defines the general problem (\Cref{prob:informal_setup}) of detecting planted structures, and proves the main lower bound (\Cref{thm:planted_informal}) for this problem. Next, in \Cref{sec:biclique}, we instantiate the planted bi-clique problem within this framework, and show the formal memory lower bound (\Cref{thm:clique_informal}) for it. We also outline the densest at-most $\beta$ subgraph application (\Cref{cor:densest-informal}) here. In \Cref{sec:semi-random-planted-biclique}, we continue to study the planted bi-clique problem in the presence of a monotone adversary, and derive a memory lower bound (\Cref{thm:monotone_informal}). Here, we also state the application about approximating the maximum bi-clique (\Cref{cor:max-biclique-informal}). \Cref{sec:gaussians} proves our results on detecting sparse Gaussians (\Cref{thm:gaussian_informal}). Finally, \Cref{sec:pca} shows our bounds for the sparse PCA detection problem (\Cref{thm:pca_informal}).

%% file: prelims.tex
\section{Preliminaries}\label{sec:prelims}
We use the notation $[n]$ to denote the set $\{1,2,\ldots, n\}$. Given an $n$-bit vector $x$, we use $x_S$ to denote the projection of $x$ on set $S$ (ordered lexicographically), that is, $\{x_i\}_{i\in S}$. We use $ \allone$ to denote the all 1s vector when dimension is clear from the context. We use $|x|$ to denote the number of 1s in $x$. Given two distributions $\mu_0, \mu_1:\calX\rightarrow [0,1]$, we say $\mu_0\le c\cdot \mu_1$, if for all $x\in \calX$, $\mu_0(x)\le c\mu_1(x)$.
We use capital letters or bold letters, such as $X, Y, Z$, $\boldsymbol\theta$, etc., to denote random variables and $x, y, z,\theta$ etc., to denote the values these random variables take. Given a probability distribution $\calD:\X\rightarrow [0,1]$, we use the notation $x\sim \calD$ when value $x$ is sampled according to distribution $\calD$.   Similarly, we use the notation $z\sim Z$ to denote the process that $Z$ takes value $z$ with probability $\Pr[Z=z]$. We use $\Ber(q)$ to denote the Bernoulli distribution which takes value $1$ with probability $q$ and $0$ with probability $1-q$. We use notations $\bbE[Z]$ to denote the expectation and  of random variable $Z$, and $\bbE[Z|Y=y]$ to denote the expectation of the random variable $Z$ conditioned on the event $Y=y$. We also use the notation $Z_{\mid Y=y}$ to denote the random variable $Z$ conditioned on the event $Y=y$.

\paragraph{Basics of information theory.} Given two distributions $P,Q$ over $\mathcal X$, $D_{KL}(P||Q)$ represents the Kullback–Leibler (KL) divergence of $P$ w.r.t. $Q$, that is,
\[D_{KL}(P||Q)=\int_{x\in\calX}\log{\frac{dP(x)}{dQ(x)}}dP(x).\]
We will use $\log$ base 2 unless stated otherwise. Note that $D_{KL}(P||Q)\ge 0$ for all $P$ and $Q$. 

For random variables $X$ and $Y$ (not necessarily discrete) having joint distribution $P_{XY}$, and  marginals $P_X$, $P_Y$ respectively, the mutual information between $X$ and $Y$, denoted $I(X;Y)$, is defined as
\begin{align}
    \label{eqn:mutual-information}
    \I(X;Y) = D_{KL}(P_{XY} || P_X \otimes P_Y),
\end{align}
where $P_X \otimes P_Y$ is the product distribution of the marginals of $X$ and $Y$. Note that $I(X;Y)=I(Y;X)$, and $I(X;Y) \ge 0$ always, by the non-negativity of KL divergence. Note also that mutual information between $X$ and $Y$ can be expressed in terms of the KL divergence as follows:
\[\I(X;Y)=\bbE_{y\sim Y}\left[D_{KL}(X_{\mid Y=y}|| X)\right].\]

$\I(X;Y|Z)$ represents the mutual information between $X$ and $Y$ 
conditioned on the random variable $Z$; and is defined as
\begin{align}
    \label{eqn:conditional-mutual-information}
    \I(X;Y|Z) = \bbE_{Z}\left[D_{KL}(P_{XY|Z} || P_{X|Z} \otimes P_{Y|Z})\right],
\end{align}
where $P_{XY|Z}$ denotes the joint distribution of $X,Y$ conditioned on $Z$, and $P_{X|Z}$ and $P_{Y|Z}$ denote the marginal distributions of $X$ and $Y$ conditioned on $Z$, respectively.

We will routinely use the chain rule for mutual information:
\begin{equation}
    \label{eqn:chain-rule}
    I(X;Y,Z) = I(X;Y) + I(X;Z|Y).
\end{equation}

Using the chain rule, we can derive the following simple facts, which will also come in handy:
\begin{enumerate}
    \item If $I(A,B;C|D)=0$, then $I(A,B;C|B,D)=0$. This follows since $I(A,B;C|D)=I(C;B|D)+I(C;A|B,D)$ by the chain rule; since $I(A,B;C|D)=0$, by non-negativity of the chain rule, we have that $I(C;B|D)=I(C;A|B,D)=0$. Finally, by another application of the chain rule, $I(A,B;C|B,D)=I(C;A|B,D)+I(C;B|B,D)$; the first summand is 0 by the preceding argument, and the latter summand is 0 since conditioning on $B$ fully determines $B$.
    \item If $I(A;B|C,D)=0$, then $I(C;B|D,A) \le I(C;B|D)$. To see this, apply the chain rule twice on $I(C,A;B|D)$, to get
    \begin{align*}
        I(C,A;B|D) &= I(C;B|D) + I(A;B|C,D) \\
        &=I(A;B|D)+I(C;B|A,D).
    \end{align*}
    Since $I(A;B|C,D)=0$, we get that $I(C;B|A,D)=I(C;B|D)-I(A;B|D)$. Since $I(A;B|D) \ge 0$, the fact follows.
\end{enumerate}

We will also use Hellinger distance and TV distance as other measures of distance between two distributions. For distributions $P$ and $Q$ over $\mathcal X$ having densities $p$ and $q$ respectively, these quantities are defined as follows,
\[h^2(P||Q)=1-\int_{x\in\calX} \sqrt{p(x)\cdot q(x)}dx, \quad \text{ and } \quad ||P-Q||_{TV}=\frac{1}{2}\int_{x\in\calX}|p(x)-q(x)|dx.\]

\paragraph{Multi-pass streaming algorithms.} Given a stream of $n$ input elements, $x^1,\ldots,x^n$, we say $\sfM$ is a $p$-pass algorithm (for $p\ge 1$) when it goes over the entire stream $p$ times in order. We use $m_{(\ell,i)}$, for $\ell\in [p], i\in [n]$, to denote the memory state of $\sfM$ in the $\ell$-th pass after reading $i$ input elements. Let $m_0=m_{(1,0)}$ denote the starting memory state and for ease of notation, let $m_{(\ell+1,0)}=m_{(\ell,n)}$ for all $\ell\in[p]$. When the distribution on the input stream is specified, we will use $\sfM_{(\ell,i)}$, for $\ell\in [p], i\in \{0,1,\ldots,n\}$, to denote the random variable over the corresponding memory states. We will use the random variables $\X^1,\ldots, \X^n$ to denote the joint distribution over the input stream.

We use notation $[a,b]$ in the subscript to represent random variables indexed from $a$ to $b$, for example, $\sfM_{([1,\ell],i)}$ represents $i$-th memory states for the first $\ell$ passes, that is, $\sfM_{(1,i)},\ldots, \sfM_{(\ell,i)}$. We use notations $<b$, $\le b$ in the subscript to represent all the corresponding random variables with index less than $b$ or at most $b$ respectively. For example, $\sfM_{(\ell,\le i)}$ represents random variables $\sfM_{(\ell,[0,i])}$. 

We will require the following result of \cite{braverman2024new} which establishes independence between inputs and private randomness in two segments of the stream, once we condition on the public randomness and the memory states at two different time-steps for all passes. While \cite{braverman2024new} do not provide an explicit proof of this particular result, we give a proof in \Cref{sec:appendix-multi-ic} for completeness.

\begin{restatable}[Claim 3.4 in \cite{braverman2024new}]{lemma}{lemmamicindependence}
    \label{lemma:mic-independence}
    Consider a stream $X^1,\dots,X^n$ from a product distribution, and let $\sfM$ be a $p$-pass streaming protocol that uses public randomness $P$ and private randomness $R^\sfM=\{R^\sfM_{l,i}\}_{l \in [p], i \in [n]}$, where the private randomness at every step is mutually independent, as well as independent of the public randomness. Then, for any $i,j \in [n]$, $i < j$, and any $l \in [p]$, it holds that:
    \begin{align}
        &I(X^{[i, j-1]}, R_{([p], [i, j-1])}\;;\; X^{[1,i-1]}, R_{([p], [1,i-1])}, X^{[j,n]}, R_{([p], [j,n])} \mid %
        \sfM_{<l, i-1}, \sfM_{<l, j-1}, P) = 0, \label{eqn:mic-independence-1}\\
        &I(X^{[i, j-1]}, R_{([p], [i, j-1])}\;;\; X^{[1,i-1]}, R_{([p], [1,i-1])}, X^{[j,n]}, R_{([p], [j,n])} \mid %
        \sfM_{\le l, i-1}, \sfM_{<l, j-1}, P) = 0. \label{eqn:mic-independence-2} %
    \end{align}
\end{restatable}

We will use the following notion of information cost for multi-pass streaming algorithms, introduced by \cite{braverman2024new}. For a given distribution $\mu$ over $\X^1,\ldots, \X^n$, the information cost of a $p$-pass streaming protocol $\sfM$, which uses public randomness $P$, is given by:
\begin{align}
\mic(\sfM,\mu)=&\sum_{\ell=1}^{p}\sum_{i=1}^n\sum_{j=1}^i \I\left(\sfM_{(\ell,i)};X^{j}\mid \sfM_{(\leq \ell,j-1)}, \sfM_{(\leq \ell-1,i)}, P\right)\nonumber\\
&+\sum_{\ell=1}^{p}\sum_{i=1}^n\sum_{j=i+1}^n \I\left(\sfM_{(\ell,i)};X^{j}\mid \sfM_{(\leq \ell-1,j-1)}, \sfM_{(\leq \ell-1,i)}, P\right). \label{eqn:mic-def}
\end{align}
Here, the random variables for the memory states depend both on the randomness of the input as well as private and public randomness used by the algorithm. When $\mu$ is clear from context, we will drop it from the notation.  

We will also require the following lemma established in \cite{braverman2024new}, which bounds the multi-pass information cost notion for any memory-bounded streaming algorithm.

\begin{restatable}[Lemma 1.1, \cite{braverman2024new}]{lemma}{lemmamicmemory}
    \label{lem:upperboundofmic}
    Let $(X^1,X^2,\cdots,X^n)$ be drawn from a product distribution $\mu$. Then, for any $p$-pass streaming algorithm $\sfM$ that uses public as well as private randomness, has memory size $s$ and runs on input stream $X^1,\cdots,X^n$, it holds that: 
    \[
    \mic(\sfM,\mu)\leq 2 p\cdot s\cdot n.
    \]
\end{restatable}

We note that \cite{braverman2024new} proved the above result in the setting where $\sfM$ uses only private randomness; essentially the same proof works for the definition of $\mic$ given in \eqref{eqn:mic-def} when $\sfM$ can additionally use public randomness, and we give the proof in \Cref{sec:appendix-multi-ic} for completeness.

%% file: framework.tex
\section{General Multi-IC Lower Bound for Distinguishing Problems}\label{sec:main_tech}

In this section, we will prove communication and memory lower bounds for a general distinguishing problem, where the goal is to detect if a submatrix has been planted with an \textit{outlier} distribution. Let $\calX, \Omega$ be two sets such that $\mu_0, \{\mu_\theta\}_{\theta\in \Omega}$ are distributions on $t$-dimensional vectors over $\calX$. Given a distribution $P$ over the parameter space $\Omega$, we denote the average distribution $\bbE_{\theta\sim P}\mu_\theta$ by $\mu_1$. Let $d,n>0$. We study the following distinguishing problem on $n\times d$ sized matrices, when each row of the matrix arrives in a stream.

\begin{problem} \label{prob:general}
Let $0<k\le n$. Let $\calT=\{T_{r}\}_{r\in[d/t]}$ be a partition of $[d]$, where $\forall r, |T_r|=t$. The goal is to distinguish between the following joint distributions on $d$-dimensional vectors $x^1,\ldots, x^n\in \calX^d$:
\begin{enumerate}
\item $\no$: $\forall i\in[n]$ and $\forall r\in[d/t]$, $x^i_{T_r}$ is drawn from $\mu_0$. 

\item $\yes^\calT$: Pick $r$ uniformly from $[d/t]$. $\forall i\in[n]$ and $\forall r'\neq r$, $x^i_{T_{r'}}$ is drawn from $\mu_0$. 

$R$ is drawn uniformly at random from all subsets of $[n]$ of size $k$. Pick $\theta\sim P$.

$\forall i\not \in R$, $x^i_{T_r}$ is drawn from $\mu_0$. Whereas, $\forall i\in R$, $x^i_{T_r}$ is drawn from $\mu_\theta$.

\end{enumerate}
We will refer to this distinguishing problem by $\Dis(\mu_0,\{\mu_\theta\}_{\theta\in \Omega},P, \calT,k,n)$.
\end{problem}

\begin{theorem}\label{lem:framework}
Let $1<k\le n$. Let $\mu_0, \mu_\theta, \theta\in \Omega$ be distributions on $t$-dimensional vectors in $\calX^t$, and $P$ be a distribution over parameter space $\Omega$ such that $\bbE_{\theta\sim P}\mu_\theta\le c\cdot \mu_0$. Let $\calT=\{T_{r}\}_{r\in[d/t]}$ be a partition of $[d]$, where $\forall r, |T_r|=t$. Then, any $p$-pass streaming algorithm (using public as well as private randomness) that solves the distinguishing problem $\Dis(\mu_0,\{\mu_\theta\}_{\theta\in \Omega},P, \calT,k,n)$ (as defined in \Cref{prob:general}) with large enough constant advantage, requires at least $\Omega\left(\frac{nd}{p\cdot c\cdot k^2t}\right)$ bits of memory. 
\end{theorem}

We will prove \Cref{lem:framework} using the following theorem on information cost of any multi-pass streaming algorithm that solves \Cref{prob:general}.  

\begin{theorem}\label{thm:micgeneral}
Let $1<k\le n$. Let $\mu_0, \mu_\theta, \theta\in \Omega$ be distributions on $t$-dimensional vectors in $\calX^t$, and $P$ be a distribution over parameter space $\Omega$ such that $\bbE_{\theta\sim P}\mu_\theta\le c\cdot \mu_0$. Let $\calT=\{T_{r}\}_{r\in[d/t]}$ be a partition of $[d]$, where $\forall r, |T_r|=t$. Let $\sfM$ be a $p$-pass streaming algorithm (using public as well as private randomness) that solves the distinguishing problem $\Dis(\mu_0,\{\mu_\theta\}_{\theta\in \Omega},P, \calT,k,n)$ with large enough constant probability. We add another pass to $\sfM$ such that in the $(p+1)$-th pass, $\sfM$ doesn't do any operations but stores $m_{(p,n)}$. Then, 

\[\mic(\sfM)\ge \Omega\left(\frac{n^2d}{ck^2t}\right).\]
Here, the multi-pass information cost is evaluated with respect to the distribution $\no$.
\end{theorem}

\Cref{lem:framework} then easily follows; for any $(p+1)$-pass streaming algorithm $\sfM$, that uses $s$ bits of memory, $\mic(\sfM)$ is upper bounded by $2(p+1)\cdot s\cdot n$ (\Cref{lem:upperboundofmic}).

To prove the above theorem, we will first show an information-complexity lower bound under the blackboard model, when the row set $R$ (in \Cref{prob:general}), that would contain the planted distribution, is known. To show the multi-pass information cost lower bound, we will then embed many such communication problems into the stream. The latter part is similar to the argument made in \cite[Section 5.2]{braverman2024new}. First, we study the $k$-player communication protocol that solves \Cref{prob:general} when $n=k$. Informally, in the no case, each player gets a $d$-dimensional vector from $\mu_0^{\otimes (d/t)}$, whereas is the yes case, one of the $d/t$ partitions is planted with $\mu_\theta$ for every player.

\subsection*{$\Dis(\mu_0,\{\mu_\theta\}_{\theta\in \Omega},P, \calT,k,k)$ under $k$-player number-in-hand communication model} 

Next, we will prove an $\Omega\left(\frac{d}{ct}\right)$ communication lower bound for any $k$-party communication protocol that solves the distinguishing problem $\Dis(\mu_0,\{\mu_\theta\}_{\theta\in \Omega},P, \calT,k,k)$. We define the communication problem below for completeness.

\begin{definition} ($k$-party \generalCC)
There are $k$ parties in the communication problem, where
the $i$-th party holds a $d$-dimensional vector $x^i\in\calX^d$. Let $\mu_0, \{\mu_\theta\}_ {\theta\in \Omega}$ be distributions on $t$-dimensional vectors in $\calX^t$, and $P$ be a distribution over parameter space $\Omega$ such that $\bbE_{\theta\sim P}\mu_\theta\le c\cdot \mu_0$. Let $\calT=\{T_{r}\}_{r\in[d/t]}$ be a partition of $[d]$, where $\forall r, |T_r|=t$. We promise that $(x^1,\ldots, x^k)$ are sampled from either of the following
distributions: 
\begin{enumerate}

\item (No) \label{itm:no1} $\forall i\in[k]$ and $\forall r\in[d/t]$, $x^i_{T_r}$ is drawn from $\mu_0$. 

\item (Yes) Draw $r$ uniformly from $[d/t]$. Draw $\theta\sim P$. $\forall i\in[k]$, $x^i_{T_r}$ is drawn from $\mu_\theta$ and $\forall r'\neq r$, $x^i_{T_{r'}}$ is drawn from $\mu_0$. 

\end{enumerate}
The goal of the players is to distinguish which case they are in.
Here, all parties communicate using a shared blackboard, and are allowed to use public as well as private randomness. 
\end{definition}

Given any $C$-bit communication protocol $\Pi$, we use $\Pi=(\Pi_0,\Pi_1,\ldots, \Pi_C)$ to also denote the transcript, that is, the concatenation of the public randomness with all the messages written on blackboard during the execution of $\Pi$. In the lemma below, we will measure the information complexity with respect to the No distribution. %

\begin{lemma}\label{lem:ccdirectsum}
For any communication protocol $\Pi$ that solves the $k$-party \generalCC, with probability at least $0.9$, we have that,
\[\I\left(\PiR ; X^1,X^2,\ldots,X^k\right)\ge \Omega\left(\frac{d}{c\cdot t}\right).\]
Here, $\forall i\in[k]$, $X^i$ is distributed according to the No case, and $\PiR$ is the distribution over transcripts, which depends on the input distribution and randomness used by the protocol.
\end{lemma}

Using the generalized distributed data processing inequality proven in \Cref{sec:dpis}, we can show an $\Omega(1/c)$ bound on the information complexity of any communication protocol that distinguishes between the two cases when $d=t$, that is, all players get a $t$-dimensional vector drawn either from $\mu_0$ or from $\mu_\theta$ (where $\theta\sim P$). We can then prove \Cref{lem:ccdirectsum} using a direct-sum argument. We first  state the result for when $t=d$.

\begin{lemma}\label{lem:ccone}(Corollary of \Cref{thm:gdpi} and \Cref{lem:boundhellinger})
Let $t>0$, $k>1$ and $\mu_0, \{\mu_\theta\}_{\theta\in \Omega}$ be distributions on $t$-dimensional vectors in $\calX^t$. Let $P$ be a distribution over parameter space $\Omega$ such that $\bbE_{\theta\sim P}\mu_\theta\le c\cdot \mu_0$. Let $\Pi$ be a $k$-party communication protocol that distinguishes between the following two cases, with probability at least $0.9$:
\begin{enumerate}
\item (No) All players get a vector independently drawn from $\mu_0$.
\item (Yes) $\theta$ is first drawn from P. All players then get a vector independently drawn from $\mu_\theta$.
\end{enumerate}
Then,
\[\I\left(\PiR ; Y^1,Y^2,\ldots,Y^k\right)\ge \Omega\left(1/c\right).\]
Here, $\forall i\in[k]$, $Y^i$ --- the input to the $i$-th player --- is distributed according to the No case, and $\PiR$ is the distribution over transcripts, which depends on the input distribution and randomness used by the protocol.
\end{lemma}
\proof[\textbf{Proof of \Cref{lem:ccdirectsum}}] The proof follows from a standard direct sum argument.
Let $\Pi$ be a protocol that solves the $k$-party \generalCC, with success probability 0.9. Using $\Pi$, we will construct a protocol $\Pi'$ that distinguishes between the cases when all players get a $t$-dimensional vector drawn either from $\mu_0$ or from $\mu_\theta$ (where $\theta\sim P$), with probability 0.9. We will also show that the information complexity of $\Pi'$ $w.r.t.$ the No distribution is at most $t/d$ times the information complexity of $\Pi$ $w.r.t.$ the No distribution. Formally, let $X^1,X^2,\ldots, X^k \in \calX^d$ be independently drawn from the No distribution for the $k$-party \generalCC. Let $Y^1,\ldots,Y^k\in\calX^t$ be independently drawn from $\mu_0$. Then, we will prove that 
\[\I\left(\PiR';Y^1,Y^2,\ldots,Y^k\right)\le \frac{t}{d}\cdot\I\left(\PiR;X^1,X^2,\ldots,X^k\right).\]
Here, $\PiR'$ and $\PiR$ are distributions over transcripts when the inputs are drawn from $Y$ and $X$ respectively. Hence, \Cref{lem:ccdirectsum} follows from \Cref{lem:ccone}.

\paragraph{Protocol $\Pi'$} Let $y^1,y^2,\ldots, y^k\in \calX^t$ be the input to the $k$-players. $\Pi'$ first samples $j$ uniformly at random from $[d/t]$ using public randomness. $\forall i\in[k]$, the $i$-th player prepares a $d$-dimensional vector $\tilde{x}^i$ as follows: set $\tilde{x}^i_{T_j}=y^i$ and for $j'\neq j$, draw $\tilde{x}^i_{T_{j'}}$ from $\mu_0$ using private randomness. Recall that $\calT=\{T_r\}_{r\in[d/t]}$ is a partition of $[d]$ into $t$ sized sets. All players run $\Pi$ on inputs $\tilde{x}^1,\ldots,\tilde{x}^k$ and answer whatever $\Pi$ answers. We will represent the corresponding random variables by capital letters.

Let us first calculate the success probability of $\Pi'$, which is the average of the probabilities that $\Pi'$ outputs "No" when the input to each player is $i.i.d.$ $\mu_0$ (No distribution) and $\Pi'$ outputs "Yes" when the input to each player is $i.i.d.$ $\mu_\theta$ (where $\theta$ is in turn drawn from $P$, the Yes distribution).%

\begin{align*}
&\frac{1}{2}\cdot\left(\Pr_{\forall i, y^i\sim \mu_0}\left[\Pi'(y^1,y^2,\ldots,y^k)=\text{"No"}\right]+\Pr_{\theta\sim P; \forall i, y^i\sim \mu_\theta}\left[\Pi'(y^1,y^2,\ldots,y^k)=\text{"Yes"}\right]\right)\\
&=\frac{1}{2}\cdot\left(\Pr_{\forall i, y^i\sim \mu_0; \;\;j\in_R[d/t]; \;\;\forall i, \;\;\tilde{x}^i_{T_j}=y^i,\;\;\tilde{x}^i_{T_{j'}}\sim \mu_0 \forall j'\neq j}\left[\Pi(\tilde{x}^1,\tilde{x}^2,\ldots,\tilde{x}^k)=\text{"No"}\right]\right. \\
&\left.\;\;\;\;\;\;\;\;\;\;\;\;\;\;\;\;\;\;\;\;\;\;\;\;\;\;\;\;\;\;\;\;\;\;+\Pr_{\theta\sim P; \forall i, y^i\sim \mu_\theta; \;\;j\in_R[d/t]; \;\;\forall i, \;\;\tilde{x}^i_{T_j}=y^i,\;\;\tilde{x}^i_{T_{j'}}\sim \mu_0 \forall j'\neq j}\left[\Pi(\tilde{x}^1,\tilde{x}^2,\ldots,\tilde{x}^k)=\text{"Yes"}\right]\right)\\
&=\frac{1}{2}\cdot\left(\Pr_{\forall i, j,\;\; x^i_{T_j}\sim \mu_0}\left[\Pi(x^1,x^2,\ldots,x^k)=\text{"No"}\right]+\Pr_{\theta\sim P; j\in_R[d/t]; \forall i, \;\;x^i_{T_j}\sim \mu_\theta,\;\;x^i_{T_{j'}}\sim \mu_0 \forall j'\neq j}\left[\Pi(x^1,x^2,\ldots,x^k)=\text{"Yes"}\right]\right).
\end{align*}
The last line is exactly equal to the probability of success for protocol $\Pi$ to distinguish between No and Yes distributions of the $k$-party \generalCC. Hence, $\Pi'$ succeeds with probability at least $0.9$. Next, we calculate the %
information complexity of $\Pi'$ when $Y^1,\ldots,Y^k$ are $i.i.d.$ $\mu_0$. Note that by definition of the protocol, the transcript $\PiR'=(J, \PiR'_{-J})$, where $J$ is the public randomness used by the protocol $\Pi'$ to sample uniformly from $[d/t]$, and $\PiR'_{-J}$ is the subsequent transcript generated when the players simulate $\Pi$ on the prepared inputs $\tilde{x}_1,\dots,\tilde{x}_k$. Furthermore, the public randomness $J$ is independent of the input $Y^1,\dots,Y^k$. Then, we have that
\begin{align}
    I(\PiR'; Y^1,\dots,Y^k) &= I(J, \PiR'_{-J}; Y^1,\dots,Y^k) \nonumber\\
    &= I(J;Y^1,\dots,Y^k) + I(\PiR'_{-J};Y^1,\dots,Y^k\mid J) \tag{chain rule} \nonumber\\
    &= I(\PiR'_{-J};Y^1,\dots,Y^k \mid J) \tag{$J$ independent of $Y^1,\dots,Y^k$} \nonumber \\
    &= \frac{t}{d}\cdot\sum_{j=1}^{d/t}{\I\left(\PiR'_{-j};\tilde{X}^1_{T_j},\ldots,\tilde{X}^k_{T_j}\right)} \nonumber\\
    \label{eq:ccds1}
    &=\frac{t}{d}\cdot \sum_{j=1}^{d/t}{\I\left(\PiR;X^1_{T_j},\ldots,X^k_{T_j}\right)}.
\end{align}

The last equality follows from the fact that the joint distribution on $(\PiR'_{-J}, \tilde{X})$ is the same as the joint distribution on $(\PiR,X)$ when $Y^1,\ldots,Y^k$ are $i.i.d.$ $\mu_0$. %

Now, for any $j \in [d/t]$, observe that by the chain rule,
\begin{align*}
    \I(\PiR,\{X^1_{T_{j'}},\ldots,X^k_{T_{j'}}\}_{j'< j} ; X^1_{T_j},\ldots,X^k_{T_j}) &= \I(\PiR; X^1_{T_j},\ldots,X^k_{T_j}) + \I(\{X^1_{T_{j'}},\ldots,X^k_{T_{j'}}\}_{j'< j} ; X^1_{T_j},\ldots,X^k_{T_j} \mid \PiR) \\
    &= \I(\{X^1_{T_{j'}},\ldots,X^k_{T_{j'}}\}_{j'< j} ; X^1_{T_j},\ldots,X^k_{T_j}) + \I(\PiR ; X^1_{T_j},\ldots,X^k_{T_j} \mid \{X^1_{T_{j'}},\ldots,X^k_{T_{j'}}\}_{j'< j}).
\end{align*}
Since for all $j\in[d/t]$, $X^1_{T_j},\ldots,X^k_{T_j}$ are independent of $\{X^1_{T_{j'}},\ldots,X^k_{T_{j'}}\}_{j'< j}$, we get that 
\begin{align*}
    &\I(\PiR; X^1_{T_j},\ldots,X^k_{T_j}) + \I(\{X^1_{T_{j'}},\ldots,X^k_{T_{j'}}\}_{j'< j} ; X^1_{T_j},\ldots,X^k_{T_j} \mid \PiR) = \I(\PiR ; X^1_{T_j},\ldots,X^k_{T_j} \mid \{X^1_{T_{j'}},\ldots,X^k_{T_{j'}}\}_{j'< j}) \\
    \implies \qquad & \I(\PiR; X^1_{T_j},\ldots,X^k_{T_j}) \le \I(\PiR ; X^1_{T_j},\ldots,X^k_{T_j} \mid \{X^1_{T_{j'}},\ldots,X^k_{T_{j'}}\}_{j'< j}). \tag{non-negativity of mutual information}
\end{align*}
Substituting in Equation \eqref{eq:ccds1} above, we then get that
\begin{align*}
\sum_{j=1}^{d/t}{\I\left(\PiR;X^1_{T_j},\ldots,X^k_{T_j}\right)}&\le \sum_{j=1}^{d/t}{\I\left(\PiR;X^1_{T_j},\ldots,X^k_{T_j}\mid \{X^1_{T_{j'}},\ldots,X^k_{T_{j'}}\}_{j'< j}\right)}\\
&=\I\left(\PiR;\{X^1_{T_{j'}},\ldots,X^k_{T_{j'}}\}_{j'\in [d/t]}\right)\tag{chain Rule}\\
&=\I\left(\PiR;X^1,X^2,\ldots,X^k\right).
\end{align*}
Plugging this back into Equation \eqref{eq:ccds1} completes the proof.
\qed

Now, we are ready to prove the information cost lower bound for multi-pass streaming algorithms that solve \Cref{prob:general}. 
\subsection{Proof of \Cref{thm:micgeneral}}
Let $\sfM$ be a $(p+1)$-pass algorithm that solves the distinguishing problem $\Dis(\mu_0,\{\mu_\theta\}_{\theta\in \Omega},P, \calT,k,n)$ with large enough constant probability, say $1-\delta$ (recall that, we added another pass that doesn't do any operations to $\sfM$). This implies that

\begin{align}
\nonumber
\frac{1}{2} \left(\Pr_{(x^1,x^2,\ldots,x^n)\sim \no}\left[\sfM(x^1,x^2,\ldots,x^n)=0\right]+\Pr_{(x^1,x^2,\ldots,x^n)\sim \yes^\calT}\left[\sfM(x^1,x^2,\ldots,x^n)=1\right]\right)\ge 1-\delta
\end{align}
Recall that $\calT=\{T_r\}_{r\in[d/t]}$ is a partition of $[d]$ into $t$ sized sets. Under $\no$, $\forall i\in[n], j\in[d/t]$, $x^i_{T_j}$ is drawn from $\mu_0$. Under $\yes^\calT$, first $j$ is chosen uniformly at random from $[d/t]$, $\theta\sim P$ and $R$ is chosen uniformly from $k$-sized subsets of $[n]$ (we will use the notation $R\sim \binom{[n]}{k}$ to denote this random process), such that $\forall i\in R$, $x^i_{T_j}$ is drawn from $\mu_\theta$, and everything else is drawn as in $\no$. Thus, we can rewrite the success probability of $\sfM$ as 

\begin{align}
\nonumber
&\frac{1}{2}\left(\Pr_{\forall i, j, \;x^i_{T_j}\sim \mu_0}\left[\sfM(x^1,x^2,\ldots,x^n)=0\right]\right.+\\
\label{eq:mic2}
&\;\;\;\;\;\;\left.\Pr_{\theta\sim P; \;j\in_R[d/t];\; R\sim \binom{[n]}{k};\; \forall i\in R, \;x^i_{T_j}\sim \mu_\theta;\;\;x^i_{T_{j'}}\sim \mu_0 \;\forall i,j'\;( i\not\in R\; \vee\; j'\neq j)}\left[\sfM(x^1,x^2,\ldots,x^n)=1\right]\right)\ge 1-\delta.
\end{align}
Let $q_R$ be the success probability of distinguishing between $\no$ and $\yes^\calT$, when the rows where $\mu_\theta$ is ``planted'', are fixed to be $R$, that is,
\begin{align}
\nonumber
q_R &= \frac{1}{2}\left(\Pr_{\forall i, j, \;x^i_{T_j}\sim \mu_0}\left[\sfM(x^1,x^2,\ldots,x^n)=0\right]\right.+\\
\label{eq:mic3}
&\;\;\;\;\;\;\;\; \left.\Pr_{\theta\sim P; \;j\in_R[d/t];\; \forall i\in R, \;x^i_{T_j}\sim \mu_\theta;\;\;x^i_{T_{j'}}\sim \mu_0 \;\forall i,j'\;( i\not\in R\; \vee\; j'\neq j)}\left[\sfM(x^1,x^2,\ldots,x^n)=1\right]\right).
\end{align}
By Equation \eqref{eq:mic2}, we have that $\bbE_{R\sim \binom{[n]}{k}}\left[q_R\right]\ge 1-\delta$. Let $\delta<0.01$, and we call a set $R$ good if $q_R\ge 0.9$. Then, with probability of at least $0.5$ (over $R\sim \binom{[n]}{k}$), $q_R\ge 1-2\delta \ge 0.9$ and $R$ is good.  Next, we will show that every good set $R$, using a reduction to communication protocols for $k$-party \generalCC \text{ and} \Cref{lem:ccdirectsum}, contributes $\Omega(d/ct)$ to the multi-pass information cost of $\sfM$ $w.r.t.$ $\no$. Recall that, 
\begin{align*}
\mic(\sfM)=&\sum_{\ell=1}^{p+1}\sum_{i=1}^n\sum_{j=1}^i \I\left(\sfM_{(\ell,i)};X^{j}\mid \sfM_{(\leq \ell,j-1)}, \sfM_{(\leq \ell-1,i)}, P\right)\\
&+\sum_{\ell=1}^{p+1}\sum_{i=1}^n\sum_{j=i+1}^n \I\left(\sfM_{(\ell,i)};X^{j}\mid \sfM_{(\leq \ell-1,j-1)}, \sfM_{(\leq \ell-1,i)}, P\right).
\end{align*}
\paragraph{Fix a good $R=\{i_1,i_2,\ldots,i_k\}$ in sorted order.} We will denote $R$'s contribution to $\mic(\sfM)$ by $\mic^R$, which is defined as %
\begin{align*}
\mic^R=&\sum_{\ell=1}^{p+1}\sum_{a=1}^k\sum_{b=1}^{a-1}\I\left(\sfM_{(\ell,i_a-1)}\;;\; X^{i_b}\mid \sfM_{(\le\ell,i_b-1)}, \sfM_{(<\ell,i_a-1)}, P \right)+\\
&+\sum_{\ell=1}^{p+1}\sum_{a=1}^k\sum_{b=a+1}^{k}\I\left(\sfM_{(\ell,i_a-1)}\;;\; X^{i_b}\mid \sfM_{(<\ell,i_b-1)}, \sfM_{(<\ell,i_a-1)}, P \right).
\end{align*}

Using $\sfM$, we will construct a communication protocol $\Pi=\Pi(R)$ %
for the $k$-party \generalCC, with success probability at least $0.9$, such that the information complexity of $\Pi$ is less than $\mic^R$. \Cref{lem:ccdirectsum} would then imply that $\mic^R\ge \Omega\left(\frac{d}{c\cdot t}\right)$. We restate this as the following formal claim.

\begin{restatable}{claim}{clmicr}\label{cl:micr}
For every good $R$ (where $q_R$ as defined above is at least 0.9), $\mic^R\ge \Omega\left(\frac{d}{c\cdot t}\right)$.
\end{restatable}
This claim is proved by converting the streaming algorithm into a communication protocol in the standard way so as to invoke \Cref{lem:ccdirectsum}, and using calculations similar those used in the proof of Claim 5.4 in \cite{braverman2024new}. We defer the details of this proof to \Cref{sec:appendix-multi-ic}.

We now show how to use \Cref{cl:micr} to get an $\Omega\left(\frac{n^2d}{ck^2t}\right)$ bound on $\mic(\sfM)$. The next argument is almost identical to \cite[Section 5.2]{braverman2024new}, with the difference in how $R$ is sampled. 

Notice that since $R$ is good with probability at least 0.5, the claim implies that $\bbE_{R\sim \binom{[n]}{k}} \mic^R\ge \Omega\left(\frac{d}{c\cdot t}\right)$. We will show that $\mic(\sfM)\ge \Omega\left(\left(\frac{n}{k}\right)^2\cdot \bbE_{R\sim \binom{[n]}{k}} \mic^R\right)$, which would suffice for \Cref{thm:micgeneral}.
We begin by writing
\begin{align*}
\bbE_{R\sim\binom{[n]}{k}}\mic^R=\;\;&\bbE_{R\sim\binom{[n]}{k}}\sum_{\ell=1}^{p+1}\sum_{a=1}^k\sum_{b=1}^{a-1}\I\left(\sfM_{(\ell,i_a-1)}\;;\; X^{i_b}\mid \sfM_{(\le\ell,i_b-1)}, \sfM_{(<\ell,i_a-1)}, P \right)+\\
&\;\;\;\;\;\;\;\;\;\;+\bbE_{R\sim\binom{[n]}{k}}\sum_{\ell=1}^{p+1}\sum_{a=1}^k\sum_{b=a+1}^{k}\I\left(\sfM_{(\ell,i_a-1)}\;;\; X^{i_b}\mid \sfM_{(<\ell,i_b-1)}, \sfM_{(<\ell,i_a-1)}, P \right).
\end{align*}
We will compare the first term in the expectation, with the first term of $\mic(\sfM)$, that is,
\[\sum_{\ell=1}^{p+1}\sum_{i=1}^n\sum_{j=1}^i \I\left(\sfM_{(\ell,i)};X^{j}\mid \sfM_{(\leq \ell,j-1)}, \sfM_{(\leq \ell-1,i)}, P\right).\]

For a random $R$, each term $\I\left(\sfM_{(\ell,i)};X^{j}\mid \sfM_{(\leq \ell,j-1)}, \sfM_{(\leq \ell-1,i)}, P\right)$ for an $(i,j)$ pair with $j\le i$, appears in $\mic_R$ with probability at most $ {{n-2}\choose{k-2}}/{{n}\choose {k}} = \frac{k(k-1)}{n(n-1)}$, since this happens only if both $i+ 1$ and $j$ are in $R$. Similarly, we will compare the second term in expectation with the second term of $\mic(\sfM)$, that is,
\[\sum_{\ell=1}^{p+1}\sum_{i=1}^n\sum_{j=i+1}^n \I\left(\sfM_{(\ell,i)};X^{j}\mid \sfM_{(\leq \ell-1,j-1)}, \sfM_{(\leq \ell-1,i)}, P\right).\]
For a random $R$, each term $\I\left(\sfM_{(\ell,i)};X^{j}\mid \sfM_{(\leq \ell-1,j-1)}, \sfM_{(\leq \ell-1,i)}\right)$ for an $(i,j)$ pair with $j> i+1$, appears in $\mic_R$ with probability at most $ {{n-2}\choose{k-2}}/{{n}\choose {k}} = \frac{k(k-1)}{n(n-1)}$, since again, this happens only if both $i+ 1$ and $j$ are in $R$. When $j=i+1$, no such term appears in $\mic^R$, as $b\neq a$. This implies that 

\[ \bbE_{R\sim \binom{[n]}{k}} \mic^R\le \frac{k(k-1)}{n(n-1)}\cdot\mic(\sfM) .\]

%% file: generalized_sdpi.tex
\subsection{Generalized distributed data processing inequalities}\label{sec:dpis}

\begin{theorem}\label{thm:gdpi}
    Consider a family of distributions $\{\mu_\theta\}:\calX\rightarrow [0,1]$ parameterized by a random variable $\boldsymbol{\theta}$, which takes values in some domain $\Omega$ and has distribution $P$. Let $\mu_1= \bbE_{\theta\sim P} [\mu_\theta]$. Consider the distributed detection setting where if $V=0$ then each party receives $X_i \sim \mu_0$ (for some distribution $\mu_0:\calX\rightarrow [0,1]$), and if $V=1$ then we first draw $\theta\sim P$, and then each party receives $X_i \sim \mu_\theta$. If $\mu_1\le c \mu_0$, then for some constant $K>0$, for any multi-party communication protocol $\Pi$,
    \begin{align}
        \Eb_{\theta\sim P}  \left[h^2( \PiR_{\mid V=0} \parallel \PiR_{\mid V=1,\bth=\theta}) \right]\le K(c+1) I(X; \PiR \mid V=0).\label{eq:distributed_dpi}
    \end{align}
Here, $\PiR_{\mid V=0}$ and $\PiR_{\mid V=1,\bth=\theta}$  represent the random variables for the transcript of protocol $\Pi$, when inputs to the parties are drawn from $\mu_0$ and $\mu_\theta$, respectively. 
\end{theorem}
\begin{proof}
    Our proof builds on the proof of Theorem 3.1 in \cite{braverman2016communication}. Let $\Pi$ be an $m$-party communication protocol ($m\ge 2$). We first note that since $X_i$'s are independent conditioned on $V=0$, that is, $\I(X_i; X_{<i} \mid V=0)=0$, the RHS of \eqref{eq:distributed_dpi} tensorizes and we get,
    \begin{align}
    \nonumber
        \I(X; \PiR \mid V=0)& = \sum_{i=1}^m \I(X_i; \PiR \mid V=0, X_{<i})\tag{Chain Rule}\\
        \label{eq:tensor}
        &\ge  \sum_{i=1}^m \I(X_i; \PiR \mid V=0).
    \end{align}
    Fix $i\in[m]$. To bound $\I(X_i; \Pi \mid V=0)$, we consider the following single-machine setting. Fix $\theta$. Let $W$ be a random variable which is uniformly distributed in $\{0,1\}$. Let data $X'$ be generated as follows: $X_{i}' \sim \mu^\theta_{W}$ (where $\mu^\theta_0=\mu_0$ and $\mu^\theta_1=\mu_\theta$) and for any $j \ne i$, $X_{j}' \sim \mu_{0}$. 
    We apply the protocol $\Pi$ on the input $X'$, and consider the resulting transcript $\Pi'$. We will also use $\Pi'$ to denote the one argument randomized function, that takes in $x'_i$, samples $x'_j\sim\mu_0\forall j\neq i$, and outputs $\Pi(x)$. Note that the function $\Pi'$ doesn't depend on $\theta$. Then $W\rightarrow X_i' \rightarrow \Pi'$ forms a Markov chain, and by the data processing inequality, %
    \begin{align*}
        \I(W ; \PiR') \le \I(X_i'; \PiR').
    \end{align*}
    Here, $\PiR'$ denotes the random variable for output of $\Pi'$. (When the distribution for input $x_i$ to $\Pi'$, say $x_i\sim Y$ , is not clear from the context, we will use $\Pi'(Y)$ to denote the random variable for output of $\Pi'$).
    Using Lemma 10 in \cite{braverman2016communication}, we can lower bound $\I(W ; \PiR') $ using the squared Hellinger distance,
    \begin{align}
        h^2 (\PiR'_{ \mid W=0} \parallel \PiR'_{ \mid W=1}) \le \I(W ; \PiR')\;\;\;
        \implies \;\;\;h^2 (\PiR'_{ \mid W=0} \parallel \PiR'_{ \mid W=1}) \le \I(X_i'; \PiR').\label{eq:sdpi1}
    \end{align}
    
Equation \eqref{eq:sdpi1} relates the mutual information and squared Hellinger distance for the single machine case; next we want to relate the single machine setting to the distributed setting. To do this, we first establish some notation. For any fixed vector $\bb=(\bb_1,\dots,\bb_m)\in \{0,1\}^m$, let $\mu_{\bb}^\theta$ denote the product distribution on $m$ inputs, where input to each machine $i$ is drawn independently from $\mu_{\bb_i}^\theta$. We use notation $\PiR_{\bb}^\theta$ to denote the random variable for transcript of protocol $\Pi$ when inputs $(x_1,\dots, x_m)\sim \mu_\bb^\theta$.

With this notation, we note that the random variable $\PiR'_{ \mid W=0}$ has distribution as $\PiR^\theta_{\zerob}$. And $\PiR'_{ \mid W=1}$ has the same distribution as $\PiR^\theta_{\eb_i}$. Here, $\eb_i$ is the standard basis vector with $1$ at the $i$th coordinate. Then we can rewrite \eqref{eq:sdpi1} as,
\begin{align}
        h^2 (\PiR^\theta_{\zerob} \parallel \PiR^\theta_{\eb_i}) \le \I(X'_i; \PiR')=\I(X'^\theta_i; \Pi'(X'^\theta_i)).\label{eq:sdpi2}
    \end{align}
  By taking expectation over $\theta\sim P$, we get 
\begin{align}
       \Eb_{\theta\sim P} \left[h^2 (\PiR^\theta_{\zerob} \parallel \PiR^\theta_{\eb_i}) \right]\le \Eb_{\theta\sim P}\left[ \I(X'^\theta_i; \Pi'(X'^\theta_i))\right].\label{eq:sdpi3}
    \end{align}

Next, we will show that $\Eb_{\theta\sim P}\left[ \I(X'^\theta_i; \Pi'(X'^\theta_i))\right] \le \frac{c+1}{2} \cdot \I(X_i; \PiR \mid V=0)$. Let $Y$ be a random variable that takes values according to $\mu_0$. First, note that \[\I(X_i; \PiR \mid V=0)=\I(Y; \Pi'(Y)),\]
as conditioned on $V=0$, $\forall i\in [m], X_i\sim \mu_0$; thus the joint distribution of $(X_i, \Pi)$ is identical to $(Y,\Pi'(Y))$.
Hence, it suffices to prove that $\Eb_{\theta\sim P}\left[ \I(X'^\theta_i; \Pi'(X'^\theta_i))\right] \le \frac{c+1}{2} \cdot \I(Y; \Pi'(Y))$.

We will first prove a generalization of Lemma 11 in \cite{braverman2016communication}, and then come back to proving the above inequality. 

\begin{lemma}\label{lem:lemma11}
Consider a family of distributions $\{\mu_{s}'\}$ parameterized by $s\sim S$, and let $\mu'= \Eb_{s\sim S} \left[ \mu_{s}'\right]$. For some distribution $\mu$, let $\mu\ge c \mu'$. Let $f(z)$ be a random function that depends only on $z$, and whose range is discrete. If $Z\sim \mu$ and $s\sim S, Z' \sim \mu'_s$, then we have that,
\begin{align*}
    \I(Z; f(Z) ) \ge c \cdot \I(Z'; f(Z') \mid S).
\end{align*}
\end{lemma}
\begin{proof}
    Since  $f$ is a random function which depends only on $z$ and $\mu \ge c \cdot \mu'$, we have 
    \begin{align}
        \I(Z; f(Z) ) = \Eb_{z\sim Z} \left[D_{KL}(f(z) \parallel f(Z) )\right]=\Eb_{z\sim \mu} \left[D_{KL}(f(z) \parallel f(Z) )\right] \ge c \cdot  \Eb_{z\sim \mu'} \left[D_{KL}(f(z) \parallel f(Z) )\right].\label{eq:mi_relation}
    \end{align}
    Note that,%
    \begin{align*}
        &\Eb_{z\sim \mu'} \left[D_{KL}(f(z) \parallel f(Z) )\right]
        =\Eb_{s\sim S}\Eb_{z\sim \mu'_s} \left[D_{KL}(f(z) \parallel f(Z) )\right]\\
        &=\Eb_{s\sim S}\int_z\left[\sum_{\pi} \Pr[f(z)=\pi]\log\left( \frac{\Pr[f(z)=\pi]}{\int_z \Pr[f(z)=\pi]d\mu(z)}\right)\right] d\mu'_s(z)\\
        &= \Eb_{s\sim S}\int_z\left[\sum_{\pi} \Pr[f(z)=\pi]\log\left( \frac{\Pr[f(z)=\pi]}{\int_z \Pr[f(z)=\pi]d\mu'_s(z)}\right)\right] d\mu'_s(z)\\ \\
        &\;\;\;\;\;\;\;\;\;\;\;\;\;+ \Eb_{s\sim S}\int_z\left[\sum_{\pi} \Pr[f(z)=\pi]\log\left( \frac{\int_z \Pr[f(z)=\pi]d\mu'_s(z)}{\int_z \Pr[f(z)=\pi]d\mu(z)}\right)\right] d\mu'_s(z)\\\\
        &= \Eb_{s\sim S}\Eb_{z \sim Z'_{\mid S=s}} \left[D_{KL}(f(z) \parallel f(Z'_{\mid S=s}) ) \right] +  \Eb_{s\sim S}\sum_{\pi}\left[\log\left( \frac{\int_z \Pr[f(z)=\pi]d\mu'_s(z)}{\int_z \Pr[f(z)=\pi]d\mu(z)}\right) \int_{z} \Pr[f(z)=\pi]d\mu'_s(z)\right]\\
        &= \I(Z';f(Z')\mid S) + \Eb_{s\sim S} D_{KL}(f(Z'_{|S=s}) \parallel f(Z) ).
    \end{align*}
    Since $KL$-divergence is always non-negative, plugging into Equation \eqref{eq:mi_relation}, we get that $$\I(Z; f(Z) )\ge c\cdot \I(Z';f(Z')\mid S).\qedhere$$
\end{proof}
We now apply Lemma \ref{lem:lemma11} with $\theta$ as the parameterization. We take $Z=Y$; $\mu=\mu_0$. We take $Z'$ conditioned on parameter $\theta$ to be $X_i'^\theta$; thus, $\mu_s'=\frac{\mu_0+\mu_\theta}{2}$ and $\mu'=\bbE_{\theta\sim P}\frac{\mu_0+\mu_\theta}{2}=\frac{\mu_0+\mu_1}{2}$.
 Since $\mu_1\le c\mu_0$, $\mu_0 \ge \frac{2}{c+1}\left(\frac{\mu_0+\mu_1}{2}\right)$. As $\Pi'$ is a randomized function only of $x_i$, Lemma \ref{lem:lemma11} says that, 
\begin{align*}
    \bbE_{\theta\sim P} \I(X'^\theta_i; \Pi'(X'^\theta_i))\le \frac{c+1}{2} \cdot \I(Y; \Pi'(Y)).
\end{align*}
Plugging this into \eqref{eq:sdpi3}, and using the fact that $\I(X_i; \PiR \mid V=0)=\I(Y; \Pi'(Y))$, we get,
\begin{align}
       \Eb_{\theta\sim P} \left[h^2 (\PiR^\theta_{\zerob} \parallel \PiR^\theta_{\eb_i}) \right]\le \frac{c+1}{2} \cdot \I(X_i; \PiR \mid V=0).\label{eq:sdpi4}
    \end{align}

Next, we lower bound the LHS of \eqref{eq:sdpi4}. In the following claim, we first show that the distributions $\PiR^\theta_{\bb}$ (over the transcripts under protocol $\Pi$) satisfies the cut-paste property developed in \cite{bar2004information} and used in \cite{braverman2016communication}, because after fixing $\theta$ the inputs to each machine are independent. The proof relies on basic properties of transcripts established in \cite{braverman2016communication} and is deferred to \Cref{sec:appendix-multi-ic}.

\begin{restatable}[Cut-paste property of the protocol]{claim}{lemcutpaste}\label{lem:cutpaste}
   For any $\theta$ and transcript $\pi$, and any $\bb^1, \bb^2, \bb^3, \bb^4$ with $\{b^1_i, b^2_i\} = \{b^3_{i}, b^4_{i}\}$ (in a multi-set sense) for every $i \in [m]$,
\[
\Pr\left[\PiR^\theta_{\bb^1}=\pi\right] \cdot \Pr\left[\PiR^\theta_{\bb^2}=\pi\right]  = \Pr\left[\PiR^\theta_{\bb^3}=\pi\right]  \cdot\Pr\left[\PiR^\theta_{\bb^4}=\pi\right],
\]
and therefore,
\[
h^2\left(\PiR^\theta_{\bb^1} \parallel \PiR^\theta_{\bb^2}\right) = h^2\left(\PiR^\theta_{\bb^3} \parallel \PiR^\theta_{\bb^4} \right).
\]
\end{restatable}

We now use a result for transcript distributions that satisfy the cut-paste property. 

\begin{claim}[Theorem E.1 in \cite{braverman2016communication}, corollary of Theorem 7 in \cite{jayram2009hellinger}] \label{thm:hellinger_rocks}
{Suppose a family of distribution} $\{P_{\mathbf{b}} : \mathbf{b} \in \{0,1\}^m\}$ {satisfies the cut-paste property: for any} $\mathbf{a}, \mathbf{b}$ \textit{and} $\mathbf{c}, \mathbf{d}$ {with} $\{a_i, b_i\} = \{c_i, d_i\}$ {(in a multi-set sense) for every} $i \in [m]$, $h^2(P_{\mathbf{a}}, P_{\mathbf{b}}) = h^2(P_{\mathbf{c}}, P_{\mathbf{d}})$. {Then we have}
\[
\sum_{i=1}^{m} h^2(P_{\mathbf{0}}, P_{\mathbf{e}_i}) \geq \Omega(1) \cdot h^2(P_{\mathbf{0}}, P_{\mathbf{1}}),
\] 
{where} $\mathbf{0}$ {and} $\mathbf{1}$ {are all 0’s and all 1’s vectors respectively, and} $\mathbf{e}_i$ {is the unit vector that only takes 1 in the} $i$ {th entry.}
\end{claim}
Using \Cref{thm:hellinger_rocks} and \Cref{lem:cutpaste},
\begin{align*}
   h^2(\PiR_{\zerob}^\theta \parallel \PiR_{\oneb}^\theta)&\le O(1)\sum_{i=1}^m h^2 (\PiR_{\zerob}^\theta \parallel \Pi_{\eb_i}^\theta),\\
   \implies \Eb_{\theta\sim P} \left[ h^2(\PiR_{\zerob}^\theta \parallel \PiR_{\oneb}^\theta) \right] &\le O(1)\sum_{i=1}^m \Eb_{\theta \sim P} \left[ h^2 (\PiR_{\zerob}^\theta \parallel \PiR_{\eb_i}^\theta) \right].
\end{align*}
Using \eqref{eq:sdpi4} to simplify the RHS above,
\begin{align*}
   \Eb_{\theta\sim P} \left[ h^2(\PiR_{\zerob}^\theta \parallel \PiR_{\oneb}^\theta) \right] &\le O(1)\sum_{i=1}^m \frac{c+1}{2} \cdot \I(X_i; \PiR \mid V=0),\\
   &\le O(1) \frac{c+1}{2} \cdot \sum_{i=1}^m \I(X_i; \PiR \mid V=0) \\
   &\le  O(1) \frac{c+1}{2} \I(X; \PiR \mid V=0),
\end{align*}
where in the last step we use the tensorization from  \eqref{eq:tensor}. Note that the distribution $\PiR^\theta_{\zerob}$ is identical to $\PiR_{\mid V=0}$, for all $\theta$. And distribution $\Pi_{\oneb}^\theta$ is identical to $\PiR_{\mid V=1,\bth=\theta}$. Therefore for some constant $K>0$ we get,
\begin{align*}
        \Eb_{\theta\sim P}  \left[h^2( \PiR_{\mid V=0} \parallel \PiR_{\mid V=1,\bth=\theta}) \right]\le K(c+1) \I(X; \PiR \mid V=0),
    \end{align*}
proving the theorem.\end{proof}

Finally, we prove that, for any protocol $\Pi$ that solves the distributed detection problem with probability at least $0.9$, the expected hellinger distance as in \Cref{thm:gdpi} is $\Omega(1)$. The proof of this result is a calculation and is deferred to \Cref{sec:appendix-multi-ic}.

\begin{restatable}{lemma}{lemboundhellinger}\label{lem:boundhellinger}
Consider a family of distributions $\{\mu_\theta\}:\calX\rightarrow [0,1]$ parameterized by a random variable $\boldsymbol{\theta}$, which takes values in some domain $\Omega$ and has distribution $P$. Consider the distributed detection setting where if $V=0$ then each party receives $X_i \sim \mu_0$ (for some distribution $\mu_0:\calX\rightarrow [0,1]$), and if $V=1$ then we first draw $\theta\sim P$, and then each party receives $X_i \sim \mu_\theta$. 
 Suppose there is an $m$-party communication protocol $\Pi$ that detects whether $V=0$ or $V=1$ with probability at least $0.9$. Then
    \begin{align*}
        \Eb_{\theta\sim P}  \left[h^2( \PiR_{\mid V=0} \parallel \PiR_{\mid V=1,\bth=\theta}) \right]\ge \Omega(1).
    \end{align*}
\end{restatable}

%% file: optbiclique.tex
\section{Multi-pass Streaming Lower Bound for Bi-Clique}
\label{sec:biclique}

In this section, we will prove multi-pass memory lower bounds for detecting planted bi-cliques in random bipartite graphs. Formally, we study the following distinguishing problem. 

\begin{problem}[Planted Bi-Clique] \label{prob:biclique}
    Let $1<k \le \min(m,n)$, and $0 < q \le 1/2$. The goal is to distinguish between the following joint distributions on $n$-bit vectors $x^1,\ldots, x^m$: %
    \begin{enumerate}
    \item $\Du$: $\forall i\in[m]$, $\forall j \in [n]$, $x^i_j$ is drawn as $\Ber(q)$.
    \item $\Dp$: $S\subseteq[n]$ is drawn uniformly at random from all subsets of $[n]$ of size $k$. $R$ is drawn uniformly at random from all subsets of $[m]$ of size $k$. 
    
    $\forall i\not \in R$, $\forall j \in [n]$, $x^i_j$ is drawn as $\Ber(q)$. \\
    $\forall i\in R$, $\forall j\in S$, $x^i_j=1$, and $\forall j\not\in S$, $x^i_j$ is drawn as $\Ber(q)$.
    \end{enumerate}
\end{problem}

Our main hardness result for \Cref{prob:biclique} is the following:

\begin{restatable}[Memory Lower Bound for Planted Bi-clique]{theorem}{theoremplantedclique}\label{thm:main}
    Let $0 < q \le 1/2$ and $0<k < O\left(\sqrt{\frac{q\cdot n}{\log(nm)}}\right)$. Any $p$-pass streaming algorithm (using public as well as private randomness), that distinguishes between $\Du$ and $\Dp$ (as in \Cref{prob:biclique}) when $x^1,x^2,\ldots,x^m$ arrive in a stream requires at least $\Omega\left(\frac{nmq}{pk^4\log(nm)}\right)$ bits of memory.
\end{restatable}
\begin{remark}
It is straightforward to modify the above theorem to obtain a memory lower bound of $\Omega\left(\frac{nmq}{pk'^2k^2\log(nm)}\right)$ for any $p$-pass streaming algorithm detecting cliques of size $(k'\times k)$ in $G(m,n,q)$. Taking $k'=\frac{k}{n}m$ yields the bound stated in Equation \eqref{eq:distributional}. The only subtlety is that \Cref{prob:dist} is a distributional version, whereas the target statement concerns exact detection of a clique of size $(k'\times k)$. This can be resolved by noting that any algorithm for the distributional version works for some $k'\approx \frac{k}{n}m$.
\end{remark}

Our objective will be to frame \Cref{prob:biclique} as an instantiation of \Cref{prob:general}, and thereafter leverage the lower bound for the general problem. For this, we will require partitioning $[n]$ into $n/t$ subsets of size $t\ge \Omega((k^2\log(nm))/q)$. We define the distributions $\mu_0, \{\mu_\theta\}_{\theta \in \Omega}$ over such $t$-sized subsets in terms of the following specialized distributions $\pt^0$ and $\left\{\pt^{1, S}\right\}_S$

Let $C>0$ be a large enough constant. We define $\pt^0$ to be the uniform distribution over the set
\begin{align}
    \label{eqn:def-T}
    T := \left\{x \in \{0,1\}^t: |x|\in \left[tq-C\sqrt{tq\log(nm)}, tq+C\sqrt{tq\log(nm)}\right]\right\}.
\end{align}
In words, this set comprises of all $t$-bit vectors $x$ that have $|x|$ in the \textit{typical range} of a $Bin(t,q)$ random variable.
Additionally, for $S \subseteq [t]$, $|S|=k$, we define $\pt^{1,S}$ to be the uniform distribution over the set
\begin{align}
    \label{eqn:def-T_S}
    T_S := \left\{x \in \{0,1\}^t:x_S=\allone, |x|\in \left[tq-C\sqrt{tq\log(nm)}, tq+C\sqrt{tq\log(nm)}\right]\right\}.
\end{align}
In words, this set comprises of all $t$-bit vectors $x$ that have $S$ set to $1$, and have $|x|$ in the same typical range as the support of $\pt^0$. %

For the distributions $\pt^0$ and $\pt^{1, S}$ thus defined, we can derive the technical condition necessary in \Cref{lem:framework} about $\Eb_S\left[\pt^{1, S}\right]$ being pointwise upper-bounded by $\pt^0$.

\begin{restatable}{claim}
{claimcboundptrunc}
    \label{claim:c-bound-p-trunc}
    Let $t \ge \frac{C k^2\log(nm)}{q}$ for some large enough constant $C$. Suppose $S$ is drawn uniformly at random from all subsets of $[t]$ of size $k$. Let $\mu_0$ and $\mu_1$ be the probability mass functions of $\pt^0$ and $\Eb_{S}[\pt^{1,S}]$ respectively. Then,
    \begin{align*}
        \mu_1 \le O(1) \cdot \mu_0.
    \end{align*}
\end{restatable}
The proof of \Cref{claim:c-bound-p-trunc} is a calculation, and is deferred to \Cref{sec:appendix-biclique}. It crucially uses two properties: that $S$ is chosen uniformly at random over subsets of $[t]$, together with the fact that the sparsity of vectors in the supports of both $\pt^0$ and $\pt^{1,S}$ are constrained to be in the typical range of a $Bin(t,q)$ random variable.

Now, we define the following distinguishing problem defined over a given fixed partition of $[n]$, when all sub-vectors in each vector in the stream are ``typical", and also, the planted set of coordinates (in the planted distribution) belongs wholly to a single $t$-sized partition. As we will prove formally, any algorithm that solves \Cref{prob:biclique} also solves the following distinguishing problem. %

\begin{problem}[Partition Planted Bi-Clique]
    \label{prob:partition}
    Let $0 < k \le \min(m,n)$ and $\frac{Ck^2\log(nm)}{q} \le t \le n$, where $C$ is a large enough constant. Let $n'=t \cdot \left\lfloor \frac{n}{t}\right\rfloor$. Let $\calT=\{T_{r}\}_{r\in[n'/t]}$ be a partition of $[n']$, where $\forall r, |T_r|=t$. The goal is to distinguish between the following joint distributions on $n'$-bit vectors $x^1,\ldots, x^m$:
    \begin{enumerate}
    \item $\no$: $\forall i\in[m]$ and $\forall r\in[n'/t]$, $x^i_{T_r}$ is drawn from $\pt^0$. 
    
    \item $\yes^\calT$: Draw $r$ uniformly from $[n'/t]$. $\forall i\in[m]$ and $\forall r'\neq r$, $x^i_{T_{r'}}$ is drawn from $\pt^0$. 
    
    Draw a uniformly random subset $S \subseteq T_r$ of size $k$. 
    
    Draw a uniformly random subset $R \subseteq [m]$ of size $k$. 
    
    $\forall i\not \in R$, $x^i_{T_r}$ is drawn from $\pt^0$, whereas, $\forall i\in R$, $x^i_{T_r}$ is drawn from $\pt^{1,S}$.
    \end{enumerate}
\end{problem}
Informally, under distribution $\no$, at each time-step, $x^i$ is drawn from the uniform distribution on $n'$-bit vectors conditioned on the number of ones in each partition being typical. On the other hand, under distribution $\yes^\calT$, for all but $k$ time-steps, $x^i$ is drawn as in distribution $\no$; otherwise, some partition in $x^i$ is drawn from the planted distribution while still conditioning on the number of ones being typical. Note again that we assume $k<< \sqrt{t}$.

\Cref{prob:partition} fits into the framework of \Cref{prob:general}, and we can therefore show the following hardness result for it.

\begin{lemma}[Memory Lower Bound for Partition Planted Bi-Clique]\label{lem:partition}
    Let $0 < k \le \min(m,n)$ and $\frac{Ck^2\log(nm)}{q} \le t \le n$, where $C$ is a large enough constant. Let $n'=t \cdot \left\lfloor \frac{n}{t}\right\rfloor$. Let $\calT=\{T_{r}\}_{r\in[n'/t]}$ be a partition of $[n']$, where $\forall r, |T_r|=t$. Then, any $p$-pass streaming algorithm (using public as well as private randomness), that distinguishes between $\no$ and $\yes^\calT$ (as defined in \Cref{prob:partition}) requires at least $\Omega\left(\frac{mn'}{pk^2t}\right)$ bits of memory.
    \end{lemma}
\begin{proof}
    Observe that \Cref{prob:partition} is a specific instantiation of \Cref{prob:general} with $n=m$ and $d=n'$. %
    Let $\mu_0, \mu_1$ denote the probability mass functions of $\pt^0$ and $\Eb_{S}[\pt^{1,S}]$ respectively. \Cref{claim:c-bound-p-trunc} shows that $\mu_1 \le O(1) \cdot \mu_0$, which satisfies the assumption of \Cref{lem:framework}. The result follows.
\end{proof}

With \Cref{lem:partition} established, we now sketch how \Cref{thm:main} is derived. The proof is a reduction: given a low-memory streaming algorithm $\mathcal{A}$ for \Cref{prob:biclique}, we obtain a low-memory streaming algorithm for \Cref{prob:partition}, with the partition size $t$ set to $t=\left\lceil \frac{Ck^2\log (nm)}{q} \right\rceil$. For simplicity, assume that $t$ divides $n$. The high-level strategy is as follows: Given an input from \Cref{prob:partition}, $\mathcal{A'}$ first uses public randomness to permute the coordinates of all the inputs consistently according to a uniformly random permutation, and feeds it to $\mathcal{A}$. This has the effect of turning the fixed partition into a uniformly random partition of $[n]$. Under the null, every group in the partition for each input is a draw from $\pt^0$. We now realize that the inputs under the null distribution of \Cref{prob:biclique} follow the same distribution, except that every group in the partition for each input is a draw from $\{0,1\}^t$ where every bit is drawn as $\Ber(q)$. Since $\pt^0$ is the uniform distribution over the subset of $\{0,1\}^t$ that is typical, these distributions are close. Conversely, under the planted distribution, we have that precisely $k$ inputs each have a group in the partition which is a draw from $\pt^{1,S}$, where $S$ is a random subset of size $k$ within the group. All the other groups are draws from $\pt^0$. Again, we realize that the inputs under the planted distribution of \Cref{prob:biclique} follow the same distribution, except that the planted groups are drawn from $\{0,1\}^t$ where every bit is drawn as $\Ber(q)$, and then a uniformly random subset of size $k$ is forced to 1. The typical support of this distribution is precisely the support $T_S$ of $\pt^{1,S}$, and hence we can again show that the planted distributions of both the problems are close in TV distance. So, $\mathcal{A}'$ can solve \Cref{prob:partition} by simply returning the output of $\mathcal{A}$. The formal details are given in \Cref{sec:appendix-biclique}.

\subsection{Application: Densest at-most \texorpdfstring{$\beta$}{B} Subgraph}
\label{sec:application-densest-at-most-d-subgraph}

\Cref{thm:main} allows us to derive a hardness of approximation result for the ``Densest at-most $\beta$ Subgraph Problem`` (see Section 3.12 in \cite{lanciano2024survey}) in the \textit{Vertex Arrival} streaming model. We define the model and problem here.

\begin{definition}[Vertex Arrival Streaming Model]
    \label{def:vertex-arrival-streaming-model}
    In the vertex arrival streaming model, the algorithm is presented with vertices from an undirected graph, and their neighbors amongst previously revealed vertices in an arbitrary, worst-case order. That is, the algorithm sees a stream $\{(v^i, E_{\le i})\}_{i \le n}$, where $E_{\le i}$ only contains edges that $v^i$ shares with vertices $v^1,\ldots,v^i$.
\end{definition}

\begin{problem}[Densest at-most $\beta$ Subgraph] \label{prob:densest-at-most-d-subgraph}
    Consider an undirected graph $G=(V,E)$ on $n$ vertices (self-edges allowed), and let $1 \le \beta \le n$. For any subset $H \subseteq V$, let $G(H)=(H, E(H))$ be the induced subgraph (i.e., $G(H)$ has vertex set $H$ and all edges $(u, v) \in E$ that satisfy $u, v \in H$). The edge density of $G(H)$ is defined as $\frac{|E(H)|}{|H|}$. The goal is to approximate the largest edge density among all subgraphs of size at most $\beta$ in $G$, i.e., $\max_{H \subseteq V, 1 \le |H|\le \beta}\left\{\frac{|E(H)|}{|H|}\right\}$. For $\alpha \ge 1$, an $\alpha$-approximation to a quantity $y$ is any number $x$ such that $(1/\alpha)y \le x \le y$.
\end{problem}

\begin{corollary}[Memory Lower Bound for Densest at-most $\beta$ Subgraph]\label{thm:densest-at-most-d-subgraph-approx}
    Consider any $\alpha \ge 1$ and $800\alpha\log n \le \beta \le o\left(\frac{n}{\alpha^2\log^2n}\right)$. Any $p$-pass streaming algorithm that approximates the size of the largest edge density among all subgraphs of size at most $\beta$ in an undirected graph that is presented in the Vertex Arrival Model to a factor $\alpha$ requires at least $\tilde{\Omega}\left( \frac{n^2}{p\beta\alpha^4}\right)$ bits of memory.
\end{corollary}

\begin{proof}
    We will reduce from the planted bi-clique problem with $m=n$, and an appropriate choice of $k$ and $q$. Let $\mathcal{A}$ be a $p$-pass streaming algorithm that uses $\tilde{o}\left( \frac{n^2}{p\beta\alpha^4}\right)$ bits of memory and always approximates the size of the densest at-most $\beta$ subgraph in a graph presented in the Vertex Arrival Model to a factor $\alpha$. Using $\mathcal{A}$, we will construct a $p$-pass streaming algorithm $\mathcal{A}'$ that processes $x^1,\ldots,x^n$ arriving in a stream, which solves \Cref{prob:biclique} for $k = 1200\alpha\log n$ and $q=\frac{\log n}{\beta}$, while using only $\tilde{o}\left( \frac{n^2}{p\beta\alpha^4}\right)=o\left(\frac{n^2q}{pk^4\log n}\right)$ bits of memory. This would contradict \Cref{thm:main}, and give us the claimed result.

    The algorithm $\mathcal{A}'$ operates as follows. Given an input stream $x^1,x^2,\ldots,x^n$, $\mathcal{A}'$ interprets each $x^i$ in the stream as a vertex $v^i$ in an undirected graph $G$, and presents it to $\mathcal{A}$ in the Vertex Arrival Model. For each $v^i$, it will read off connectivity to $v^1,\dots,v^i$ from $x^i_{[1:i]}$. That is, for $j\le i$, there is an undirected edge between $v^j$ and $v^i$ iff $x^i_{j}=1$. Note that $\mathcal{A}'$ can simulate this input space-efficiently for $\mathcal{A}$ (it only needs to keep track of a counter).

    Now, suppose $x^1,x^2,\ldots,x^n$ are drawn from $\Du$. Then, observe that the graph $G$ that $\mathcal{A}'$ presents to $\mathcal{A}$ is a random graph, where every $(v^i, v^j)$ is connected by an edge with probability $q$. We claim that the maximum edge density in this graph is at most $O(\log n)$ with high probability. To see this, observe that for every fixed $H \subseteq V$ of size at most $\beta$, we have that $\mathbb{E}\left[|E(H)|\right] = |H|^2q$, since $|E(H)|$ is precisely the sum of $|H|^2$ independent $Ber(q)$ random variables. By a Chernoff bound, we have that
    \begin{align*}
        \Pr\left[|E(H)| \ge (1+\delta)|H|^2q\right] \le \exp\left(-\frac{|H|^2q\delta^2}{2 + \delta}\right).
    \end{align*}
    Plugging in $\delta = \frac{100\log n}{|H|q}$, we get that
    \begin{align*}
        \Pr\left[|E(H)| \ge |H|^2q+100|H|\log n\right] \le \exp\left(-\Omega\left(|H|^2q \cdot \frac{\log n}{|H|q}\right)\right) =\exp\left(-\Omega\left(|H|\log n\right)\right) \le n^{-2|H|}.
    \end{align*}
    Therefore, with probability at least $1-n^{-2|H|}$, $|E(H)|$ is at most 
    $$|H|^2q+100|H|\log n = |H| \cdot |H|q + 100|H|\log n \le 101|H|\log n,$$
    where we plugged in $q=\frac{\log n}{\beta}$ and used $|H| \le \beta$ in the last inequality. By a union bound, the probability that $|E(H)| \le 101|H|\log n$ for \textit{every} $H \subseteq [n], |H| \le \beta$ is at most
    \begin{align*}
        \sum_{i=1}^\beta \binom{n}{i} \cdot n^{-2i} \le \sum_{i=1}^\beta n^{-i} \le O\left(\frac{1}{n}\right).
    \end{align*}
    
    This means that the edge density of every $H \subseteq [n], |H| \le \beta$ is at most $\frac{|E(H)|}{|H|} \le 101\log n$ with probability $O(1/n)$, which means that the output of $\mathcal{A}$ when $\mathcal{A}'$ presents this undirected random graph $G$ to it will be at most $101\log n$.

    On the other hand, suppose $x^1,x^2,\ldots,x^n$ are drawn from $\Dp$. Recall that $S, R$ are uniformly random subsets of $[n]$ drawn without replacement of size $k$. By Hoeffding's bound (e.g., Proposition 1.2 in \cite{bardenet2015concentration}), the size of $\{i \in S: i \ge n/2\}$ is at least $\frac{k}{3}$ with probability at least $1-e^{-\Omega(k)}$. Similarly, the size of $\{j \in R: j \le n/2\}$ is at least $\frac{k}{3}$ with probability at least $1-e^{-\Omega(k)}$. Together with a union bound, we get that the size of both these sets is at least $\frac{k}{3}$ with probability at least $1-e^{-\Omega(k)}$. But then note that conditioned on this event, in the undirected graph $G$ that $\mathcal{A}'$ presents to $\mathcal{A}$, at least $\frac{k}{3}$ vertices in $v^{n/2},\ldots,v^n$ are all connected to at least $\frac{k}{3}$ vertices in $v^{1},\ldots,v^{n/2}$. Note also that $2k/3 \le \beta$ by assumption, and hence the maximum edge density amongst at most $\beta$-sized subgraphs of $G$ is at least $\frac{k^2/9}{2k/3}=k/6$, %
    meaning that output of $\mathcal{A}$ will be at least $k/6\alpha = 200 \log n$.

    Therefore, $\mathcal{A}'$ can distinguish between $\Du$ and $\Dp$ with constant advantage by checking if the output of $\mathcal{A}$ is at most $101\log n$ or at least $200\log n$. It does so using only $\tilde{o}\left( \frac{n^2}{p\beta\alpha^4}\right)=o\left(\frac{n^2q}{pk^4\log n}\right)$ bits of memory, which gives us the desired contradiction.
    
\end{proof}

%% file: semi-random.tex
\section{Multi-pass Streaming Lower Bounds in the Semi-random Model}
\label{sec:semi-random-planted-biclique}

In this section, we will prove multi-pass memory lower bounds for detecting planted bi-cliques in random bipartite graphs under the presence of a monotone adversary. While the general outline will follow that of the previous section, we will require making subtle and crucial updates, which will allow us to prove a \textit{stronger} lower bound for this model.

Formally, we will study the following distinguishing problem. 

\begin{problem}[Semi-random Planted Bi-Clique] \label{prob:semi-random}
    Let $0<k_1,k_2 \le n$. Consider the following joint distributions on $n$-bit vectors $x^1,\ldots, x^n$: 
    \begin{enumerate}
    \item $\Du$: $\forall i\in[n]$, $x^i$ is drawn from the uniform distribution over $\{0,1\}^n$.
    \item $\Dp$: $S\subseteq[n]$ is drawn uniformly at random from all subsets of $[n]$ of size $k_2$. $R$ is drawn uniformly at random from all subsets of $[n]$ of size $k_1$. 
    
    $\forall i\not \in R$, $x^i$ is drawn from the uniform distribution over $\{0,1\}^n$. \\
    $\forall i\in R$, $\forall j\in S$ $x^i_j=1$, and $\forall j\not\in S$ $x^i_j$ is a uniform $\{0,1\}$ bit.
    \end{enumerate}
    Let $A$ be a matrix with rows as $x^1,\ldots, x^n$, we consider it as an adjacency matrix of a bipartite graph with $n$ left vertices and $n$ right vertices. A (computationally unbounded) monotone adversary is allowed to examine the rows of $A$, and if the matrix was drawn from $\Dp$ the adversary is allowed to delete any edges which did not belong to the planted bi-clique. More formally, for $i\notin R$, the adversary can set $x_j^i=0$ for any $j \in [n]$; for $i\in R$, the adversary can set $x_j^i=0$ for any $j \notin S$. Given (possibly modified) vectors $x^1,\ldots, x^n$, the goal is to distinguish if the vectors were originally drawn from $\Du$ or $\Dp$.
\end{problem}

To show hardness for Problem \ref{prob:semi-random}, we define a distinguishing problem which is similar to the Planted Bi-clique problem, but instead of planting the $\allone$ pattern, allows planting an arbitrary pattern on some of the rows of the data.

\begin{problem}[Pattern Planted Bi-Clique]
    \label{prob:fixed_pattern}
    Let $0<k \le n$. The goal is to distinguish between the following joint distributions on $n$-bit vectors $x^1,\ldots, x^n$: 
    \begin{enumerate}
    \item $\Du$: $\forall i\in[n]$, $x^i$ is drawn from uniform distribution over $\{0,1\}^n$.
    \item $\Dp$: $S\subseteq[n]$ is drawn uniformly at random from all subsets of $[n]$ of size $k$. Let $S=\{j_1,j_2,\ldots,j_{k}\}$. A vector $v$ is drawn uniformly at random from $\{0,1\}^{k}$. %
    $R$ is drawn uniformly at random from all subsets of $[n]$ of size $k$. 
    
    $\forall i\not \in R$, $x^i$ is drawn from uniform distribution over $\{0,1\}^n$. \\
    $\forall i\in R$, $\forall m\in [k]$ $x^i_{j_m}=v_{m}$, and $\forall j\not\in S$ $x^i_j$ is a uniform $\{0,1\}$ bit.
    \end{enumerate}
\end{problem}

We can show the following memory lower bound for \Cref{prob:fixed_pattern}: 

\begin{restatable}{theorem}{thmpatternplanted}[Memory Lower Bound for Pattern Planted Bi-Clique]
    \label{thm:pattern-planted}
    Let $0<k \le n$. Any $p$-pass streaming algorithm that solves \Cref{prob:fixed_pattern}, when $x^1,x^2,\ldots,x^n$ arrive in a stream, requires at least $\Omega\left(\frac{n^2}{pk^3}\right)$ bits of memory.
\end{restatable}

The proof of \Cref{thm:pattern-planted} (given in \Cref{sec:appendix-semi-random}) uses arguments similar to those in the proof of \Cref{thm:main}, by first decomposing the problem into a partitioned version, and showing hardness for the partitioned version. Crucially, because the planted pattern is a random pattern, this allows us to use smaller-sized partitions, and also simplifies the calculations involved in upper-bounding $\mu_1$ by $\mu_0$.

Using the hardness of \Cref{prob:fixed_pattern}, we can derive the following memory lower bound for detecting planted bi-cliques in the monotone adversary/semi-random model.

\begin{theorem}[Memory Lower Bound for Semi-random Planted Bi-Clique]
    \label{thm:semi-random}
    Consider any $0< k \le n$. %
    For any $p$-pass streaming algorithm that processes $x^1,x^2,\ldots,x^n$ arriving in a stream and only ever uses $o\left(\frac{n^2}{pk^3}\right)$ bits of memory, there exists some integer $k' \in \left[\frac{k}{3},\frac{2k}{3}\right]$ and instance of \Cref{prob:semi-random} with $k_1=k, k_2=k'$,  %
    for which the algorithm does not have advantage better than 0.9.

\end{theorem}
We note that the theorem above holds for any algorithm that knows $k$, but does not know the \textit{precise} value of $k'\in \left[\frac{k}{3}, \frac{2k}{3}\right]$.
\begin{proof}
    Let $\mathcal{A}$ be any $p$-pass streaming algorithm that processes $x^1,\ldots,x^n$ arriving in a stream, uses only $o\left(\frac{n^2}{pk^3}\right)$ bits of memory, and satisfies that: 
    \begin{enumerate}
        \item[(1)] If $x^1,\ldots,x^n \sim \Du$, then $\mathcal{A}$ outputs $\Du$ with probability at least 0.9.
        \item[(2)] For \textit{every} $k' \in \left[\frac{k}{2}-100\sqrt{k\log n}, \frac{k}{2}+100\sqrt{k\log n}\right]$: if $x^1,\ldots,x^n \sim \Dp$  (as in \Cref{prob:fixed_pattern}) \textit{conditioned} on $|v|=k'$, then $\mathcal{A}$ outputs $\Dp$ with probability at least 0.9.
    \end{enumerate}

    But notice that when $v$ is drawn uniformly at random from $\{0,1\}^{k}$, the probability that $|v| \in \left[\frac{k}{2}-100\sqrt{k\log n}, \frac{k}{2}+100\sqrt{k\log n}\right]$ is at least $1-n^{-10}$. Together with (2) above, we conclude that: if $x^1,\ldots,x^n \sim \Dp$  (as in \Cref{prob:fixed_pattern}), then $\mathcal{A}$ outputs $\Dp$ with probability at least 0.89. But this contradicts the lower bound from \Cref{thm:pattern-planted}. Thus, it must be the case that there exists $k' \in \left[\frac{k}{2}-100\sqrt{k\log n}, \frac{k}{2}+100\sqrt{k\log n}\right]$ such that $\mathcal{A}$ \textit{does not} distinguish between $\Du$ and $\Dp$ conditioned on $|v|=k'$, with advantage $0.9$.

    For such a $k'$, consider a problem instance of \Cref{prob:semi-random} with $k_1=k$, $k_2=k'$, and a monotone adversary, who upon seeing $x^1,\ldots,x^n \sim \Dp$ (as in \Cref{prob:semi-random}), corrupts the non-planted columns (i.e., $[n] \setminus S$) in the planted rows $R$ as follows: the adversary chooses a uniformly random subset $I \subseteq [n] \setminus S$ of size $k-k'$, and for every $i \in R$, the adversary sets $x^i_{I}=\zerob$. Then, observe that the final distribution of $x_1,\ldots,x_n$ after the corruption is exactly the distribution $\Dp$ (from \Cref{prob:fixed_pattern}), \textit{conditioned} on $|v|=k'$. From our reasoning in the above paragraph, for this particular monotone adversary, $\mathcal{A}$ cannot distinguish between $\Du$ and $\Dp$ (as in \Cref{prob:semi-random}) with advantage $0.9$. The theorem follows.
    \end{proof}

\subsection{Application: Maximum Bi-Clique}
\label{sec:application-maximum-biclique}

Theorem \ref{thm:semi-random} implies a multi-pass streaming lower bound for approximating the size of the largest bi-clique in a graph in the Vertex Arrival Model (\Cref{def:vertex-arrival-streaming-model}).

\begin{problem}[Maximum Bi-Clique] \label{prob:undirected-biclique}
    Consider an undirected graph $G$ on $n$ vertices (self-edges allowed). Let $1< k' \le n$. A $k'$-biclique in $G$ corresponds to subsets $S, R \subseteq [n]$, $|S|=|R|=k'$, such that for every $u \in S, v \in R$, $u$ and $v$ are connected by an edge in $G$. The goal is to approximate the size of the largest biclique in $G$, i.e., $k = \max\{k': \exists k'\text{-biclique in $G$}\}$. For $\alpha \ge 1$, an $\alpha$-approximation to $k$ is any number $x$ such that $(1/\alpha)k \le x \le k$.
\end{problem}

\begin{corollary}[Memory Lower Bound for Maximum Bi-Clique]\label{thm:biclique-approx}
    Consider any $1< \alpha \le n$. Any $p$-pass streaming algorithm that approximates the size of the largest bi-clique in an undirected graph that is presented in the Vertex Arrival Model to a factor $\alpha$ requires at least $\tilde{\Omega}\left(\frac{n^2}{p\alpha^3}\right)$ bits of memory.
\end{corollary}

\begin{proof}

    Let $\mathcal{A}$ be a $p$-pass streaming algorithm that uses only $\tilde{o}\left(\frac{n^2}{p\alpha^3}\right)$ bits of memory and always approximates the size of the largest biclique in a graph presented in the worst-case, vertex-arrival model to a factor $\alpha$. We will set $k=40\alpha \log n$. Using $\mathcal{A}$, we will construct a $p$-pass streaming algorithm $\mathcal{A}'$ that processes $x^1,\ldots,x^n$ arriving in a stream, which solves every instance of \Cref{prob:semi-random} for which $k_1=k, k_2 \in \left[\frac{k}{3}, \frac{2k}{3}\right]$, while using only $\tilde{o}\left(\frac{n^2}{p\alpha^3}\right)=o\left(\frac{n^2}{pk^3}\right)$ bits of memory. This would contradict \Cref{thm:semi-random}, and give us the claimed result.

    The algorithm $\mathcal{A}'$ operates as follows. Given an input stream $x^1,x^2,\ldots,x^n$, $\mathcal{A}'$ interprets each $x^i$ in the stream as a vertex $v^i$ in an undirected graph $G$, and presents it to $\mathcal{A}$ in the vertex-arrival model. For each $v^i$, it will read off connectivity to $v^1,\dots,v^i$ from $x^i_{[1:i]}$. That is, for $j\le i$, there is an undirected edge between $v^j$ and $v^i$ iff $x^i_{j}=1$. Note that $\mathcal{A}'$ can simulate this input space-efficiently for $\mathcal{A}$ (it only needs to keep track of a counter). %

    Now, suppose $x^1,x^2,\ldots,x^n$ are drawn from $\Du$. Then, observe that the graph $G$ that $\mathcal{A}'$ presents to $\mathcal{A}$ is a random graph, where every $(v^i, v^j)$ is connected by an edge with probability $1/2$. The size of the maximum biclique in such a graph is at most $3\log n$ with probability $1-o(1)$ \cite{trevisannotes}, and hence, the output of $\mathcal{A}$ will be at most $3\log n$.

    On the other hand, suppose $x^1,x^2,\ldots,x^n$ are drawn from $\Dp$ in the semi-random model for $k_1=k$, and any $k_2 \in \left[\frac{k}{3}, \frac{2k}{3}\right]$. Recall that $S, R$ are uniformly random subsets of $[n]$ drawn without replacement of size $k_2, k_1$ respectively. By Hoeffding's bound (e.g., Proposition 1.2 in \cite{bardenet2015concentration}), the size of $\{i \in S: i \ge n/2\}$ is at least $\frac{k_2}{3}\ge \frac{k}{10}$ with probability at least $1-e^{-\Omega(k_2)}=1-e^{-\Omega(k)}$. Similarly, the size of $\{j \in R: j \le n/2\}$ is at least $\frac{k_1}{3}=\frac{k}{3}$ with probability at least $1-e^{-\Omega(k_1)}=1-e^{-\Omega(k)}$. Together with a union bound, we get that the size of both these sets is at least $\frac{k}{10}$ with probability at least $1-e^{-\Omega(k)}$. But then note that conditioned on this event, in the undirected graph $G$ that $\mathcal{A}'$ presents to $\mathcal{A}$, at least $\frac{k}{10}$ vertices in $v^{n/2},\ldots,v^n$ are all connected to at least $\frac{k}{10}$ vertices in $v^{1},\ldots,v^{n/2}$, meaning that the size of the largest biclique in $G$ is at least $\frac{k}{10}$. %
    Hence, the output of $\mathcal{A}$ will be at least $\frac{k}{10\alpha}=4\log n$.

    Therefore, $\mathcal{A}'$ can distinguish between $\Du$ and $\Dp$ with constant advantage by checking if the output of $\mathcal{A}$ is at most $3\log n$ or at least $4\log n$. It does so using only $\tilde{o}\left(\frac{n^2}{p\alpha^3}\right)=o\left(\frac{n^2}{pk^3}\right)$ bits of memory, which gives us the desired contradiction.

\end{proof}

%% file: gaussian_detection.tex
\section{Memory-Sample Tradeoffs for Distinguishing Sparse Gaussians}
\label{sec:gaussians}

In this section, we prove our result for the sparse Gaussian distinguishing problem. We begin by stating the formal definition of the problem.
\begin{problem} \label{prob:gaussian_mean_general}
Let $0<\sparsity \le d$ and $\alpha \in (0,1]$. The goal is to distinguish between the following joint distributions on $d$-dimensional vectors $x^1,\ldots, x^n$: 
\begin{enumerate}
\item $\Dn$: $\forall i\in[n]$, $x^i$ is drawn from the standard Gaussian distribution $N(0,I_d)$.
\item $\Dp$: A vector $v\in \mathbb{R}^d$ is drawn as follows: first, choose $S \subseteq [d], |S|=\sparsity$  uniformly at random. For every $i \in S$, set $v_i=\alpha$, and for every $i \in [d] \setminus S$, set $v_i=0$.
Then $\forall \; i \in [n]$, 
\begin{align}
x^i \sim
\begin{cases}
    N(v, I_d), & \text{ w.p. } q,\\
    N(0,I_d), &\text{ w.p. } 1-q.
\end{cases}
\end{align}
\end{enumerate}
\end{problem}

A primary qualitative difference in the problem above compared to the distinguishing problems stated previously is that earlier, the planted distribution had exactly a \textit{fixed} number $k$ of vectors from amongst $x^1,\dots,x^n$ that were drawn from the planted distribution; in contrast, in the problem above, each $x^i$ independently has a probability $q$ of being drawn from the planted distribution.
 
We show the following memory-sample tradeoff for \Cref{prob:gaussian_mean_general}.

\begin{theorem}\label{thm:gaussian_general}
Let $\epsilon \in (0,0.01)$ be a constant, $d$ be sufficiently large, $\sparsity \leq d$, $n \leq d^{10}$ and $\alpha \in \left (\frac{1}{\sparsity \sqrt{  \log d}}, 1 \right ]$. Any $s$-bit, $p$-pass algorithm (using public as well as private randomness) that solves \Cref{prob:gaussian_mean_general} for every $\sparsity' \in [2\sparsity/3, 4\sparsity/3]$ satisfies that $s\cdot n\ge  \tilde{\Omega}\left(\frac{d^{1-\epsilon}}{p(\alpha \sparsity)^2q^2}\right)$.
\end{theorem}

Again, the theorem above holds for any algorithm that knows $\ell$, but does not know the precise value of $\ell' \in [2\ell/3, 4\ell/3]$.

\begin{proof}

As in the other proofs, we consider a partition version of the problem, where the planted coordinates in the vector $v$ are confined to being within a partition. Furthermore, while \Cref{prob:gaussian_mean_general} has the property that $nq$ vectors, \textit{in expectation}, have a plant corresponding to a vector $v$ of \textit{fixed sparsity} $\sparsity$, this property is flipped in the partition version; namely, it will be the case that a \textit{fixed number} $k=nq$ of the vectors have a plant corresponding to a vector $v$ of \textit{expected sparsity} $\sparsity$.  %

Concretely, consider the following distribution $D$ over vectors in $\mathbb{R}^t$, for $t\ge (\alpha \sparsity)^2 d^\epsilon \log^2(200 nd)$. Independently, for every co-ordinate $v_j$, $j\in [t]$,
\begin{align}
\label{eqn:v-distribution}
v_j=
\begin{cases}
    \alpha, & \text{ w.p. } \sparsity/t,\\
    0,& \text{ w.p. } 1-\sparsity/t.
\end{cases}
\end{align}
We can see that originally in \Cref{prob:gaussian_mean_general}, there were \textit{exactly} $\sparsity$ coordinates where $v$ was non-zero, whereas $v \sim D$ above has $\sparsity$ coordinates that are non-zero \textit{in expectation}, and all these coordinates are contained within the same partition (of size $t$). Consider now the following problem, for which we will show hardness.

\begin{problem} \label{prob:partition_mean_general}
Let $t\ge \max \{ (\alpha \sparsity)^2 d^\epsilon \log^2(200 nd), 2 \sparsity \}$ and suppose that $t$ divides $d$. Let $\calT=\{T_{r}\}_{r\in[d/t]}$ be a partition of $[d]$, where $\forall r, |T_r|=t$. The goal is to distinguish between the following joint distributions on $d$-bit vectors $x^1,\ldots, x^n$:

\begin{enumerate}
\item $\no$ (no instance): $\forall i\in[n]$ and $\forall r\in[d/t]$, $x^i_{T_r}$ is drawn from $N(0,I_t)$. 

\item $\yes^\calT$ (yes instance): Draw $r$ uniformly from $[d/t]$. $\forall i\in[n]$ and $\forall r'\neq r$, $x^i_{T_{r'}}$ is drawn from $N(0,I_t)$. 

Draw $v \sim D$. Draw a uniformly random subset $R \subseteq [n]$ of size $k$. 

$\forall i\not \in R$, $x^i_{T_r}$ is drawn from $N(0, I_t)$, whereas, $\forall i\in R$, $x^i_{T_r}$ is drawn from $N(v,I_t)$.
\end{enumerate}
\end{problem}

In the following lemma, we show that proving hardness for \Cref{prob:partition_mean_general}  is enough to prove the theorem. The proof of this lemma is a sequence of reductions that uses arguments similar to those used earlier in the paper, and is deferred to \Cref{sec:appendix-gaussians}.

\begin{restatable}{lemma}{lemmachirag}\label{lem:chirag}
Let $\epsilon \in (0,0.01)$ be a constant, $d$ be sufficiently large, $\sparsity \leq d$, $n \leq d^{10}$ and $\alpha \in \left (\frac{1}{\sparsity\sqrt{\log d}}, 1 \right ]$.
    Let $\mathcal{A}$ be a $p$-pass streaming algorithm that uses $s$ bits of memory and $n/400$ samples, and solves \Cref{prob:gaussian_mean_general} with probability 0.99 for every value of $\sparsity' \in [2\sparsity/3, 4\sparsity/3]$. Then, there exists a $p$-pass streaming algorithm $\mathcal{A}'$ that uses $s+\tilde{O}(1)$ bits of memory and $n$ samples, and solves \Cref{prob:partition_mean_general} for $k=nq$ with probability 0.97
\end{restatable}

We note that the assumption of $t$ dividing $d$ in \Cref{prob:partition_mean_general} is for convenience; we can handle the technicality of $t$ not dividing $d$ similarly to how we did in the proof of \Cref{thm:main}; this is fleshed out in more detail in the proof of the lemma.

Using \Cref{lem:chirag}, it suffices to show hardness for \Cref{prob:partition_mean_general}. For this, however, we will need to define a truncation, on both the $x^i$s and $v$. There are some steps where we will not be able to get a good bound for all vectors $v$. So, we define a set $V_{good}=\{v: \|v\|_0 \le 100\sparsity\}$ and let $D_{good}$ be the distribution $D$ that is restricted to vectors in the set $V_{good}$. We will also need to truncate the distribution over $x^i$s to get our bound. We define the set 
\begin{align}
    \label{eqn:gaussian-truncation-good-set}
    T = \left  \{ x \in \mathbb{R}^t : \sum_{j=1}^t e^{\alpha x_j} \le te^{\alpha^2/2}+  (C_1 \alpha) \sqrt{t}  d^{\epsilon/ 2} \log(200 nd)  \right \}
\end{align}
for a constant $C_1$ to be later determined. Let $\pt^0$ be the restriction of the ($t$-dimensional) Gaussian distributions $N(0, I_t)$ to this set $T$. For a vector $v \in \mathbb R^t$, we let $\pt^{1, v}$ denote the restriction of the Gaussian distribution $N(v, I_t)$  to the set $T$.

We now further define a \textit{truncated} version of \Cref{prob:partition_mean_general}.

\begin{problem} \label{prob:partition_mean_trunc}
Let $t\ge \max \{ (\alpha \sparsity)^2 d^\epsilon \log^2(200 nd), 2 \sparsity \}$ and suppose that $t$ divides $d$.
Let $\calT=\{T_{r}\}_{r\in[d/t]}$ be a partition of $[d]$, where $\forall r, |T_r|=t$. The goal is to distinguish between the following joint distributions on $d$-bit vectors $x^1,\ldots, x^n$:
\begin{enumerate}
\item $\no$: $\forall i\in[n]$ and $\forall r\in[d/t]$, $x^i_{T_r}$ is drawn from $\pt^0$. 

\item $\yes^\calT$: Draw $r$ uniformly from $[d/t]$. $\forall i\in[n]$ and $\forall r'\neq r$, $x^i_{T_{r'}}$ is drawn from $\pt^0$. 

Draw $v\sim D_{good}$. Draw a uniformly random subset $R \subseteq [n]$ of size $k$. 

$\forall i\not \in R$, $x^i_{T_r}$ is drawn from $\pt^0$, whereas, $\forall i\in R$, $x^i_{T_r}$ is drawn from $\pt^{1,v}$.

\end{enumerate}
\end{problem}

The following lemma, proved in \Cref{sec:appendix-gaussians}, shows that the truncated distributions are close to the original ones.

\begin{restatable}{lemma}{closenessgaussian} \label{lem:closeness_general}
Let $v \in V_{good}$ be arbitrary. The distributions $N(0, I_t)$ and $N(v,I_t)$ are close (in TV distance) to their respective truncations  $\pt^0$ and $\pt^{1,v}$:
\begin{align*}
 \left \| \pt^0 - N(0, I_t) \right\|_{TV} &\leq 0.01 / (nd/t), \\
 \left \| \pt^{1, v} - N(v,I_t) \right\|_{TV} &\leq 0.01 / (nd/t).
\end{align*}
Also, $\Pr_{v \sim D} [ v \in V_{good}] \geq 0.99$.
\end{restatable}

By \Cref{lem:closeness_general} and the triangle inequality applied to the partitions of $t$ coordinates, the TV distance between $\no$ defined in \Cref{prob:partition_mean_general} and in \Cref{prob:partition_mean_trunc} is at most $0.01$. Similarly, the TV distance between $\yes^\calT$ in these problems, 
taking additionally into account that $\Pr_{v \sim D} [ v \notin V_{good}] \leq 0.01$, is at most $0.02$; note that in both cases, the TV distance for a fixed setting of $r, R, v$ is at most $0.02$, and the random processes by which $r, R, v$ are selected in both problems are identical.

Since the respective distributions in \Cref{prob:partition_mean_general} and \Cref{prob:partition_mean_trunc} are close, it follows that if some algorithm can solve \Cref{prob:partition_mean_general} with advantage $0.99$, then the algorithm can also solve \Cref{prob:partition_mean_trunc} with advantage $0.97$. Therefore, we will now show a lower bound for \Cref{prob:partition_mean_trunc}, which as we can observe, conveniently fits the template of \Cref{prob:general}. With a view to invoke \Cref{lem:framework}, the next claim, whose proof is a calculation and is also given in \Cref{sec:appendix-gaussians}, bounds the ratio $\mathbb E_{v \sim D_{good}} [\pt^{1, v}]/\pt^{0}$.

\begin{restatable}{claim}{gaussiandensity}
    \label{claim:gaussian-density-technical-condition}
    Let $\mu_0$ and $\mu_1^v$ be the probability density functions of $\pt^0$ and $\pt^{1, v}$ respectively. Then, there exists a positive constant $C$ such that
    \begin{align*}
        \mathbb E_{v \sim D_{good}} [\mu_1^v ] \le C \mu_0.
    \end{align*}
\end{restatable}

We have thus shown that \Cref{prob:partition_mean_trunc} fits the generic description of \Cref{prob:general}, and also satisfies the requirement of \Cref{lem:framework}. The statement of the theorem then implies that any $s$-bit, $p$-pass streaming algorithm which solves the problem satisfies $s=\Omega\left(\frac{nd}{pk^2 t}\right)$. Taking $k=nq$ and $t = (\alpha \sparsity)^2 d^\epsilon \log^2(200 nd)$  gives us that $s\cdot n\ge  \tilde{\Omega}\left(\frac{d^{1-\epsilon}}{p(\alpha \sparsity)^2q^2}\right)$. Invoking \Cref{lem:closeness_general} and \Cref{lem:chirag} completes the proof of the theorem.

\end{proof}

We now discuss how our memory-sample tradeoff established in \Cref{thm:gaussian_general} relates to those of algorithms that solve Problem \ref{prob:gaussian_mean_general}.
First, consider a statistical test that computes the sum of all the coordinate values across the $n$ samples it receives (i.e., the sum $\sum_{i = 1}^n\sum_{j = 1}^d x^{i}_j$). When samples are drawn from $\Dn$, the sum is a Gaussian with mean $0$ and variance $nd$. On the other hand, when samples are drawn from $\Dp$, the sum is a Gaussian with mean $n q \sparsity \alpha$ and variance $nd$. Now suppose that the test declares $\Dn$ if and only if the sum is at most $n q \sparsity \alpha / 2$. By a Gaussian tail bound, the test's failure probability is at most $\delta$ if $n = O(d \log(1/\delta)/ (q\sparsity \alpha)^2)$. 
For a constant success probability, this test would require $O(d / (\alpha \sparsity q)^2)$ samples. Note also that the test can be computed with $O(\log d)$ bits of precision.\footnote{The accumulation error for the sum computation is at most $nd \cdot 2^{-\rho}$, where $\rho$ is the precision of bits for each coordinate's floating points. Since we require the accumulation error to not exceed each distribution's standard deviation of $\sqrt{nd}$, we can take $\rho = O(\log n + \log d)$ bits. Furthermore, Chebyshev's inequality implies that the computed sum is at most $2nd$ with high probability, so we can also take $O (\log n + \log d)$ bits for the integral component of the sum.  Since we assume $n \leq d^{10}$, we therefore need $O(\log d)$ bits altogether.}
Therefore, our bound given in Theorem \ref{thm:gaussian_general} is nearly optimal for algorithms in the $O(\log d)$ memory regime.

We remark that Problem \ref{prob:gaussian_mean_general} can also be solved with improved sample complexity via the following procedure. The algorithm fixes a set $R$ of 
randomly chosen coordinates and for each sample $x^j$ it receives, it stores all the entries at those coordinates (that is, the value $ x^{j}_R $).
For each subset $S_1 \subseteq [n] $ of the received samples with $|S_1| = s_1$ and each subset of the coordinates $S_2 \subseteq R$ with $|S_2| = s_2$, the algorithm computes the statistic 
$Y_{S_1, S_2} = \sum_{j \in S_1} \sum_{i \in S_2} x^{j}_i$. The algorithm declares $\Dn$ if no statistic $Y_{S_1,S_2}$ exceeds some fixed threshold; otherwise, it declares $\Dp$.
In the following claim, we show that there is a regime in which this test is able to distinguish between the distributions $\Dn$ and $\Dp$ using $\tilde{O}(1/(q \alpha^2))$ samples and $\tilde{O}\left(\frac{d}{\sparsity \alpha^2q}\right)$ memory.

\begin{restatable}{claim}
{clmgaussianub}\label{clm:gaussian-ub}
Fix a constant $\delta \in (0, 1)$ and let $C_{\delta, \alpha} = \left ( \frac{8 + 4 \log(4/ \delta)}{\alpha^2} \right )$. 
For all $n, d$ sufficiently large that satisfy 
$nq \geq 2 C_{\delta, \alpha} \log (nd)$, the following holds. If $|R| = 2 C_{\delta, \alpha} (d / \sparsity) \log(nd/ \delta) \log (nd)$, $\sparsity \geq s_1 = s_2 = C_{\delta, \alpha} \log (nd)$ and $\tau = \sqrt{2s_1 s_2 \log \left (2 \binom{n}{s_1} \binom{|R|}{s_2} / \delta \right) }$, then 
\begin{align*}
    \max \left \{
    \Pr_{\Dn} \left [ \max_{\substack{S_1 \subseteq [n], |S_1| = s_1 \\ S_2 \subseteq R, |S_2| = s_2}}  \sum_{j \in S_1} \sum_{i \in S_2} x^{j}_i \geq \tau  \right ] ,
        \Pr_{\Dp} \left [ \max_{\substack{S_1 \subseteq [n], |S_1| = s_1 \\ S_2 \subseteq R, |S_2| = s_2}}  \sum_{j \in S_1} \sum_{i \in S_2} x^{j}_i \leq \tau  \right ] ,
    \right \}
    \leq \delta. 
\end{align*}

\end{restatable}

%% file: sparse_pca.tex
\section{Memory-Sample Tradeoffs for Sparse PCA Detection}\label{sec:pca}
In this section, we prove our result for the sparse PCA detection problem. We begin by stating the formal definition of the problem.

\begin{problem}\label{prob:pca_general}
Let $\sparsity \leq d$ be some integers and let $\alpha > 0$ be some parameter. The goal is to distinguish between the following joint distributions on $d$-dimensional vectors $x^1,\ldots, x^n$: 
\begin{enumerate}
\item $\Dn$: $\forall i\in[n]$, $x^i$ is drawn from $N(0,I_d)$.
\item $\Dp$: 
Draw a uniformly random subset of $\sparsity$ indices, $S \subseteq [d]$ and let $v = \frac{1}{\sqrt{\sparsity}} 1_S$.

$\forall i\in[n]$, $x^i$ is drawn from $N(0, \Sigma_S)$, where $\Sigma_S = I_d + \alpha v v^\intercal$.
\end{enumerate}
\end{problem}

We show the following memory-sample tradeoff for \Cref{prob:pca_general}.

\begin{theorem}\label{thm:pca_general}
Let $\epsilon \in (0,0.01)$ be a constant, $\alpha \in \left (0, \frac{\epsilon}{22} \right)$ be a constant, $d$ be sufficiently large, $\sparsity \leq d$, and $n \leq d^{10}$. Then, any $s$-bit, $p$-pass algorithm (using public as well as private randomness) that solves \Cref{prob:pca_general}
 satisfies $s\cdot n\ge  \tilde{\Omega}\left(\frac{d^{1-\epsilon}}{p\sparsity}\right)$.

\end{theorem}

\begin{proof}[Proof]
Our approach is to show a lower bound for the following simpler distinguishing problem where the planted distribution is a more structured mixture of $\sparsity$-sizes subsets of indices. A lower bound for this problem implies a lower bound for the more general Problem \ref{prob:pca_general}.
\begin{problem} \label{prob:pca_block-general}
Let $\sparsity \leq d$ be some integers and let $\alpha > 0$ be some parameter. The goal is to distinguish between the following joint distributions on $d$-dimensional vectors $x^1,\ldots, x^n$: 
\begin{enumerate}
\item $\Dn$: $\forall i\in[n]$, $x^i$ is drawn from $N(0,I_d)$.
\item $\Dp$: 
Draw $S$ uniformly from $\{ [1, \sparsity], [\sparsity + 1, 2 \sparsity], \dots, [d - \sparsity + 1 , d]  \} $ and let $v = \frac{1}{\sqrt{\sparsity}} 1_S$.

$\forall i\in[n]$, $x^i$ is drawn from $N(0, \Sigma_S)$, where $\Sigma_S = I_d + \alpha v v^\intercal$.
\end{enumerate}
\end{problem}

As in the other proofs, we consider a partition version of the problem, where the planted coordinates in the vector $v$ are confined to being within a partition. Consider now the following problem for which we will show hardness.

\begin{problem} \label{prob:partition_block-pca_general}
Let $t \geq \sparsity d^{\epsilon} \log (400 nd)$ and suppose that $\sparsity$ divides $t$ and that $t$ divides $d$.
Let $\calT=\{T_{r}\}_{r\in[d/t]}$ be a partition of $[d]$, where $\forall r, |T_r|=t$. The goal is to distinguish between the following joint distributions on $d$-bit vectors $x^1,\ldots, x^n$:

\begin{enumerate}
\item $\no$ (no instance): $\forall i\in[n]$ and $\forall r\in[d/t]$, $x^i_{T_r}$ is drawn from $N(0,I_t)$. 

\item $\yes^\calT$ (yes instance): Draw $r$ uniformly from $[d/t]$. $\forall i\in[n]$ and $\forall r'\neq r$, $x^i_{T_{r'}}$ is drawn from $N(0,I_t)$. 

Draw $S$ uniformly from $\mathcal S = \{ [1, \sparsity], [\sparsity + 1, 2 \sparsity], \dots, [t - \sparsity + 1 , t]  \}$ and let $v = \frac{1}{\sqrt{\sparsity}} 1_S$.

$\forall i \in [n]$, $x^i_{T_r}$ is drawn from $N(0,\Sigma_S)$, where $\Sigma_S = I_t + \alpha v v^\intercal$.
\end{enumerate}
\end{problem}
We will  need to truncate the distribution over $x^i$s to get our bound. We define the set 
\begin{align}
    \label{eqn:pca-truncation-good-set}
    T = \left  \{ x \in \mathbb R^t : 
    \sum_{R \in \mathcal S}  \exp \left (  \frac{\alpha}{2(\alpha + 1)} \cdot \frac{1}{\sparsity} (x^\intercal 1_R)^2 \right )  \leq 
    (t / \sparsity) (1- \alpha)^{-1/2}  +
    \delta \right \},
\end{align}
where $\delta = Cd^{\epsilon/2}\sqrt{(t / \sparsity) \log(400 nd )}$.
Let $\pt^0$ be the restriction of the Gaussian distribution $N(0, I_t)$ to this set $T$. For each set $S \in \mathcal S$, we let
$\pt^{1, S}$ denote the restriction of the Gaussian distribution $N(0, \Sigma_S)$ to the set $T$, where $\Sigma_S = I_t + \alpha v v^\intercal$ and  $v = \frac{1}{\sqrt{\sparsity}} 1_S$.

We now further define a \textit{truncated} version of \Cref{prob:partition_block-pca_general}.

\begin{problem} \label{prob:partition_pca_trunc}
Let $t \geq \sparsity d^{\epsilon} \log (400 nd)$ and suppose that $\sparsity$ divides $t$ and that $t$ divides $d$.
Let $\calT=\{T_{r}\}_{r\in[d/t]}$ be a partition of $[d]$, where $\forall r, |T_r|=t$. The goal is to distinguish between the following joint distributions on $d$-bit vectors $x^1,\ldots, x^n$:
\begin{enumerate}
\item $\no$: $\forall i\in[n]$ and $\forall r\in[d/t]$, $x^i_{T_r}$ is drawn from $\pt^0$. 

\item $\yes^\calT$: Draw $r$ uniformly from $[d/t]$. $\forall i\in[n]$ and $\forall r'\neq r$, $x^i_{T_{r'}}$ is drawn from $\pt^0$. 

Draw $S$ uniformly from $\mathcal S = \{ [1, \sparsity], [\sparsity + 1, 2 \sparsity], \dots, [t - \sparsity + 1 , t]  \}$.

$\forall i \in [n]$, $x^i_{T_r}$ is drawn from $\pt^{1,S}$.

\end{enumerate}
\end{problem}
The following lemma, proved in \Cref{sec:appendix-pca}, shows that the truncated distributions are close to the original ones.

\begin{restatable}{lemma}{closenesspca}
\label{lem:closeness_pca}
For any set $S \in \mathcal S$,
the distributions $N(0, I_t)$ and $N(0,\Sigma_S)$ are close (in TV distance) to their respective truncations  $\pt^0$ and $\pt^{1,S}$:
\begin{align*}
 \left \| \pt^0 - N(0, I_t) \right\|_{TV} &\leq 0.01 / (nd/t), \\
 \left \| \pt^{1, S} - N(0, \Sigma_S) \right\|_{TV} &\leq 0.01 / (nd/t).
\end{align*}
\end{restatable}

By \Cref{lem:closeness_pca} and the triangle inequality applied to the partitions of $t$ coordinates, the TV distance between $\no$ defined in \Cref{prob:partition_block-pca_general} and in \Cref{prob:partition_pca_trunc} is at most $0.01$. Similarly, the TV distance between $\yes^\calT$ in these problems is at most $0.01$.

Since the respective distributions in \Cref{prob:partition_block-pca_general} and \Cref{prob:partition_pca_trunc} are close, it follows that if some algorithm can solve \Cref{prob:partition_block-pca_general} with advantage $0.99$, then the algorithm can also solve \Cref{prob:partition_pca_trunc} with advantage $0.98$. Therefore, we will now show a lower bound for \Cref{prob:partition_pca_trunc}, which as we can observe, conveniently fits the template of \Cref{prob:general}. With a view to invoke \Cref{lem:framework}, the next claim, whose proof is a calculation and is also given in \Cref{sec:appendix-pca}, bounds the ratio $\mathbb E_{S} [\pt^{1, S}]/\pt^{0}$.

\begin{restatable}{claim}{pcadensity}\label{claim:pca-density-technical-condition}
Let $\mu_0$ and $\mu_1^S$ be the probability density functions of $\pt^0$ and $\pt^{1, S}$ respectively. Then,
there exists a positive constant $C$ such that
\begin{align*}
    \mathbb E_{S \sim \mathcal S} [\mu_1^S ] \le C \mu_0.
\end{align*}
\end{restatable}
We have thus shown that \Cref{prob:partition_pca_trunc} fits the generic description of \Cref{prob:general}, and also satisfies the requirement of \Cref{lem:framework}. The statement of \Cref{lem:framework} then implies that any $s$-bit, $p$-pass streaming algorithm which solves the problem satisfies $s=\Omega\left(\frac{nd}{pk^2 t}\right)$. Taking $k=n$ and $t = \sparsity d^{\epsilon} \log (400 nd)$  gives us that $s\cdot n\ge  \tilde{\Omega}\left(\frac{d^{1-\epsilon}}{p\sparsity}\right)$. Invoking \Cref{lem:closeness_pca} completes the proof of the theorem.
\end{proof}

We now discuss how our memory-sample tradeoff established in Theorem \ref{thm:pca_general} relates to those of algorithms that solve Problem \ref{prob:pca_block-general}. Consider the statistical test that squares the sum of coordinates within each block and computes the cumulative sum across samples (i.e, the sum $\sum_{j=1}^n \sum_{R \in \mathcal S} (\sum_{i \in R} x_i^j ) ^2$). If the sum exceeds the threshold $ \tau = nd + n\alpha \sparsity/2$, then the test declares $\Dp$; otherwise it declares $\Dn$.
In the following claim, whose proof is a calculation given in Appendix \ref{sec:appendix-pca}, we show that the test is able to distinguish between $\Dn$ and $\Dp$ with a constant failure probability using $O( d/  \sparsity)$ samples.

\begin{restatable}{claim}{pcaupperbound}\label{claim:ub-sparse-pca}
Fix a constant $\delta \in (0, 1)$
and suppose that  $n \geq \log \left ( \frac{2}{\delta} \right ) \left [ \frac{4 C_1^2 (1 + \alpha)^2}{c \alpha^2 } \cdot \frac{d}{\sparsity} \right ] $. Then,
\begin{align*}
    \max \left \{
    \Pr_{\Dn} \left [ \sum_{j = 1}^n  \sum_{R \in \mathcal S} \left (\sum_{i \in R} x_i^j \right ) ^2 \geq  \tau \right ],
     \Pr_{ \Dp} \left [  \sum_{j = 1}^n \sum_{R \in \mathcal S} \left (\sum_{i \in R} x_i^j \right ) ^2 \leq \tau \right ]
     \right \} 
    \leq \delta 
\end{align*}
\end{restatable}
It is straightforward to verify that the statistical test can be computed with $O(\log  d)$ bits of precision (and the justification mirrors that given for the Gaussian mean distinguishing in Section~\ref{sec:gaussians}). Therefore, our bound given in Theorem \ref{thm:pca_general} is nearly optimal for algorithms in the $O(\log d)$ memory regime that solve \Cref{prob:pca_block-general}.

%% file: appendix-prelims.tex
\section{Proofs from \Cref{sec:prelims}}
\label{sec:appendix-prelims}

We first restate and prove \Cref{lemma:mic-independence}.
\lemmamicindependence*

\begin{proof}
    We will prove this lemma by induction. For the base case, consider \eqref{eqn:mic-independence-1} for $l=1$. In this case, the memory states conditioned on are simply the initial memory state $\sfM_0$. So, we have that
    \begin{align*}
        I(X^{[i, j-1]}, R_{([p], [i, j-1])}\;;\; X^{[1,i-1]}, R_{([p], [1,i-1])}, X^{[j,n]}, R_{([p], [j,n])} \mid \sfM_{0}, P) = 0,
    \end{align*}
    precisely because the inputs, the public randomness, the private randomness at every time step, and the initial memory state, are all independent of each other. Now, assume as the induction hypothesis that \eqref{eqn:mic-independence-1} holds for some $l$. We will first show that, if $l \le p$, then \eqref{eqn:mic-independence-2} holds. For this, observe that
    \begin{align*}
        &I(X^{[i, j-1]}, R_{([p], [i, j-1])}\;;\; X^{[1,i-1]}, R_{([p], [1,i-1])}, X^{[j,n]}, R_{([p], [j,n])} \mid \sfM_{\le l, i-1}, \sfM_{<l, j-1}, P) \\
        &\le I(X^{[i, j-1]}, R_{([p], [i, j-1])}\;;\; X^{[1,i-1]}, R_{([p], [1,i-1])}, X^{[j,n]}, R_{([p], [j,n])},  \sfM_{l, i-1} \mid \sfM_{<l, i-1}, \sfM_{<l, j-1}, P) \tag{chain rule and non-negativity of mutual information} \\
        &= I(X^{[i, j-1]}, R_{([p], [i, j-1])}\;;\; X^{[1,i-1]}, R_{([p], [1,i-1])}, X^{[j,n]}, R_{([p], [j,n])} \mid \sfM_{<l, i-1}, \sfM_{<l, j-1}, P) \\
        &\qquad+ I(X^{[i, j-1]}, R_{([p], [i, j-1])}\;;\; \sfM_{l, i-1} \mid \sfM_{<l, i-1}, \sfM_{<l, j-1}, X^{[1,i-1]}, R_{([p], [1,i-1])}, X^{[j,n]}, R_{([p], [j,n])}, P) \tag{chain rule} \\
        &= I(X^{[i, j-1]}, R_{([p], [i, j-1])}\;;\; \sfM_{l, i-1} \mid \sfM_{<l, i-1}, \sfM_{<l, j-1}, X^{[1,i-1]}, R_{([p], [1,i-1])}, X^{[j,n]}, R_{([p], [j,n])}, P) \tag{induction hypothesis} \\
        &= 0 \tag{$\sfM_{l, i-1}$ is determined by $\sfM_{<l, j-1}, X^{[j,n]}, R_{([p], [j,n])}, X^{[1,i-1]}, R_{([p], [1,i-1])}, P$}.
    \end{align*}
    Next, we will show that if $l \le p-1$, then \eqref{eqn:mic-independence-1} holds for $l+1$. Namely, observe that
    \begin{align*}
        &I(X^{[i, j-1]}, R_{([p], [i, j-1])}\;;\; X^{[1,i-1]}, R_{([p], [1,i-1])}, X^{[j,n]}, R_{([p], [j,n])} \mid \sfM_{\le l, i-1}, \sfM_{\le l, j-1}, P) \\
        &\le I(X^{[i, j-1]}, R_{([p], [i, j-1])}\;;\; X^{[1,i-1]}, R_{([p], [1,i-1])}, X^{[j,n]}, R_{([p], [j,n])},  \sfM_{l, j-1} \mid \sfM_{\le l, i-1}, \sfM_{<l, j-1}, P) \tag{chain rule and non-negativity of mutual information} \\
        &= I(X^{[i, j-1]}, R_{([p], [i, j-1])}\;;\; X^{[1,i-1]}, R_{([p], [1,i-1])}, X^{[j,n]}, R_{([p], [j,n])} \mid \sfM_{\le l, i-1}, \sfM_{<l, j-1}, P) \\
        &\qquad+ I(X^{[i, j-1]}, R_{([p], [i, j-1])}\;;\; \sfM_{l, j-1} \mid \sfM_{\le l, i-1}, \sfM_{<l, j-1}, X^{[1,i-1]}, R_{([p], [1,i-1])}, X^{[j,n]}, R_{([p], [j,n])}, P) \tag{chain rule} \\
        &= I(X^{[i, j-1]}, R_{([p], [i, j-1])}\;;\; \sfM_{l, j-1} \mid \sfM_{\le l, i-1}, \sfM_{<l, j-1}, X^{[1,i-1]}, R_{([p], [1,i-1])}, X^{[j,n]}, R_{([p], [j,n])}, P) \tag{since we showed above that \eqref{eqn:mic-independence-2} holds for $l$} \\
        &= 0 \tag{$\sfM_{l, j-1}$ is determined by $\sfM_{\le l, i-1}, X^{[i,j-1]}, R_{([p], [i,j-1])}, P$}.
    \end{align*}
    This completes the proof by induction.
\end{proof}

We now restate and prove \Cref{lem:upperboundofmic}.
\lemmamicmemory*

\begin{proof}
    The proof mimics the proof of Lemma 1.1 in \cite{braverman2024new}, albeit with the addition of public randomness $P$. We will prove that, for every pass $\ell \in [p]$, it holds that
    \begin{align}
        &\sum_{i=1}^n\sum_{j=1}^i \I\left(\sfM_{(\ell,i)};X^{j}\mid \sfM_{(\leq \ell,j-1)}, \sfM_{(\leq \ell-1,i)}, P\right) \le s \cdot n \label{eqn:mic-memory-eqn-1}\\
        &\sum_{i=1}^n\sum_{j=i+1}^n \I\left(\sfM_{(\ell,i)};X^{j}\mid \sfM_{(\leq \ell-1,j-1)}, \sfM_{(\leq \ell-1,i)}, P\right) \le s \cdot n \label{eqn:mic-memory-eqn-2},
    \end{align}
    which implies the lemma.

    We start by establishing \eqref{eqn:mic-memory-eqn-1}.
    \begin{align*}
        &\sum_{i=1}^n\sum_{j=1}^i \I\left(\sfM_{(\ell,i)};X^{j}\mid \sfM_{(\leq \ell,j-1)}, \sfM_{(\leq \ell-1,i)}, P\right) \\
        &\le \sum_{i=1}^n\sum_{j=1}^i \I\left(\sfM_{(\ell,i)};X^{j}, X^{[1,j-1]}, \sfM_{\le \ell, \le j-2}\mid \sfM_{(\leq \ell,j-1)}, \sfM_{(\leq \ell-1,i)}, P\right) \tag{chain rule and non-negativity of mutual information} \\
        &= \sum_{i=1}^n\sum_{j=1}^i \I\left(\sfM_{(\ell,i)}; X^{[1,j-1]}, \sfM_{\le \ell, \le j-2}\mid \sfM_{(\leq \ell,j-1)}, \sfM_{(\leq \ell-1,i)}, P\right) \\
        &\qquad+ \sum_{i=1}^n\sum_{j=1}^i \I\left(\sfM_{(\ell,i)};X^{j}\mid \sfM_{(\leq \ell,\le j-1)}, \sfM_{(\leq \ell-1,i)}, X^{[1,j-1]}, P\right) \tag{chain rule}\\
        &= \sum_{i=1}^n\sum_{j=1}^i \I\left(\sfM_{(\ell,i)};X^{j}\mid \sfM_{(\leq \ell,\le j-1)}, \sfM_{(\leq \ell-1,i)}, X^{[1,j-1]}, P\right) \tag{$\star$}\\
        &\le \sum_{i=1}^n\sum_{j=1}^i \I\left(\sfM_{(\ell,i)};X^{j}, \sfM_{(\le \ell, j)}\mid \sfM_{(\leq \ell,\le j-1)}, \sfM_{(\leq \ell-1,i)}, X^{[1,j-1]}, P\right) \tag{chain rule and non-negativity of mutual information} \\
        &= \sum_{i=1}^n \I\left(\sfM_{(\ell,i)}; X^{[1,i]}, \sfM_{(\le \ell, \le i)}\mid \sfM_{(\leq \ell-1,i)}, P\right) \tag{chain rule} \\
        &\le s \cdot n, \tag{\Cref{claim:info-ub-entropy}}
    \end{align*}
    which establishes \eqref{eqn:mic-memory-eqn-1}. In the above, $(\star)$ follows due to the fact that every summand in the first double-summation in the previous step is 0, namely, 
    \begin{align}
        \label{eqn:mic-memory-info-zero-1}
        \I\left(\sfM_{(\ell,i)}; X^{[1,j-1]}, \sfM_{\le \ell, \le j-2}\mid \sfM_{(\leq \ell,j-1)}, \sfM_{(\leq \ell-1,i)}, P\right) = 0.
    \end{align}
    To see this, recall that \eqref{eqn:mic-independence-2} in \Cref{lemma:mic-independence} implies that
    \begin{align*}
        I(X^{[j, i]}, R_{([p], [j, i])}\;;\; X^{[1,j-1]}, R_{([p], [1,j-1])}, X^{[i+1,n]}, R_{([p], [i+1,n])} \mid %
        \sfM_{\le \ell, j-1}, \sfM_{\le \ell-1, i}, P) = 0.
    \end{align*}
    Now, observe that $\sfM_{(\ell, i)}$ is a deterministic function of $X^{[j,i]}, R_{[p], [j,i]}$, conditioned on $\sfM_{(\le \ell, j-1)}, P$ Similarly, $X^{[1,j-1]}, \sfM_{\le \ell, \le j-2}$ is a deterministic function of $X^{[1,j-1]}, R_{([p], [1,j-1])}, X^{[i+1,n]}, R_{([p], [i+1,n])}$, conditioned on $\sfM_{(\leq \ell-1,i)}, P$. This implies \eqref{eqn:mic-memory-info-zero-1}.
    
    Similarly, for \eqref{eqn:mic-memory-eqn-2}, we have that
    \begin{align*}
        &\sum_{i=1}^n\sum_{j=i+1}^n \I\left(\sfM_{(\ell,i)};X^{j}\mid \sfM_{(\leq \ell-1,j-1)}, \sfM_{(\leq \ell-1,i)}, P\right) \\
        &\le \sum_{i=1}^n\sum_{j=i+1}^n \I\left(\sfM_{(\ell,i)};X^{j}, X^{[i+1, j-1]}, \sfM_{(\le \ell-1, [i+1, j-2])}\mid \sfM_{(\leq \ell-1,j-1)}, \sfM_{(\leq \ell-1,i)}, P\right) \tag{chain rule and non-negativity of mutual information} \\
        &= \sum_{i=1}^n\sum_{j=i+1}^n \I\left(\sfM_{(\ell,i)}; X^{[i+1, j-1]}, \sfM_{(\le \ell-1, [i+1, j-2])}\mid \sfM_{(\leq \ell-1,j-1)}, \sfM_{(\leq \ell-1,i)}, P\right) \\
        &\qquad+ \sum_{i=1}^n\sum_{j=i+1}^n \I\left(\sfM_{(\ell,i)};X^{j}\mid \sfM_{(\le \ell-1, [i+1, j-1])}, \sfM_{(\leq \ell-1,i)}, X^{[i+1, j-1]}, P\right) \tag{chain rule} \\
        &= \sum_{i=1}^n\sum_{j=i+1}^n \I\left(\sfM_{(\ell,i)};X^{j}\mid \sfM_{(\le \ell-1, [i+1, j-1])}, \sfM_{(\leq \ell-1,i)}, X^{[i+1, j-1]}, P\right) \tag{$\star$} \\
        &\le \sum_{i=1}^n\sum_{j=i+1}^n \I\left(\sfM_{(\ell,i)};X^{j}, \sfM_{(\le \ell-1, j)}\mid \sfM_{(\le \ell-1, [i+1, j-1])}, \sfM_{(\leq \ell-1,i)}, X^{[i+1, j-1]}, P\right) \tag{chain rule and non-negativity of mutual information} \\
        &= \sum_{i=1}^n \I\left(\sfM_{(\ell,i)};X^{[i+1, n]}, \sfM_{(\le \ell-1, [i+1, n])} \mid  \sfM_{(\leq \ell-1,i)},P\right) \tag{chain rule} \\
        &\le s \cdot n, \tag{\Cref{claim:info-ub-entropy}}
    \end{align*}
    which establishes \eqref{eqn:mic-memory-eqn-2}. Again, $(\star)$ above follows because every summand in the first double-summation in the previous step is 0, namely, 
    \begin{align}
        \label{eqn:mic-memory-info-zero-2}
         \I\left(\sfM_{(\ell,i)}; X^{[i+1, j-1]}, \sfM_{(\le \ell-1, [i+1, j-2])}\mid \sfM_{(\leq \ell-1,j-1)}, \sfM_{(\leq \ell-1,i)}, P\right) = 0.
    \end{align}
    To see this, recall that \eqref{eqn:mic-independence-1} in \Cref{lemma:mic-independence} implies that
    \begin{align*}
        I(X^{[i+1, j-1]}, R_{([p], [i+1, j-1])}\;;\; X^{[1,i]}, R_{([p], [1,i])}, X^{[j,n]}, R_{([p], [j,n])} \mid %
        \sfM_{\le \ell-1, i}, \sfM_{\le \ell-1, j-1}, P) = 0,
    \end{align*}
    Now, observe that $X^{[i+1, j-1]}, \sfM_{(\le \ell-1, [i+1, j-2])}$ is a deterministic function of $X^{[i+1, j-1]}, R_{([p], [i+1, j-1])}$, conditioned on $\sfM_{\le \ell-1, i}, P$. Similarly, $\sfM_{(\ell, i)}$ is a deterministic function of $X^{[1,i]}, R_{([p], [1,i])}, X^{[j,n]}, R_{([p], [j,n])}$, conditioned on $\sfM_{(\le \ell-1, j-1)}, P$. This implies \eqref{eqn:mic-memory-info-zero-2}.
\end{proof}

\begin{claim}
    \label{claim:info-ub-entropy}
    If $A$ is a discrete random variable with probability mass function $p_A$, and $B$ is an arbitrary random variable, then $I(A;B) \le H(A)$, where $H(A)=-\mathbb{E}_A[\log(p(A))]$ is the entropy of $A$. In particular, if $A$ has finite support of size $N$, then $I(A;B) \le \log(|N|)$.
\end{claim}
\begin{proof}
    By definition,
    \begin{align*}
        &I(A;B) = \mathbb{E}_B[D_{KL}(p_{A|B} ~||~ p_A)] = \mathbb{E}_B\left[\sum_{a}p_{A|B}(a)\log\left(\frac{p_{A|B}(a)}{p_A(a)}\right)\right] \\
        &= \mathbb{E}_B\left[ \sum_ap_{A|B}(a)\log(p_{A|B}(a))\right] - \mathbb{E}_B\left[ \sum_{a}p_{A|B}(a)\log(p_A(a))\right] \\
        &= \mathbb{E}_B\left[ \sum_ap_{A|B}(a)\log(p_{A|B}(a))\right] -  \sum_{a} \log(p_A(a)) \mathbb{E}_B\left[p_{A|B}(a)\right] \\
        &= \mathbb{E}_B\left[ \sum_ap_{A|B}(a)\log(p_{A|B}(a))\right] -  \sum_{a} \log(p_A(a)) p_A(a) \\
        &=  \mathbb{E}_B\left[ \sum_ap_{A|B}(a)\log(p_{A|B}(a))\right] + H(A).
    \end{align*}
    It remains to argue that the first summand above is non-positive. Note that  $\sum_ap_{A|B}(a)\log(p_{A|B}(a))$ is the expectation of the concave function $\log(p_{A|B}(\cdot))$ where the argument is drawn from the conditional distribution $A|B$. Applying Jensen's inequality, we get that 
    \begin{align*}
        \sum_ap_{A|B}(a)\log(p_{A|B}(a)) \le \log \left(\sum_a p_{A|B}(a)^2\right) \le \log(1) = 0.
    \end{align*}
    This concludes the proof.
\end{proof}

%% file: appendix-multi-ic.tex
\section{Proofs from \Cref{sec:main_tech}}
\label{sec:appendix-multi-ic}

We restate and prove \Cref{cl:micr}
\clmicr*
\begin{proof}[Proof of \Cref{cl:micr}]
We begin by reiterating that $\sfM$ is a $(p+1)$-pass algorithm that solves the distinguishing problem $\Dis(\mu_0,\{\mu_\theta\}_{\theta\in \Omega},P, \calT,k,n)$ with large enough constant probability, say $1-\delta$ (and recalling that, we added another pass that doesn't do any operations to $\sfM$).

We will convert the streaming algorithm into a communication protocol and use calculations similar those used in the proof of Claim 5.4 in \cite{braverman2024new}.

\paragraph{Communication Protocol $\Pi$ for $k$-party \generalCC \text{ using} $\sfM$.}
Given input $z^1, z^2,\ldots z^k\in\calX^d$ to the $k$-parties respectively, all the parties together prepare an input to the multi-pass streaming algorithm $\sfM$ with $z$ embedded at rows in $R$. For all $a\in [k-1]$, the $a$-th player, using private randomness, samples $\{x^i\}_{i_a<i< i_{a+1}}$ independently according to $\no$ (under $\no$, each row is $i.i.d.$) and sets $x^{i_a}=z^a$. The last player sets $x^{i_k}=z^k$, and samples $\{x^i\}_{i_k<i\le n}$ as well as $\{x^i\}_{1\le i<i_{1}}$. 

All players then simulate $\sfM$ one pass at a time, using their part of the input stream, public randomness $P$, and any additional private randomness that $\sfM$ requires. Knowing $\{x^i\}_{1\le i<i_{1}}$, the $k$-th player publishes memory state $m_{(1,i_1-1)}$, then knowing $\{x^i\}_{i_1\le i<i_{2}}$, $1$st player adds $m_{(1,i_2-1)}$ to the blackboard and so on. Finally, the last player adds the output of $\sfM$, given $m_{(p,i_k-1)}$ and knowing $\{x^i\}_{i_k \le i\le n}$. As the $(p+1)$th pass doesn't do any operations, $\forall i, m_{(p+1,i)}=m_{(p,n)}$. Thus, the transcript under $\Pi$ is the public randomness $P$, together with a sequence of memory states $m_{(1,i_1-1)}, m_{(1,i_2-1)},\ldots, m_{(1,i_k-1)}, m_{(2,i_1-1)},\ldots, m_{(2,i_k-1)},\ldots, m_{(p,i_k-1)}, m_{(p+1,i_1-1)}$. %

When $Z^1, Z^2,\ldots Z^k$ are distributed according to the No distribution for the $k$-party \generalCC, then $X^1,X^2,\ldots,X^n$ are distributed according to $\no$, and when $Z^1, Z^2,\ldots Z^k$ are distributed according to the Yes distribution for the $k$-party \generalCC, then $X^1,X^2,\ldots,X^n$ are distributed according to $\yes^\calT$ with the fixed $R$. As $R$ is good, the success probability of $\Pi$ is at least 0.9. By \Cref{lem:ccdirectsum},

\[\I\left(\PiR; Z^1,\ldots,Z^k\right)\ge \Omega\left(\frac{d}{c\cdot t}\right).\]

When $Z^1,\ldots,Z^k$ are distributed according to the No distribution for $k$-party \generalCC, $X^1,X^2,\ldots,X^n$ are distributed according to $\no$, and we can rewrite the information complexity of $\Pi$ as
\begin{align*}
\I\left(P,\sfM_{(1,i_1-1)}, \sfM_{(1,i_2-1)},\ldots, \sfM_{(1,i_k-1)}, \sfM_{(2,i_1-1)},\ldots, \sfM_{(2,i_k-1)},\ldots, \sfM_{(p,i_k-1)}, \sfM_{(p+1,i_1-1)}\;;\; X^{i_1},X^{i_2}\ldots,X^{i_k}\right).
\end{align*}

Using Chain Rule, we can now rewrite the  above mutual information as 
\begin{align}
\nonumber
&\underbrace{I\left(P;X^{i_1},X^{i_2}\ldots,X^{i_k}\right)}_{=\,0} +\sum_{\ell=1}^p\sum_{a=1}^k\I\left(\sfM_{(\ell,i_a-1)}\;;\; X^{i_1},X^{i_2}\ldots,X^{i_k}\mid \sfM_{(<\ell,\{i_1-1,\ldots,i_k-1\})}, \sfM_{(\ell,\{i_1-1,\ldots,i_{a-1}-1\})}, P\right) +\\
\nonumber
&\;\;\;\;\;\;\;\;\;\;\;\;\;\;\;\;\;\;\;\;\;\;\;\;\;\;\;\;\;\;\;\;\;\;\;\;\;\;\;\;\;\;\;\;\;\;\;\;\I\left(\sfM_{(p+1;i_1-1)}\;;\; X^{i_1},X^{i_2}\ldots,X^{i_k}\mid \sfM_{(\le p,\{i_1-1,\ldots,i_k-1\})}, P\right)\\
\label{eq:mic5}
&\le \sum_{\ell=1}^{p+1}\sum_{a=1}^k\I\left(\sfM_{(\ell,i_a-1)}\;;\; X^{i_1},X^{i_2}\ldots,X^{i_k}\mid \sfM_{(<\ell,\{i_1-1,\ldots,i_k-1\})}, \sfM_{(\ell,\{i_1-1,\ldots,i_{a-1}-1\})}, P\right).
\end{align}

    In the inner summation above, let us first consider any summand corresponding to a value of $a$ satisfying $1 < a \le k$. For such a summand, define:
    \begin{align}
        &\sfM_{neighbors} = (\sfM_{<l, i_{a}-1}, \sfM_{\le l, i_{a-1}-1}, P), \label{eqn:neighbors}\\
        &\sfM_{non-neighbors} = (\sfM_{\le l, i_{1}-1}, \dots, \sfM_{\le l, i_{a-2}-1}, \sfM_{<l, i_{a+1}-1}, \dots, \sfM_{<l, i_{k}-1}). \label{eqn:non-neighbors}
    \end{align}
    Observe that
    \begin{align*}
        (\sfM_{neighbors}, \sfM_{non-neighbors}) &= (\sfM_{l, \{i_1-1,\dots,i_{a-1}-1\}}, \sfM_{<l, \{i_1-1,\dots,i_k-1\}}, P).
    \end{align*}
    Now, let $X^{\neq i_{a-1}}=(X^{i_1},\dots,X^{i_{a-2}}, X^{i_a},\dots,X^{i_k})$. Notice then that we can write the summand in the inner summation, using the chain rule, as  
    \begin{align}
        &I\left(\sfM_{(\ell,i_a-1)}\;;\; X^{i_1},X^{i_2}\ldots,X^{i_k}\mid \sfM_{(<\ell,\{i_1-1,\ldots,i_k-1\})}, \sfM_{(\ell,\{i_1-1,\ldots,i_{a-1}-1\})}, P\right) \nonumber \\
        &= I\left(\sfM_{(\ell,i_a-1)}\;;\; X^{i_{a-1}}\mid \sfM_{neighbors}, \sfM_{non-neighbors}\right) + I\left(\sfM_{(\ell,i_a-1)}\;;\; X^{\neq i_{a-1}}\mid \sfM_{neighbors}, \sfM_{non-neighbors}, X^{i_{a-1}}\right) \label{eqn:neighbor-term-that-is-zero-1}
    \end{align}
    We now focus on the second term in \eqref{eqn:neighbor-term-that-is-zero-1}. From \eqref{eqn:mic-independence-2} in \Cref{lemma:mic-independence}, where we consider $i=i_{a-1}, j=i_{a}$, we know that,
    \begin{align*}
        &I(X^{[i_{a-1}, i_{a}-1]}, R_{([p], [i_{a-1}, i_{a}-1])}\;;\; X^{[1,i_{a-1}-1]}, R_{([p], [1,i_{a-1}-1])}, X^{[i_{a},n]}, R_{([p], [i_{a},n])} \mid \underbrace{\sfM_{\le l, i_{a-1}-1}, \sfM_{<l, i_{a}-1}, P}_{\sfM_{neighbors}})= 0 \\
        \implies & I(X^{[i_{a-1}, i_{a}-1]}, R_{([p], [i_{a-1}, i_{a}-1])}\;;\; X^{[1,i_{a-1}-1]}, R_{([p], [1,i_{a-1}-1])}, X^{[i_{a},n]}, R_{([p], [i_{a},n])} \mid X^{i_{a-1}}, \sfM_{neighbors}) = 0 \tag{since $I(A,B;C|D)=0 \implies I(A,B;C|B,D)=0$}. %
    \end{align*}
    Now observe that $\sfM_{(l, i_{a}-1)}$ is a deterministic function of $X^{[i_{a-1}, i_{a}-1]}, R_{([p], [i_{a-1}, i_{a}-1])}$---the first argument in the mutual information above---conditioned on $X^{i_{a-1}}, \sfM_{neighbors}$. Similarly, observe that $(X^{\neq i_{a-1}}, \sfM_{non-neighbors})$ are deterministic functions of  $X^{[1,i_{a-1}-1]}, R_{([p], [1,i_{a-1}-1])}, X^{[i_{a},n]}, R_{([p], [i_{a},n])}$---the second argument in the mutual information above---conditioned on $X^{i_{a-1}}, \sfM_{neighbors}$. Accounting this in, we get that 
    \begin{align}
        &I(\sfM_{(l, i_{a}-1)}\;;\; X^{\neq i_{a-1}}, \sfM_{non-neighbors} \mid X^{i_{a-1}}, \sfM_{neighbors}) = 0 \label{eqn:neighbor-blocker-1}\\
        \implies\qquad &I(\sfM_{(l, i_{a}-1)}\;;\; \sfM_{non-neighbors} \mid X^{i_{a-1}}, \sfM_{neighbors}) = 0 \label{eqn:neighbor-blocker-2} \\
        & \text{as well as} \qquad  I\left(\sfM_{(\ell,i_a-1)}\;;\; X^{\neq i_{a-1}}\mid \sfM_{neighbors}, \sfM_{non-neighbors}, X^{i_{a-1}}\right) = 0. \label{eqn:neighbor-blocker-3}
    \end{align}
    where the last two implications follow by chain rule. Substituting \eqref{eqn:neighbor-blocker-3} in \eqref{eqn:neighbor-term-that-is-zero-1}, we get that
    \begin{align}
        &I\left(\sfM_{(\ell,i_a-1)}\;;\; X^{i_1},X^{i_2}\ldots,X^{i_k}\mid \sfM_{(<\ell,\{i_1-1,\ldots,i_k-1\})}, \sfM_{(\ell,\{i_1-1,\ldots,i_{a-1}-1\})}, P\right) \nonumber\\
        &= I\left(\sfM_{(\ell,i_a-1)}\;;\; X^{i_{a-1}}\mid \sfM_{neighbors}, \sfM_{non-neighbors}\right) \nonumber \\
        &\le I\left(\sfM_{(\ell,i_a-1)}\;;\; X^{i_{a-1}}\mid \sfM_{neighbors}\right) = I\left(\sfM_{(\ell,i_a-1)}\;;\; X^{i_{a-1}}\mid \sfM_{<l, i_{a}-1}, \sfM_{\le l, i_{a-1}-1}, P\right) \label{eqn:one-bound}.
    \end{align}
    In the inequality above, we used \eqref{eqn:neighbor-blocker-2}, together with the fact that, if $I(A;B|C,D)=0$, then $I(C;B|D,A) \le I(C;B|D)$ (for us, $A=\sfM_{non-neighbors}, B=\sfM_{(l, i_{a}-1)}, C=X^{i_{a-1}}, D = \sfM_{neighbors}$). In the last equality, we simply recalled the definition \eqref{eqn:neighbors} of $\sfM_{neighbors}$.

    We now consider the summand in the inner summation in \eqref{eq:mic5} corresponding to $a=1$. For this summand, define
    \begin{align}
        &\sfM_{neighbors} = (\sfM_{<l, i_{1}-1}, \sfM_{< l, i_{k}-1}, P), \label{eqn:neighbors-2}\\
        &\sfM_{non-neighbors} = (\sfM_{<l, \{i_2-1,\dots,i_{k-1}-1\}}). \label{eqn:non-neighbors-2}
    \end{align}
    We have that
    \begin{align*}
        (\sfM_{neighbors}, \sfM_{non-neighbors}) &= (\sfM_{<l, \{i_1-1,\dots,i_k-1\}}, P).
    \end{align*}
    Now, let $X^{\neq i_{k}}=(X^{i_1},\dots,X^{i_{k-1}})$. Notice then that we can write the summand in the inner summation, using the chain rule, as  
    \begin{align}
        &I\left(\sfM_{(\ell,i_1-1)}\;;\; X^{i_1},X^{i_2}\ldots,X^{i_k}\mid \sfM_{(<\ell,\{i_1-1,\ldots,i_k-1\})}, P\right) \nonumber \\
        &= I\left(\sfM_{(\ell,i_1-1)}\;;\; X^{i_{k}}\mid \sfM_{neighbors}, \sfM_{non-neighbors}\right) + I\left(\sfM_{(\ell,i_1-1)}\;;\; X^{\neq i_{k}}\mid \sfM_{neighbors}, \sfM_{non-neighbors}, X^{i_{k}}\right) \label{eqn:neighbor-term-that-is-zero-2}
    \end{align}

    We now focus on the second term in \eqref{eqn:neighbor-term-that-is-zero-2}. From \eqref{eqn:mic-independence-1} in \Cref{lemma:mic-independence} where we consider $i=i_{1}, j=i_{k}$, we know that,
    \begin{align*}
        &I(X^{[i_1, i_k-1]}, R_{([p], [i_1, i_k-1])}\;;\; X^{[1,i_1-1]}, R_{([p], [1,i_1-1])}, X^{[i_k,n]}, R_{([p], [i_k,n])} \mid  \underbrace{\sfM_{<l, i_1-1}, \sfM_{<l, i_k-1}, P}_{\sfM_{neighbors}})= 0 \\
        \implies&I(X^{[i_1, i_k-1]}, R_{([p], [i_1, i_k-1])}\;;\; X^{[1,i_1-1]}, R_{([p], [1,i_1-1])}, X^{[i_k,n]}, R_{([p], [i_k,n])} \mid X^{i_k}, \sfM_{neighbors})= 0 \tag{since $I(A,B;C|D)=0 \implies I(A,B;C|B,D)=0$}.
    \end{align*}
    Now observe that $\sfM_{(l, i_{1}-1)}$ is a deterministic function of $X^{[1,i_1-1]}, R_{([p], [1,i_1-1])}, X^{[i_k,n]}, R_{([p], [i_k,n])}$---the second argument in the mutual information above---conditioned on $X^{i_{k}}, \sfM_{neighbors}$. Similarly, observe that $(X^{\neq i_{k}}, \sfM_{non-neighbors})$ are deterministic functions of  $X^{[i_1, i_k-1]}, R_{([p], [i_1, i_k-1])}$---the first argument in the mutual information above---conditioned on $X^{i_{k}}, \sfM_{neighbors}$. Accounting for this, we get that 
    \begin{align}
        &I(\sfM_{(l, i_{1}-1)}\;;\; X^{\neq i_{k}}, \sfM_{non-neighbors} \mid X^{i_{k}}, \sfM_{neighbors}) = 0 \label{eqn:neighbor-blocker-11}\\
        \implies\qquad &I(\sfM_{(l, i_{1}-1)}\;;\; \sfM_{non-neighbors} \mid X^{i_{k}}, \sfM_{neighbors}) = 0 \label{eqn:neighbor-blocker-22} \\
        & \text{as well as} \qquad  I\left(\sfM_{(\ell,i_1-1)}\;;\; X^{\neq i_{k}}\mid \sfM_{neighbors}, \sfM_{non-neighbors}, X^{i_{k}}\right) = 0. \label{eqn:neighbor-blocker-33}
    \end{align}
    where the last two implications follow by chain rule. Substituting \eqref{eqn:neighbor-blocker-33} in \eqref{eqn:neighbor-term-that-is-zero-2}, we get that
    \begin{align}
        &I\left(\sfM_{(\ell,i_1-1)}\;;\; X^{i_1},X^{i_2}\ldots,X^{i_k}\mid \sfM_{(<\ell,\{i_1-1,\ldots,i_k-1\})}, P\right) = I\left(\sfM_{(\ell,i_1-1)}\;;\; X^{i_{k}}\mid \sfM_{neighbors}, \sfM_{non-neighbors}\right) \nonumber \\
        &\le I\left(\sfM_{(\ell,i_1-1)}\;;\; X^{i_{k}}\mid \sfM_{neighbors}\right) = I\left(\sfM_{(\ell,i_1-1)}\;;\; X^{i_{k}}\mid \sfM_{<l, i_{1}-1}, \sfM_{< l, i_{k}-1}, P\right) \label{eqn:two-bound}.
    \end{align}
    In the inequality above, we used \eqref{eqn:neighbor-blocker-22}, together with the fact that, if $I(A;B|C,D)=0$, then $I(C;B|D,A) \le I(C;B|D)$ (for us, $A=\sfM_{non-neighbors}, B=\sfM_{(l, i_{1}-1)}, C=X^{i_{k}}, D = \sfM_{neighbors}$). In the last equality, we simply recalled the definition \eqref{eqn:neighbors-2} of $\sfM_{neighbors}$.

    To conclude, observe that \eqref{eqn:one-bound} and \eqref{eqn:two-bound} together imply that \eqref{eq:mic5} is upper bounded as
    \begin{align*}
        &\sum_{\ell=1}^{p+1}\sum_{a=1}^k\I\left(\sfM_{(\ell,i_a-1)}\;;\; X^{i_1},X^{i_2}\ldots,X^{i_k}\mid \sfM_{(<\ell,\{i_1-1,\ldots,i_k-1\})}, \sfM_{(\ell,\{i_1-1,\ldots,i_{a-1}-1\})}, P\right) \\
        &\le \sum_{\ell=1}^{p+1} I\left(\sfM_{(\ell,i_1-1)}\;;\; X^{i_{k}}\mid \sfM_{<l, i_{1}-1}, \sfM_{< l, i_{k}-1}, P\right) + \sum_{\ell=1}^{p+1}\sum_{a=2}^k I\left(\sfM_{(\ell,i_a-1)}\;;\; X^{i_{a-1}}\mid \sfM_{<l, i_{a}-1}, \sfM_{\le l, i_{a-1}-1}, P\right) \\
        &\le \sum_{\ell=1}^{p+1}\sum_{a=1}^k\sum_{b=a+1}^{k} I\left(\sfM_{(\ell,i_a-1)}\;;\; X^{i_{b}}\mid \sfM_{<l, i_{a}-1}, \sfM_{< l, i_{b}-1}, P\right) + \sum_{\ell=1}^{p+1}\sum_{a=1}^k\sum_{b=1}^{a-1} I\left(\sfM_{(\ell,i_a-1)}\;;\; X^{i_{b}}\mid \sfM_{<l, i_{a}-1}, \sfM_{\le l, i_{b}-1}, P\right) \tag{non-negativity of mutual information}\\
        &= \mic^R.
    \end{align*}
    This completes the proof.

\end{proof}

We restate and prove \Cref{lem:cutpaste}.
\lemcutpaste*
\begin{proof}[Proof of \Cref{lem:cutpaste}]
As in \cite{braverman2016communication} we use certain basic properties of transcripts established in \cite{bar2004information}. We first note that fixing input $x_1,x_2,\ldots,x_m$, the probability of any transcript can be factored as,
\begin{align}
    \Pr [ \Pi(x)=\pi ] = p_{1,\pi}(x_1) \cdots p_{m,\pi}(x_m),
\end{align}
where $p_{i,\pi}(x_i)$ is some function which depends only on $\pi$ and $x_i$. Recall that  $\PiR^\theta_{\bb}$ is the distribution of $\Pi(x_1,\dots, x_m)$ when $(x_1,\dots, x_m)\sim \mu^\theta_{\bb}$, which is a product distribution. Therefore, if $\tilde X\sim \mu^\theta_{\bb}$ and since $\mu^\theta_{\bb}$ is a product measure (for fixed $\theta$), we can marginalize over $\tilde X$ and obtain the marginal distribution over the transcripts $\Pi$ for all $\bb$;
\begin{align}
    \Pr [ \Pi(\tilde X)= \pi] = q_{1,\pi,\theta}(\bb_1)\cdots q_{m,\pi,\theta}(\bb_m),\label{eq:protocol}
\end{align}
where $q_{i,\pi,\theta}(\bb_i) = \int_{x_i} p_{i,\pi}(x_i) d\mu_{\bb_i}^\theta$ is the marginal distribution of  $p_{i,\pi}(x_i)$ over $x_i \sim \mu_{\bb_i}^\theta$. 
Therefore, for all $\bb$;
$$\Pr\left[\PiR^\theta_{\bb}=\pi\right]=q_{1,\pi,\theta}(\bb_1)\cdots q_{m,\pi,\theta}(\bb_m), $$
and the claim follows because of this decomposition.
Since the squared Hellinger distance ${h^2\left(\PiR^\theta_{\bb^1} \parallel \PiR^\theta_{\bb^2} \right)}$ only depends on the two distributions through the product of the probabilities, that is, $\Pr\left[\PiR^\theta_{\bb^1}=\pi\right] \cdot \Pr\left[\PiR^\theta_{\bb^2}=\pi\right]$ for all transcripts $\pi$, the result follows.
\end{proof}

We now restate and prove \Cref{lem:boundhellinger}
\lemboundhellinger*

\begin{proof}[Proof of \Cref{lem:boundhellinger}]
Fix $\theta$. If the protocol succeeds with probability $\alpha$ conditioned on $\bth=\theta$, then we have,
    \begin{gather*}
(1/2)\Pr_{\forall i, x^i\sim \mu_0}\left[\Pi(x_1,x_2,\ldots,x_m)=\text{"No"}\right]+(1/2)\Pr_{ \forall i, x^i\sim \mu_{\theta}}\left[\Pi(x_1,x_2,\ldots,x_k)=\text{"Yes"}\right]=\alpha,\\
\implies \Pr_{\forall i, x^i\sim \mu_0}\left[\Pi(x_1,x_2,\ldots,x_m)=\text{"No"}\right]-\Pr_{ \forall i, x^i\sim \mu_{\theta}}\left[\Pi(x_1,x_2,\ldots,x_k)=\text{"No"}\right]\ge 2\alpha -1,\\
\implies \| \PiR_{\mid V=0} - \PiR_{\mid V=1,\bth=\theta}\|_{TV} \ge  2\alpha-1.
\end{gather*}
Since the protocol has overall success probability at least $0.9$ on average over the randomness in the choice of $\theta$, we have
 \begin{align}
        \Eb_{\theta\sim P}  \| \PiR_{\mid V=0} - \PiR_{\mid V=1,\bth=\theta}\|_{TV} \ge 2(0.9)-1=0.8.\label{eq:hellinger1}
    \end{align}
We will next use the following folklore result to relate the TV distance to the Hellinger distance. 
    
\begin{fact}\label{fact:hellinger}
    For any two distributions $P$ and $Q$ we have,
    \begin{align*}
        h^2(P,Q) \le \| P-Q \|_{TV} \le \sqrt{2}h(P,Q).
    \end{align*}
\end{fact}

Using \Cref{fact:hellinger},
\begin{align}
     \Eb_{\theta\sim P}  \left[h( \PiR_{\mid V=0} \parallel \PiR_{\mid V=1,\bth=\theta}) \right]\ge \frac{1}{\sqrt{2}}\Eb_{\theta\sim P}  \| \PiR_{\mid V=0} - \PiR_{\mid V=1,\bth=\theta}\|_{TV}.\label{eq:hellinger2}
\end{align}

By Jensen's inequality,
\begin{align*}
     \Eb_{\theta\sim P}  \left[h^2( \PiR_{\mid V=0} \parallel \PiR_{\mid V=1,\bth=\theta}) \right]\ge \left(\Eb_{\theta\sim P}  \left[h( \PiR_{\mid V=0} \parallel \PiR_{\mid V=1,\bth=\theta}) \right]\right)^2.
\end{align*}

Using \eqref{eq:hellinger2},

\begin{align*}
     \Eb_{\theta\sim P}  \left[h^2( \PiR_{\mid V=0} \parallel \PiR_{\mid V=1,\bth=\theta}) \right]\ge (1/2) \left(\Eb_{\theta\sim P}  \| \PiR_{\mid V=0} - \PiR_{\mid V=1,\bth=\theta}\|_{TV}\right)^2.
\end{align*}

Combining this with \eqref{eq:hellinger1} completes the proof.
\end{proof}

%% file: appendix_planted_clique.tex
\section{Proofs from Section \ref{sec:biclique}}
\label{sec:appendix-biclique}

We first restate and prove \Cref{claim:c-bound-p-trunc}
\claimcboundptrunc*
\begin{proof}[Proof of \Cref{claim:c-bound-p-trunc}]
    Note that a draw $X$ from the distribution $\Eb_{S}[\pt^{1,S}]$ corresponds to first drawing a $k$-sized $S \subseteq [t]$ uniformly at random, and then drawing $X \sim \pt^{1, S}$.  Fix any $x$: we will upper bound the mass that $\mu_1$ assigns to $x$ in terms of the mass that $\mu_0$ assigns to $x$. This will establish the claim.

    Note first that $\mu_0(x)=1/|T|$, where $T$ was defined in \eqref{eqn:def-T}. Now let $I_x = \{i \in [t]:x_i=1\}$, and observe that 
    \begin{align*}
        \mu_1(x) = \sum_{S \subseteq I_x, |S|=k} \mu_1(S) \cdot \mu_1(x|S).
    \end{align*}
    But conditioned on $S$, the distribution $\mu_1(\cdot|S)$ is uniform on the set $T_S$ defined in \eqref{eqn:def-T_S}. Noting that every $T_S$ has the same size, this immediately gives us that 
    \begin{align*}
        \mu_1(x) &= \frac{1}{|T_S|}\sum_{S \subseteq I_x, |S|=k}\mu_1(S) = \frac{\binom{|I_x|}{k}}{|T_S|\binom{t}{k}}.
    \end{align*}
    Thus, we obtain
    \begin{align*}
        \frac{\mu_1(x)}{\mu_0(x)} = \frac{\binom{|I_x|}{k}|T|}{\binom{t}{k}|T_S|} &= \frac{\binom{|I_x|}{k}\sum_{y=tq-C\sqrt{tq\log(nm)}}^{tq+C\sqrt{tq\log(nm)}}\binom{t}{y}}{\binom{t}{k}\sum_{y=tq-C\sqrt{tq\log(nm)}}^{tq+C\sqrt{tq\log(nm)}}\binom{t-k}{y-k}} \le  \frac{\binom{|I_x|}{k}}{\binom{t}{k}} \max_y \frac{\binom{t}{y}}{\binom{t-k}{y-k}} \le \frac{\binom{tq+C\sqrt{tq \log(nm)}}{k}}{\binom{t}{k}} \max_y \frac{\binom{t}{y}}{\binom{t-k}{y-k}} \\
        &= \frac{(tq+C\sqrt{tq \log(nm)})!(t-k)!}{t!(tq+C\sqrt{tq \log(nm)}-k)!}\max_y \frac{\binom{t}{y}}{\binom{t-k}{y-k}},
    \end{align*}
    where we used that $\frac{\sum_{i=1}^n a_i}{\sum_{i=1}^nb_i}\le \max_i\frac{a_i}{b_i}$ for positive $a_i, b_i$. Observe that
    \begin{align*}
        \frac{(tq+C\sqrt{tq \log(nm)})!(t-k)!}{t!(tq+C\sqrt{tq \log(nm)}-k)!} &= \frac{(tq+C\sqrt{tq \log(nm)})(tq+C\sqrt{tq \log(nm)}-1)\ldots(tq+C\sqrt{tq \log(nm)}-(k-1))}{t(t-1)\ldots(t-(k-1))} \\
        &= \left(\frac{tq+C\sqrt{tq\log(nm)}}{t}\right)\ldots\left(\frac{(t-(k-1))q+C\sqrt{tq\log(nm)}-(1-q)(k-1)}{t-(k-1)}\right) \\
        &= q^k\left(1+\frac{C\sqrt{tq \log(nm)}}{tq}\right)\ldots\left(1+\frac{C\sqrt{tq\log(nm)}-(1-q)(k-1)}{(t-(k-1))q}\right) \\
        &\le q^k\left(1+\frac{C\sqrt{tq \log(nm)}}{(t-(k-1))q} \right)^{k} \\
        &= q^k\left[1+O\left(k\sqrt{\frac{\log(nm)}{tq}}\right) \right],
    \end{align*}
    where in the last step, we used that $t \ge \frac{Ck^2\log(nm)}{q}$. 
    
    Furthermore, for any $y$, observe also that
    \begin{align*}
        \frac{\binom{t}{y}}{\binom{t-k}{y-k}} &= \frac{t!(y-k)!}{y!(t-k)!} = \frac{t(t-1)\ldots(t-(k-1))}{y(y-1)\ldots(y-(k-1))} = \left(1+\frac{t-y}{y}\right)\ldots\left(1+\frac{t-y}{y-(k-1)}\right) \\
        &\le \left(1+\frac{t-y}{y-(k-1)}\right)^{k} \le \left(1+\frac{t-tq+C\sqrt{tq \log(nm)}}{tq-C\sqrt{tq\log(nm)}-k+1}\right)^{k} \\
        &\le \left(1+\frac{t-tq+C\sqrt{tq \log(nm)}}{tq-2C\sqrt{tq\log(nm)}}\right)^{k} \qquad \left(\text{since $k \le \sqrt{\frac{tq}{C\log(nm)}} < C\sqrt{tq\log(nm)}$}\right)\\
        &= \left(\frac{1}{q}+\frac{2C\sqrt{tq\log(nm)}-Cq\sqrt{tq\log(nm)}}{q(tq-2C\sqrt{tq\log(nm)})}\right)^{k} =\frac{1}{q^k}\left(1+\frac{C(2-q)\sqrt{tq\log(nm)}}{tq-2C\sqrt{tq\log(nm)}}\right)^{k} \\
        &\le \frac{1}{q^k}\left(1+\frac{2C\sqrt{tq\log(nm)}}{tq-2C\sqrt{tq\log(nm)}}\right)^{k} \le \frac{1}{q^k}\left[1+O\left(k\sqrt{\frac{\log(nm)}{tq}}\right)\right],
    \end{align*}
    where in the last step, we again used the assumption that $t \ge \frac{Ck^2\log(nm)}{q}$.
    Thus, we get that
    \begin{align*}
        \frac{\mu_1(x)}{\mu_0(x)} \le q^k\left[1+O\left(k\sqrt{\frac{\log(nm)}{tq}}\right) \right]\cdot \frac{1}{q^k}\left[1+O\left(k\sqrt{\frac{\log(nm)}{tq}}\right) \right] \le 1+O\left(k\sqrt{\frac{\log(nm)}{tq}}\right) = O(1).
    \end{align*}
\end{proof}

We now restate and prove \Cref{thm:main}
\theoremplantedclique*

\begin{proof}
    Suppose there was a $p$-pass streaming algorithm $\mathcal{A}$ that solves \Cref{prob:biclique} using only $o\left(\frac{mnq}{pk^4\log n}\right)$ bits of memory. Our approach will be to use the existence of $\mathcal{A}$ to construct a $p$-pass streaming algorithm $\mathcal{A}$' for \Cref{prob:partition} that circumvents the lower bound of \Cref{lem:partition}, yielding a contradiction.
    
    For this, let $t=\left\lceil \frac{Ck^2\log (nm)}{q} \right\rceil$ for a suitably large constant $C$. Let $n'=t \cdot \left\lfloor \frac{n}{t}\right\rfloor$. The algorithm $\mathcal{A}'$ operates as follows: First, using public randomness, it draws $I \subseteq [n]$ of size $n-n'$ uniformly at random, and then draws $b^{1},\ldots,b^{m} \in \{0,1\}^{n-n'}$, where every bit in every $b^i$ is independently drawn as $\Ber(q)$. Then, using public randomness again, it draws a uniformly random permutation $\pi$ of $[n']$. Upon receiving a stream $z^1,\dots,z^m$ of $n'$-bit vectors from an instance of \Cref{prob:partition}, algorithm $\mathcal{A'}$ translates this stream into a stream of $n$-bit vectors $y^1,\dots,y^m$. Namely, $y^i$ is constructed from $z^i$ as follows: First, $y^i_I$ is assigned to be $b^i$. Then, $z^i$ is permuted according to $\pi$, yielding $z^i_\pi$. Finally, $y^i_{[n] \setminus I}$ is assigned to be $z^i_\pi$. Observe that $\mathcal{A}'$ can construct this stream $y^1,\dots,y^m$ using only a constant memory overhead (since the public randomness does not contribute to the memory requirement). The algorithm $\mathcal{A}'$ then feeds this stream $y^1,\dots,y^m$ to $\mathcal{A}$, and returns the output of $\mathcal{A}$. The total memory requirement of $\mathcal{A}'$ and $\mathcal{A}$ is thus the same, upto an additive constant.

    We will now argue that $\mathcal{A}'$ correctly solves \Cref{prob:partition}.

    \paragraph{Case 1:} First, consider the case that $z^1,\dots,z^m$ were draws from $D_0$ in \Cref{prob:partition}. We observe that on account of permuting uniformly at random according to $\pi$, the distribution of $y^1,\dots,y^m$ that $\mathcal{A}'$ constructs is equivalent to the following random process:
    \begin{enumerate}
        \item[(1)] Draw a subset $I \subseteq [n]$ of size $n-n'$ uniformly at random. 
        \item[(2)] Draw $b^{1},\ldots,b^{m} \in \{0,1\}^{n-n'}$, where every bit in every $b^i$ is independently drawn as $\Ber(q)$. 
        \item[(3)] Set $y^i_I = b^i$ for every $i \in [m]$. 
        \item[(4)] Draw a uniformly random partition $\calT=\{T_{r'}\}_{r'\in[n'/t]}$ of $[n] \setminus I$, where $\forall r', |T_{r'}|=t$. 
        \noindent
        \tikz{\draw[dashed] (0,0) -- (\linewidth,0);}
        \item[(5)] For every $i \in [m], r' \in [n'/t]$, draw $y^i_{T_{r'}} \sim \pt^0$.        
    \end{enumerate}

    But on the other hand, notice also that a draw $x^1,\dots,x^m$ from $\Du$ in \Cref{prob:biclique} is equivalent to the same random process as above, but with (5) replaced as (5') ahead:
    \begin{enumerate}
        \item[(5')] For every $i \in [m], r' \in [n'/t]$, draw $x^i_{T_{r'}} \in \{0,1\}^t$, where every bit in $x^i_{T_{r'}}$ is drawn as $\Ber(q)$. 
    \end{enumerate}

    We observe that the distributions of $y^1,\dots,y^m$ and $x^1,\dots,x^m$ can be decomposed respectively as as $D_y = \sum_{I,b,\calT} D^{I,b,\calT} \cdot D_y^{I,b,\calT}$ and $D_x = \sum_{I,b,\calT} D^{I,b,\calT} \cdot D_x^{I,b,\calT}$, where the marginal distribution over $I$, bit-strings $b=\{b^1,\ldots,b^m\}$ and partition $\calT$, corresponding to Steps (1)-(4) above, is the \textit{same} in both cases, and $D_y^{I,b,\calT}$ and $D_x^{I,b,\calT}$ are the conditional distributions according to Step (5) and (5') respectively. The difference between $D_y^{I, b, \calT}$ and $D_x^{I, b, \calT}$ is that in $D_y^{I, b, \calT}$, every $y^i_{T_{r'}}$ is drawn from $\pt^0$; on the other hand, in $D_x^{I, b, \calT}$, every $x^i_{T_{r'}}$ is drawn such that every bit in it is an independent $\Ber(q)$.
    Let $A,B$ be random variables such that $A$ has the distribution of $y^i_{T_{r'}}$ in the former case, whereas $B$ has the distribution of $x^i_{T_{r'}}$ in the latter case. Observe that the distribution of $A$ is identical to the distribution of $B$, \textit{conditioned} on the event that the number of ones in $B$ is in the range $tq\pm C\sqrt{tq \log(nm)}$. We therefore have that the TV distance between the distributions of $A$ and $B$ is at most the probability that a $Bin(t, q)$ random variable is not in the range $tq\pm C\sqrt{tq \log(nm)}$, which, by a Chernoff bound, is at most $1/(nm)^{10}$ (for suitably large $C$). Therefore, the TV distance between $D_y^{I,b,\calT}$ and $D_x^{I,b,\calT}$, which is the TV distance between the product distribution of $m\cdot n'/t$ such random variables, is at most $1/(nm)^9$. 
    
    Summarily, we have shown that the distribution of $y^1,\dots,y^m$, in the case that $z^1,\dots,z^m$ were drawn from $D_0$ in \Cref{prob:partition} is $o(1)$ close in TV distance to the distribution $\Du$ in \Cref{prob:biclique}.

    \paragraph{Case 2:} Now, consider the case that $z^1,\dots,z^m$ were draws from $D_1^\calT$ in \Cref{prob:partition}. The distribution of $y^1,\dots,y^m$ that $\mathcal{A}'$ constructs can then be described by the random process comprising of Steps (1)-(4) above in Case 1, followed by the steps ahead:

    \begin{enumerate}
        \item[(5)] Draw $r$ uniformly at random from $[n'/t]$.
        \item[(6)] Draw $S \subseteq T_r$ uniformly at random of size $k$, and $R \subseteq [m]$ uniformly at random of size $k$. 
        \tikz{\draw[dashed] (0,0) -- (\linewidth,0);}
        \item[(7)] For every $i \notin R$ and $r' \in [n'/t]$, draw $y^i_{T_{r'}} \sim \pt^0$.
        \item[(8)] For every $i \in R$, draw $y^i_{T_r} \sim \pt^{1,S}$. Whereas for every $r' \neq r$, draw $y^i_{T_{r'}} \sim \pt^0$.
    \end{enumerate}

    But on the other hand, notice also that a draw $x^1,\dots,x^m$ from $\Dp$ in \Cref{prob:biclique} is equivalent to the same random process as above, but with Steps (7) and (8) replaced as (7') and (8') ahead:
    \begin{enumerate}
        \item[(7')] For every $i \notin R$, $r' \in [n'/t]$, draw $x^i_{T_{r'}} \in \{0,1\}^t$, where every bit in $x^i_{T_{r'}}$ is drawn as $\Ber(q)$. 
        \item[(8')] For every $i \in R$, for every $j \in T_r$, set $x^i_j = 1$ if $j \in S$, else set it to $\Ber(q)$. Whereas for every $r' \neq r$, draw $x^i_{T_{r'}} \in \{0,1\}^t$, where every bit in $x^i_{T_{r'}}$ is drawn as $\Ber(q)$.
    \end{enumerate}

    Again, we observe that the distributions of $y^1,\dots,y^m$ and $x^1,\dots,x^m$ can be decomposed respectively as 
    $$D_y = \sum_{I,b,\calT,r,S,R} D^{I,b,\calT,r,S,R} \cdot D_y^{I,b,\calT,r,S,R}, \qquad D_x = \sum_{I,b,\calT,r,S,R} D^{I,b,\calT,r,S,R} \cdot D_x^{I,b,\calT,r,S,R},
    $$
    where the marginal distribution over $I$, bit-strings $b=\{b^1,\ldots,b^m\}$, partition $\calT$, planted partition $r$, planted columns $S$ and planted rows $R$ corresponding to Steps (1)-(6) above, is the \textit{same} in both cases, and $D_y^{I,b,\calT,r,S,R}$ and $D_x^{I,b,\calT,r,S,R}$ are the conditional distributions corresponding to Steps (7), (8) and (7'), (8') respectively. Furthermore, both these conditional distributions are product distributions on $m \cdot n'/t := M$ random variables---denote these as $A^1,\ldots,A^{M}$ and $B^1,\ldots,B^{M}$ respectively. Observe that all but $k$ of the random variables $A^i$ are distributed identically as the random variable $A$ in Case 1 above, and the corresponding random variables $B^i$ are distributed identically as the random variable $B$---the TV distance between the distribution of each such $A^i$ and $B^i$ is hence at most $1/(nm)^{10}$ as reasoned there. The distribution of each of the remaining $k$ random variables $A^i$ is identical to the distribution of the corresponding $B^i$, if we further condition on the number of ones in $B^i$ to be in the range $tq\pm C\sqrt{tq \log(nm)}$. The probability that the number of ones in $B^i$ is not in this range is the probability that a $Bin(t-k,q)$ random variable is not in the range $tq-k\pm C\sqrt{tq\log(nm)}$, which, by %
    a Chernoff bound, is again at most $1/(nm)^{10}$. Thus, the TV distance between $D_y^{I,b,\calT,r,S,R}$ and $D_x^{I,b,\calT,r,S,R}$ is again at most $M \cdot (1/(nm)^{10}) \le 1/(nm)^9$. 

    Summarily, we have shown that the distribution of $y^1,\dots,y^m$, in the case that $z^1,\dots,z^m$ were drawn from $D_1^\calT$ in \Cref{prob:partition}, is $o(1)$ close in TV distance to the distribution $\Dp$ in \Cref{prob:biclique}.

    To conclude, the analysis in Cases 1 and 2 above shows that if $\mathcal{A}$ distinguishes between $\Du$ and $\Dp$ with advantage $0.9$ using only $o\left(\frac{nmq}{pk^4\log(nm)}\right)=o\left(\frac{mn'}{pk^2t}\right)$ bits of memory, $\mathcal{A}'$ distinguishes between $D_0$ and $D_1^\calT$ with advantage $0.89$ using (asymptotically) the same amount of memory. This contradicts \Cref{lem:partition}. 
\end{proof}

%% file: appendix_semi_random.tex
\section{Proofs from Section \ref{sec:semi-random-planted-biclique}}
\label{sec:appendix-semi-random}

We restate and prove \Cref{thm:pattern-planted} below.
\thmpatternplanted*
\begin{proof}[Proof of \Cref{thm:pattern-planted}]

    First, we consider the following distinguishing problem:

    \begin{problem}
        \label{prob:pattern-planted-partition}
        Let $0<k, n'\le n$. Let $t$ divide $n'$, $\calT=\{T_{r}\}_{r\in[n'/t]}$ be a partition of $[n']$, where $\forall r, |T_r|=t$. The goal is to distinguish between the following joint distributions on $n'$-dimensional vectors $z^1,\ldots, z^{n}$:
        \begin{enumerate}
        \item $\no$: $\forall i\in[n]$ and $\forall r'\in[n'/t]$, $z^i_{T_{r'}}$ is drawn uniformly at random from $\{0,1\}^t$. 
        
        \item $\yes^\calT$: Pick $r$ uniformly from $[n'/t]$. $\forall i\in[n]$ and $\forall r'\neq r$, $z^i_{T_{r'}}$ is drawn uniformly at random from $\{0,1\}^t$. 
        
        $R$ is drawn uniformly at random from all subsets of $[n]$ of size $k$. Draw $v$ uniformly at random from $\{0,1\}^t$.
        
        $\forall i\not \in R$, $z^i_{T_r}$ is drawn uniformly at random from $\{0,1\}^t$. Whereas, $\forall i\in R$, $z^i_{T_r}$ is set to $v$.
        \end{enumerate}
    \end{problem}

    We first note that \Cref{prob:pattern-planted-partition} is a specific instantiation of \Cref{prob:general}, with  $n=n$ and $d=n'$. Furthermore, for the purposes of instantiating \Cref{lem:framework}, if we denote by $\mu_v$ a point mass on a vector $v \in \{0,1\}^t$, we have that $\mu_1=\Eb_v\mu_v=\mu_0$. Thus, \Cref{lem:framework} guarantees that any algorithm, which is allowed to use public randomness, that distinguishes between $D_0$ and $D_1^\calT$ above requires at least $\Omega\left(\frac{nn'}{pk^2t}\right)$ bits of memory.
    
    Now, suppose there was a $p$-pass streaming algorithm $\mathcal{A}$ that solves \Cref{prob:fixed_pattern} using only $o\left(\frac{n^2}{pk^3}\right)$ bits of memory. Our approach will be to use the existence of $\mathcal{A}$ to construct a $p$-pass streaming algorithm $\mathcal{A}$' that circumvents the lower bound for \Cref{prob:pattern-planted-partition}, yielding a contradiction. 
    
    For this, let $t=k$ and $n'=t \cdot \left\lfloor \frac{n}{t}\right\rfloor$. The algorithm $\mathcal{A}'$ operates as follows: First, using public randomness, it draws $I \subseteq [n]$ of size $n-n'$ uniformly at random, and then draws $b^{1},\ldots,b^{n} \in \{0,1\}^{n-n'}$, where every $b^i$ is sampled uniformly at random from $\{0,1\}^{n-n'}$. Then, using public randomness again, it draws a uniformly random permutation $\pi$ of $[n']$. Upon receiving a stream $z^1,\dots,z^n$ of $n'$-bit vectors from an instance of \Cref{prob:pattern-planted-partition}, algorithm $\mathcal{A'}$ translates this stream into a stream of $n$-bit vectors $y^1,\dots,y^n$. Namely, $y^i$ is constructed from $z^i$ as follows: First, $y^i_I$ is assigned to be $b^i$. Then, $z^i$ is permuted according to $\pi$, yielding $z^i_\pi$. Finally, $y^i_{[n] \setminus I}$ is assigned to be $z^i_\pi$. Observe that $\mathcal{A}'$ can construct this stream $y^1,\dots,y^n$ using only a constant memory overhead (since the public randomness does not contribute to the memory requirement). The algorithm $\mathcal{A}'$ then feeds this stream $y^1,\dots,y^n$ to $\mathcal{A}$, and returns the output of $\mathcal{A}$. The total memory requirement of $\mathcal{A}'$ and $\mathcal{A}$ is thus the same, upto an additive constant.

    We will now argue that $\mathcal{A}'$ correctly solves \Cref{prob:partition}.

    \paragraph{Case 1:} First, consider the case that $z^1,\dots,z^n$ were draws from $D_0$ in \Cref{prob:pattern-planted-partition}. The distribution of $y^1,\dots,y^n$ constructed by $\mathcal{A}'$, since $\pi$ is a random permutation, is equivalent to the following:
    \begin{enumerate}
        \item[(1)] Draw a subset $I \subseteq [n]$ of size $n-n'$ uniformly at random. 
        \item[(2)] Draw $b^{1},\ldots,b^{n} \in \{0,1\}^{n-n'}$, where every $b^i$ is sampled uniformly at random from $\{0,1\}^{n-n'}$. 
        \item[(3)] Set $y^i_I = b^i$ for every $i \in [n]$. 
        \item[(4)] Draw a uniformly random partition $\calT=\{T_{r'}\}_{r'\in[n'/t]}$ of $[n] \setminus I$, where $\forall r', |T_{r'}|=t$. 
        \noindent
        \tikz{\draw[dashed] (0,0) -- (\linewidth,0);}
        \item[(5)] For every $i \in [n], r' \in [n'/t]$, draw $y^i_{T_{r'}}$ uniformly at random from $\{0,1\}^t$.        
    \end{enumerate}

    Observe that the distribution of each $y_i$ thus drawn is simply the uniform distribution over $\{0,1\}^n$. %

    \paragraph{Case 2:} We will now reason about the distribution of $y^1,\dots,y^n$ above when $z^1,\dots,z^n$ were draws from $D_1^\calT$ in \Cref{prob:pattern-planted-partition}. The distribution of $y^1,\dots,y^n$ that $\mathcal{A}'$ constructs can then be described by the random process comprising of Steps (1)-(4) above in Case 1, followed by the steps ahead:
    \begin{enumerate}
        \item[(5)] Draw $r$ uniformly at random from $[n'/t]$, $v$ uniformly at random from $\{0,1\}^t$, a subset $R \subseteq [n]$ uniformly at random of size $k$.
        \item[(6)] For every $i \notin R$, $r' \in [n'/t]$, draw $y^i_{T_{r'}}$ uniformly at random from $\{0,1\}^t$. 
        \item[(7)] For every $i \in R$, set $y^i_{T_r}=v$, and for every $r' \neq r$, draw $y^i_{T_{r'}}$ uniformly at random from $\{0,1\}^t$. 
    \end{enumerate}

    Observe that $y_1,\dots,y_n$ thus drawn is identical in distribution to a draw from $\Dp$ in \Cref{prob:fixed_pattern}.

    Summarily, we have shown that the the distribution of $y^1,\dots,y^n$, in the case that $z^1,\dots,z^n$ were drawn from $D_0$ in \Cref{prob:pattern-planted-partition}, is identical to the distribution $\Du$ in \Cref{prob:fixed_pattern}. Similarly, we have also shown that the distribution of $y^1,\dots,y^n$, in the case that $z^1,\dots,z^n$ were drawn from $D_1^\calT$ in \Cref{prob:pattern-planted-partition}, is identical to the distribution $\Dp$ in \Cref{prob:fixed_pattern}. So, it follows that if $\mathcal{A}$ distinguishes between $\Du$ and $\Dp$ with advantage $0.9$ using only $o\left(\frac{n^2}{pk^3}\right)=o\left(\frac{nn'}{pk^2t}\right)$ bits of memory, $\mathcal{A}'$ distinguishes between $D_0$ and $D_1^\calT$ with advantage $0.9$ using the same amount of memory (asymptotically). This contradicts the lower bound that we derived in the first paragraph. 
\end{proof}

%% file: appendix-gaussian.tex
\section{Proofs from Section \ref{sec:gaussians}}
\label{sec:appendix-gaussians}

We first restate and prove Lemma \ref{lem:chirag}
\lemmachirag*

\begin{proof}[Proof of \Cref{lem:chirag}]
    In this proof, we will construct a reduction between \Cref{prob:gaussian_mean_general} and \Cref{prob:partition_mean_general}.

    Consider first the intermediate Problem A of distinguishing between:
    \begin{enumerate}
        \item $\no$ (no instance): $\forall i\in[n]$, $x^i$ is drawn from $N(0,I_d)$.
        \item $\yes$ (yes instance): Draw $s \sim Bin(t, \sparsity/t)$. Draw $S \subseteq [d]$ of size $s$. Obtain $v \in \mathbb{R}^d$, where $v_j=\alpha$ for every $j \in S$, and $v_j=0$ otherwise. Draw $R \subseteq [n]$ uniformly at random of size $nq$. For every $i \in [n] \setminus R$, $x^i \sim N(0, I_d)$, whereas for every $i \in R$, $x^i \sim N(v, I_d)$.
    \end{enumerate}

    The main reason to introduce Problem A above is to get rid of the partition in \Cref{prob:partition_mean_general}---observe that there is no notion of such a ``partition'' in \Cref{prob:gaussian_mean_general}. We will first relate the hardness of Problem A to \Cref{prob:partition_mean_general}. 
    
    Suppose there were a $p$-pass algorithm $\mathcal{A}$ that uses $s$ bits of memory and $n$ samples to solve Problem A with probability 0.97. We will show that there exists a $p$-pass algorithm $\mathcal{A}'$ that uses $s+\tilde{O}(1)$ bits of memory and $n$ samples to solve \Cref{prob:partition_mean_general} with probability 0.97. The algorithm $\mathcal{A}'$ operates on the input of \Cref{prob:partition_mean_general} as follows. First, using public randomness, it draws a uniformly random permutation $\pi$ of $[d]$. Upon receiving the stream $x^1,\dots,x^{n}$, it permutes each of $x^1,\dots,x^{n}$ according to $\pi$, and feeds $x^1_\pi,\dots,x^{n}_\pi$ to $\mathcal{A}$. Observe that if $x^1,\dots,x^{n}$ were drawn from the no instance of \Cref{prob:partition_mean_general}, then $x^1_\pi,\dots,x^{n}_\pi \sim \no$ above, and if they were drawn from the yes instance, $x^1_\pi,\dots,x^{n}_\pi \sim \yes$ above.\footnote{Recall that we assumed for convenience that $t$ divides $d$. We can handle the case of $t$ not dividing $d$ similarly as we did in the proof of \Cref{thm:main}. That is, we may instead consider \Cref{prob:partition_mean_general} with $d'=t \cdot \lfloor d/t \rfloor$. In order to prepare $d$-dimensional inputs to Problem A from $d'$-dimensional inputs of \Cref{prob:partition_mean_general}, $\mathcal{A}$' can first draw $n$ i.i.d. vectors from $N(0,I_{d-d'})$, and assign these at coordinates corresponding to a uniformly random subset of $[d]$ of size $d-d'$. The rest of the $d'$ coordinates may then be assigned to be $x^1_\pi, \dots, x^n_\pi$. This generates an input for Problem A.} Therefore, $\mathcal{A}'$ can simply return the output of $\mathcal{A}$, and solve \Cref{prob:partition_mean_general}. %

    In the yes instance of Problem A above, there are a fixed number $nq$ of planted vectors; however, in \Cref{prob:gaussian_mean_general}, the number of planted vectors is $nq$ only \textit{in expectation}. The next intermediate problem bridges this. Concretely, consider Problem B of distinguishing between:
    \begin{enumerate}
        \item $\no$ (no instance): $\forall i\in[n/400]$, $x^i$ is drawn from $N(0,I_d)$.
        \item $\yes$ (yes instance): Draw $s \sim Bin(t, \sparsity/t)$. Draw $S \subseteq [d]$ of size $s$. Obtain $v \in \mathbb{R}^d$, where $v_j=\alpha$ for every $j \in S$, and $v_j=0$ otherwise. For every $i \in [n/400]$, $x^i \sim N(v, I_d)$ with probability $q$, and $x^i \sim N(0, I_d)$ with probability $1-q$.
    \end{enumerate}
    Suppose there were a $p$-pass algorithm $\mathcal{A}$ that uses $s$ bits of memory and $n/400$ samples to solve Problem B with probability 0.98. We will show that there exists a $p$-pass algorithm $\mathcal{A}'$ that uses $s+\tilde{O}(1)$ bits of memory and $n$ samples to solve Problem A with probability 0.97. The algorithm $\mathcal{A}'$ operates as follows. Upon receiving an input stream $x^1,\dots,x^{n}$ from Problem A, it feeds a uniformly random subset of $n/400$ of these vectors to $\mathcal{A}$. Observe first that if the input $x^1,\dots,x^{n}$ was from the no instance of Problem A, then the input given to $\mathcal{A}$ is also distributed as the no instance of Problem B. On the other hand, if the input was drawn from the yes instance of Problem A, exactly $nq$ of the vectors in the input were drawn from the planted distribution. Let $X$ denote the number of vectors drawn from the planted distribution that get included in the uniformly random subset of $n/400$ vectors that $\mathcal{A}'$ feeds to $\mathcal{A}$. Now, let $Y$ denote the number of vectors that get drawn from the planted distribution, when the input $x^1,\dots,x^{n/400}$ is drawn from the yes instance of Problem B above. Observe that the distribution of $X$ corresponds to the number of red balls drawn, when one draws $n/400$ balls uniformly at random from an urn containing $n$ (red and blue) balls of which $nq$ are red \textit{without replacement}, while the distribution of $Y$ corresponds to the number of red balls, when one draws $n/400$ balls uniformly at random from an urn containing $n$ balls of which $nq$ are red \textit{with replacement}. From Theorem (4) in \cite{diaconis1980finite}, we know that the TV distance between the distributions of $X$ and $Y$ is at most $0.01$. That is, the input that $\mathcal{A}'$ feeds to $\mathcal{A}$ comprises of $n/400$ vectors, of which a uniformly random subset of $X$ vectors are drawn from the planted distribution, whereas the input of the yes instance of Problem B corresponds to a uniformly random subset of $Y$ vectors drawn from the planted distribution, where $TV(X,Y) \le 0.01$. Summarily, we conclude that if $\mathcal{A}$ solves Problem B  with probability 0.98, $\mathcal{A}'$ solves Problem A with probability 0.97.

    Finally, we relate \Cref{prob:gaussian_mean_general} to Problem B above. Let $\mathcal{A}$ be a $p$-pass streaming algorithm that uses $s$ bits of memory and $n/400$ samples, and solves \Cref{prob:gaussian_mean_general} with probability 0.99 for every value of $l' \in [2l/3,4l/3]$. Namely, $\mathcal{A}$ processes $x^1,\ldots,x^{n/400}$ arriving in a stream, and satisfies that: 
    \begin{enumerate}
        \item[(1)] If $x^1,\ldots,x^{n/400} \sim N(0, I_d)$, then $\mathcal{A}$ outputs no with probability at least 0.99.
        \item[(2)] For \textit{every} $\sparsity' \in \left[\frac{2\sparsity}{3}, \frac{4\sparsity}{3}\right]$: if $x^1,\ldots,x^{n/400} \sim \Dp$ in \Cref{prob:gaussian_mean_general} for sparsity $\sparsity'$, then $\mathcal{A}$ outputs yes with probability at least 0.99.
    \end{enumerate}
    We will argue that $\mathcal{A}$ also solves Problem $B$. Notice that in the yes instance of Problem B, when $s \sim Bin(t,l/t)$, the probability that $s \in \left[\frac{2\sparsity}{3}, \frac{4\sparsity}{3}\right]$ is at least $1-o(1)$. Together with (2) above, we conclude that, if we simply run $\mathcal{A}$ on input $x^1,\ldots,x^{n/400}$ from Problem B above, then $\mathcal{A}$ outputs the correct answer with probability at least 0.98. This concludes the sequence of reductions. %

\end{proof}

In what follows, we will make use of the elementary claim below at multiple places.
\begin{claim}
    \label{claim:TV-for-trunc}
    Let $D$ be a distribution and let $D_{\text{trunc}}$ be the restriction of $D$ to the set $T$.
    Then,
    \[
    \| D -D_{\text{trunc}} \|_{TV} = \Pr_{x \sim D} \left [ x \notin T \right].
    \]
\end{claim}
\begin{proof}
    Let $f$ and $g$ be the probability density functions for $D$  and $D_{\text{trunc}}$ respectively. Note that the definition of the truncated distributions implies that 
    \begin{align*}
        g(x) = \begin{cases} \frac{f(x)}{\int_T f(y) \ dy} & \text{for }x \in T, \\ 0 & \text{for }x \notin T. \end{cases}
    \end{align*}
    Let $\overline T$ denote the complement of the set $T$, and let $p_T = \int_T f(y) \ dy$. Note that $p_T < 1$. Then,
    \begin{align*}
        \| D -D_{\text{trunc}} \|_{TV}
        &= \frac{1}{2} \int |f(x) - g(x)| \ dx \\
        &= \frac{1}{2} \int_{\overline T} |f(x) - g(x)| \ dx  + \frac{1}{2} \int_T |f(x) - g(x) | \ dx \\
        &= \frac{1}{2} \int_{\overline T} f(x) \ dx + \frac{1}{2} \int_T \left|f(x) - \frac{f(x)}{p_T} \right|  \ dx \\
        &=  \frac{1}{2} \int_{\overline T} f(x) \ dx + \frac{1}{2p_T}   \int_T  \left|p_T-1 \right| f(x)  \ dx \\
        &= \frac{1}{2} \int_{\overline T} f(x) \ dx + \frac{1}{2p_T}   \int_T  \left(1-p_T\right)f(x)  \ dx \\
        &= \frac{1}{2} \int_{\overline T} f(x) \ dx + \frac{p_T-p_T^2}{2p_T}
        = \frac{1}{2} \int_{\overline T} f(x) \ dx + \frac{1}{2}(1-p_T) \\
        &= \frac{1}{2} \int_{\overline T} f(x) \ dx + \frac{1}{2} \int_{\overline T} f(x) \ dx
        =\Pr_{x \sim D} \left [ x \notin T \right].
    \end{align*}
\end{proof}

We restate and prove Lemma \ref{lem:closeness_general}
\closenessgaussian*

\begin{proof}[Proof of \Cref{lem:closeness_general}]

The claim that  $\Pr_v [ v \in V_{good}] \geq 0.99$ follows from Markov's inequality, as the expected sparsity of a vector drawn from $D$ is $\sparsity$, and $V_{good}$ consists of vectors whose sparsity is at most $100 \sparsity$. The remainder of this proof is devoted to bounding the TV distances for the truncated distributions. 
Since $0 \in V_{good}$, it suffices to show that $\|\pt^{v} - N(v,I_t)\|_{TV} \leq 0.01/(nd/t)$ for an arbitrary vector $v \in V_{good}$

\paragraph{Computing $\|\pt^{v} - N(v,I_t)\|_{TV}$.}
By \Cref{claim:TV-for-trunc}, the TV distance is precisely the probability $\Pr_{x\sim N(0, I_t)} [ x \notin T]$, where $T$ is the set defined in \Cref{eqn:gaussian-truncation-good-set}. We split this probability into two terms and bound each separately. In particular, we define the set
$
T' =  \left \{ x \in \mathbb{R}^t:  \forall j \in [t], |x_j| \leq \sqrt{C_1 \log (200nd)} \right \}
$
of values $x$ with bounded coordinates (where $C_1$ is a positive constant that will be determined later). Note that
\begin{align}
    \Pr_{x \sim N(v, I_t)} [ x \notin T] 
    &= \Pr_{x \sim N(v, I_t)} [ x \notin T , x \in T'] + \Pr_{x \sim N(v, I_t)} [ x \notin T, x \notin T'] \nonumber \\  
    &\leq \Pr_{x \sim N(v, I_t) | x \in T' }[ x \notin T] + \Pr_{x \sim N(v, I_t)}[ x \notin T']. \label{eqn:two-terms}
\end{align}
The majority of our analysis is devoted to obtaining a bound for the first term. Recalling the definition of $T$, we have 
\[
\Pr_{x \sim N(v, I_t) | x \in T' }[ x \notin T] = \Pr_{x \sim N(v, I_t) | x \in T'} \left [\sum_{j=1}^t  e^{\alpha x_j} \ge te^{\alpha^2/2}+ \delta \right ],
\]
where $\delta = (C_1 \alpha) \sqrt{t}  d^{\epsilon/ 2} \log(200 nd) $.
We will apply a concentration inequality for the sum of independent, bounded random variables. First, we bound their expected sum.
\begin{claim}
    \label{claim:mean-small}
    For any vector $v \in V_\text{good}$,
    \begin{align*}
     \mathbb E_{x \sim N(v, I_t) | x \in T'} \left [\sum_{j=1}^t e^{\alpha x_j} \right ]  \leq te^{ \alpha^2 /2} +  \delta / 2.
    \end{align*}
\end{claim}
\begin{proof}
Let $B = \sqrt{ C_1 \log (200 nd)}$.
Since conditioning on the set $T'$ preserves independence between the coordinates of $x$, we first 
derive an upper bound for the following quantity
\begin{align*}
\mathbb E_{X \sim N(v_j,1 )} \left[e^{\alpha X} ~\big |~ |X| \leq B\right]
&= \frac{\int_{-B}^{B} e^{\alpha z} \cdot \frac{1}{\sqrt{2\pi}}e^{-(z - v_j)^2/2} dz}{\Pr_{X \sim N(v_j, 1)} [ |X| \leq B ] },
\end{align*}
We now bound the numerator. Observe that
\begin{align*}
\int_{-B}^{B} e^{\alpha z} \cdot \frac{1}{\sqrt{2\pi}}e^{-(z - v_j)^2/2} dz
&= \int_{-B}^{B} e^{\alpha^2 / 2 + v_j \alpha} \cdot \frac{1}{\sqrt{2\pi}}e^{-(z - (\alpha + v_j))^2 / 2} dz
\\
&= e^{\alpha^2 / 2 + v_j \alpha} \left ( \Pr_{X \sim N(\alpha + v_j, 1)}[ X \leq B] - \Pr_{X \sim N(\alpha +v_j, 1)}[ X \leq -B]\right ) \\
&= e^{\alpha^2 / 2 + v_j \alpha} \left ( \Pr_{X \sim N( 0, 1)}[ X   \leq B - \alpha - v_j ] - \Pr_{X \sim N(0, 1)}[ X \leq -B - \alpha - v_j ]\right ) \\
&=  e^{\alpha^2 / 2 + v_j \alpha} \left ( \Phi( B - \alpha - v_j ) - \Phi( -B - \alpha - v_j ) \right ), 
\end{align*}
where $\Phi(\cdot)$ denotes the cumulative distribution function of a $N(0, 1)$ random variable.  
Consider the function $f(y) = \Phi ( B - y ) - \Phi( -B - y )$, and observe that %
\begin{align*}
f'(y) 
= - \phi(B - y) + \phi(- B - y)
= \phi(B + y ) - \phi (B - y),
\end{align*}
where $\phi$ denotes the probability density function of a $N(0, 1)$ random variable. The second equality follows from the symmetry of the density function $\phi$. We claim that $f'(y) \leq 0$ for $y \ge 0$, and hence $f$ is non-increasing when $y \geq 0$. Indeed, if $y \leq B$, then $B+y \geq B- y \geq 0$ and since $\phi$ is non-increasing for non-negative arguments, we in turn have $\phi(B - y) \geq \phi(B + y)$. On the other hand, if $y > B$, then $- B - y \leq B - y < 0$ and since $\phi$ is increasing for negative arguments, we in turn have $\phi(B - y) > \phi(-B - y)$.
Thus, $f$ is non-increasing for non-negative arguments. Since $\alpha \geq 0$ and $v_j \ge 0$, we have $f(\alpha + v_j) \leq f(v_j)$, which   implies that 
\begin{align*}
\Phi( B - \alpha - v_j ) - \Phi( -B - \alpha - v_j )
 & \leq \Phi(B - v_j) - \Phi(- B - v_j) \\
 & = \Pr_{X \sim N(v_j, 1)}[ |X| \leq B].
\end{align*}
It follows that our numerator is bounded above by $e^{\alpha^2 / 2 + v_j \alpha} \Pr_{X \sim N(v_j, 1)}[ |X| \leq B]$ and thus, the conditional expectation is bounded above by $e^{\alpha^2 / 2 + v_j \alpha}$. By linearity of expectation, we have that 
\[
\mathbb E_{x \sim N(v, I_t) | x \in T'} \left [\sum_{j=1}^t e^{\alpha x_j} \right ]  \leq e^{ \alpha^2 /2} \sum_{j = 1}^t e^{\alpha v_j}.
\]
It remains to upper bound the right hand side of the above inequality. Since $v \in V_{good}$, it has sparsity $\kappa \leq 100 \sparsity$ and any nonzero coordinate is by definition equal to $\alpha$. We therefore have
\begin{align*}
e^{ \alpha^2 /2}\sum_{j = 1}^t e^{\alpha v_j}
&= e^{ \alpha^2 /2} \left( t-\kappa + \kappa e^{\alpha^2} \right) \\
&= t e^{\alpha^2 /2} + \kappa e^{\alpha^2/ 2} ( e^{{\alpha^2}} -1) \\
&\le t e^{\alpha^2/2} + \kappa e^{\alpha^2/2}(2 \alpha^2) \tag{since $e^x \le 1+2x$ for $x \in (0, 1]$} \\
&\le t e^{\alpha^2/2} + 200 \sparsity  e^{\alpha^2/2}\alpha^2 \tag{since $\kappa \leq 100 \sparsity$}.
\end{align*}
Furthermore, taking $d$ sufficiently large and recalling that $t \geq (\sparsity \alpha)^2 d^{\epsilon} \log^2(200nd)$, we get
\begin{align*}
\delta / 2
&=  (C_1 \alpha ) \sqrt{t}  d^{\epsilon/ 2} \log(200 nd) / 2 \\
&\geq C_1 \sparsity \alpha^2 d^{\epsilon } \log^2 (200nd)  / 2 \\
&\geq 200 \sparsity  e^{\alpha^2/2}\alpha^2, 
\end{align*}
which proves our desired result.
\end{proof}
Next, we recall that conditioning on the set $T'$ when $x$ is drawn from an isotropic normal distribution preserves independence between the coordinates of $x$. Hence, we can apply Hoeffding's inequality to the random variables $Z_j = e^{\alpha x_j}$, where $x$ is drawn from $N(v, I_t)$ restricted to the set $T'$. Note that the preceding \Cref{claim:mean-small} implies that 
\begin{align*}
\Pr \left [\sum_{j=1}^t  Z_j \ge te^{\alpha^2/2}+ \delta \right ]
&\leq \Pr \left [\sum_{j=1}^t  Z_j -  \mathbb E \left [\sum_{j=1}^t Z_j \right ]  \ge  \delta / 2\right ].
\end{align*}
We will consider two cases based on the magnitude of $\alpha$.
First, suppose that $1/ \sqrt{C_1\log (200 nd)} \le \alpha \leq 1$. Then, we have that the variables $Z_j$ are bounded as 
\begin{align*}
    0 \leq e^{\alpha x_j} := Z_j \leq e^{\sqrt{C_1\log (200 nd)}} \leq e^{\sqrt{C_1\log (200d^{11})}} \le e^{\frac{\epsilon}{2}\log d} = d^{\epsilon/2},
\end{align*}
where we used that $n \le d^{10}$, and that $d$ is sufficiently large. 
Thus, Hoeffding's inequality gives that 
\begin{align*}
    \Pr \left [\sum_{j=1}^t  Z_j -  \mathbb E \left [\sum_{j=1}^t Z_j \right ]  \ge  \delta /2 \right ] 
    &\leq \exp \left ( - \frac{2 (\delta/2)^2}{t d^{\epsilon}} \right ) \\
    &= \exp \left ( -  \frac{C_1^2 \alpha^2 t  d^{\epsilon} \log^2(200nd)  }{ 2t d^{\epsilon} }  \right ) =  \exp \left ( -  \frac{C_1^2 \alpha^2  \log^2(200nd)  }{ 2 }  \right )\\
    &\leq \exp \left ( -  \frac{C_1  \log(200nd)  }{ 2 }  \right ) \tag{since $\alpha \ge 1/\sqrt{C_1\log(200nd)}$}\\
    & \leq 0.005 / (nd) \leq 0.005 / (nd / t).
\end{align*}
The second-to-last inequality holds whenever $C_1 \geq 2$.

Now suppose that $0 < \alpha \leq 1 / \sqrt{C_1 \log (200nd)}$. In this case, we use the following bound on the variables $Z_j$:
\begin{align*}
    e^{-\alpha \sqrt{ C_1\log (200nd)}} \leq e^{\alpha x_j}:=Z_j \leq e^{\alpha \sqrt{C_1\log (200nd)}}.
\end{align*}
Applying Hoeffding's inequality then gives that 
\begin{align*}
    \Pr \left [\sum_{j=1}^t  Z_j -  \mathbb E \left [\sum_{j=1}^t Z_j \right ]  \ge  \delta/2 \right ] 
    &\leq \exp \left ( - \frac{2 (\delta/2)^2}{t \left (e^{\alpha \sqrt{C_1\log (200nd)}} - e^{- \alpha \sqrt{ C_1\log(200nd)}} \right)^2 } \right ) \\
    &\leq \exp \left ( - \frac{ \delta^2}{2t ( 3 \alpha \sqrt{C_1\log (200nd)})^2 } \right ) \tag{$e^y - e^{-y} \leq 3y $ if $0 \leq y \leq 1$} \\
    &= \exp \left ( -  \frac{C_1^2\alpha^2td^\epsilon \log^2(200nd)}{18 t\alpha^2C_1\log(200nd)}  \right ) = \exp \left ( -  \frac{C_1d^\epsilon \log(200nd)}{18 }  \right )\\
    & \leq 0.005 / (nd)
    \leq 0.005 / (nd / t).
\end{align*}
The second-to-last inequality holds whenever $C_1 \geq 18$ and $d$ is sufficiently large.

In  both cases, we have obtained an upper bound of $0.005/ (nd/t)$ for the first term $\Pr_{x \sim N(v, I_t) | x \in T' }[ x \notin T]$ in \eqref{eqn:two-terms}. Finally, we compute an upper bound on the second term   $\Pr_{x \sim N(v, I_t)}[ x \notin T']$.
\begin{claim}
    For sufficiently large $d$, 
$$\Pr_{x \sim N(v, I_t)}[ x \notin T'] \leq  0.005/(nd/t).$$
\end{claim}

\begin{proof}
By a union bound, we note that the left hand side is at most 
\begin{align*}
    &\sum_{j=1}^t \Pr_{x_j \sim N(v_j, 1)} \left [ |x_j| \geq \sqrt{C_1 \log(200 nd)} \right] 
    \leq  2 \sum_{j=1}^t  \Pr_{x_j \sim N(v_j, 1)} \left [ x_j  \geq \sqrt{C_1 \log(200 nd)} \right ] \\
    &\leq 2 t  \Pr_{z \sim N(1, 1)} \left [ z  \geq \sqrt{C_1 \log(200 nd)} \right ] = 2 t  \Pr_{z \sim N(0, 1)} \left [ z  \geq \sqrt{C_1 \log(200 nd)}-1 \right ]\\
    &\leq  \frac{2t}{\sqrt{2 \pi} \left (\sqrt{C_1 \log(200 nd) } - 1 \right )} \exp{\left ( - \frac{  \left (\sqrt{C_1 \log(200 nd) } - 1 \right )^2 }{2} \right )} \tag{Mill's inequality}\\
    &=  \frac{\sqrt{2}t}{\sqrt{\pi} \left (\sqrt{C_1 \log(200 nd) } - 1 \right )} \exp\left ( \frac{-C_1\log(200nd)-1+2\sqrt{C_1\log(200nd)}}{2} \right ) \\
    &=\frac{\sqrt{2}t(200nd)^{-C_1/2}}{\sqrt{\pi e} \left (\sqrt{C_1 \log(200 nd) } - 1 \right )}\exp\left (\sqrt{C_1\log(200nd)} \right ) \\
    &\le \frac{t(200d^{11})^{-C_1/2}}{2 \left (\sqrt{C_1 \log(200 nd) } - 1 \right )}\exp\left (\sqrt{C_1\log(200d^{11})} \right ) \le \frac{td^{-6.5C_1+\epsilon}}{\sqrt{C_1 \log(200 nd) }} \tag{$d$ sufficiently large}\\
    &\le td^{-6.5C_1+0.01} \le 0.005/(d^{11}/t) \le 0.005/(nd/t).
\end{align*}
The inequality in the first line follows from the fact that $v_j \in [0, 1]$ (and hence the right tail has more probability mass). The second inequality follows from the fact that each $v_j \leq 1$. The final inequalities use that $d$ is sufficiently large, and that $C_1 > 5$ (say).
\end{proof}
To conclude, we have upper bounded the sum of the two terms on the right in \eqref{eqn:two-terms} by $0.01/ (nd/t)$ as desired. Note that to resolve all the dependencies on $C_1$, we can take $C_1 = 20$ (say).

\end{proof}

We now restate and prove Claim \ref{claim:gaussian-density-technical-condition}.
\gaussiandensity*
\begin{proof}[Proof of \Cref{claim:gaussian-density-technical-condition}]
    Since $\mu_0$ and $\mu_1^v$ are distributions truncated to the set $T$ defined in \Cref{eqn:gaussian-truncation-good-set}, the Gaussian densities need to normalized with the appropriate normalizing constants. However, we will first get a bound for the unnormalized densities, and then deal with the normalization.
    
    Let $f_0$ and $f_v$ respectively be the probability density functions for the (non-truncated) Gaussian distributions $N(0, I_t)$ and $N(v, I_t)$. 
    Notice that for any $x$, we have 
    \begin{align*}
    f_v(x) 
    = (2 \pi)^{-t/2}\cdot  \exp \left ( -  \frac{1}{2} (x - v)^\intercal (x - v) \right) 
    &=  (2 \pi)^{-t/2} \cdot  \exp \left ( -  \frac{1}{2} x^\intercal x  \right ) \cdot \exp \left ( x^{\intercal} v - \frac{v^\intercal v}{2} \right ) \\
    &= f_0(x)  \cdot \exp \left ( x^{\intercal} v - \frac{v^\intercal v}{2} \right ).
    \end{align*}
    
    This in turn implies that, for $v \sim D$ as defined in \Cref{eqn:v-distribution} and $x \in T$,
    \begin{align*}
    \Eb_{v \sim D} \frac{f_v(x)}{f_0(x)} 
    &= \Eb_v \exp \left ( x^{\intercal} v - \frac{v^\intercal v}{2} \right ) \\
    &= \prod_{j = 1}^{t} \Eb_{v_j} \exp \left ( x_j v_j -  \frac{v_j^2}{2} \right ) \tag{coordinates of $v \sim D$ are independent}\\
    &= \prod_{j = 1}^{t}  \left ( 1 - \frac{\sparsity}{t} +  \frac{\sparsity}{t} \cdot e^{-\alpha^2/2}   e^{\alpha x_j} \right ) \tag{$v_j$ is $\alpha$ w.p. $\sparsity/t$ and 0 otherwise} \\
    &\leq \exp \left ( -{\sparsity}\right)\cdot \exp\left(\frac{\sparsity}{t} \cdot e^{-\alpha^2/2} \sum_{j = 1}^t   e^{\alpha x_j}
    \right ) \\
    &\leq \exp \left (  -\sparsity \right ) \cdot 
    \exp \left ( \frac{\sparsity}{t} e^{-\alpha^2 / 2} \left ( te^{\alpha^2/2}+{\delta}  \right ) \right ) \tag{since $ x \in T$ } \\
    &= 
    \exp \left ( \frac{\sparsity \delta}{t} e^{-\alpha^2 / 2} \right ), 
    \end{align*}
    where $\delta = (C_1 \alpha) \sqrt{t}  d^{\epsilon/ 2} \log(200 nd) $. Since $t \ge (\alpha \sparsity)^2 d^\epsilon \log^2(200 nd)$, we further have that
    \begin{align*}
    \exp \left ( \frac{\sparsity}{t} e^{-\alpha^2 / 2} {\delta} \right )
    &= \exp \left ( \frac{\sparsity e^{-\alpha^2/2} (C_1 \alpha)  d^{\epsilon/ 2} \log(200 nd) }{\sqrt{t}} \right ) \\
    &\leq \exp \left (\frac{ \sparsity e^{-\alpha^2/2}(C_1 \alpha)  d^{\epsilon/ 2} \log(200 nd)}{(\alpha \sparsity) d^{\epsilon/2} \log(200 nd)} \right ) \leq C',
    \end{align*}
    for some constant $C'$.
    Now, by the law of total probability, we have 
    \begin{align*}
    \mathbb E_{v \sim D} \left  [ \frac{f_v}{f_0} \right ] 
    &=\Pr_v [v \in V_{good}] \cdot \mathbb E_{v | v \in V_{good}} \left  [ \frac{f_v}{f_0} \right ] 
    + \Pr_v [v \notin V_{good}] \cdot \mathbb E_{v | v \notin V_{good}} \left  [ \frac{f_v}{f_0} \right ] \\
    & \geq \Pr_v [v \in V_{good}] \cdot \mathbb E_{v | v \in V_{good}} \left  [ \frac{f_v}{f_0} \right ] = \Pr_v [v \in V_{good}] \cdot \mathbb E_{v \sim D_{good}} \left  [ \frac{f_v}{f_0} \right ].
    \end{align*}
    The inequality above follows since probability density functions are non-negative. The last equality follows since the distribution of $v$ conditioned on $v \in V_{good}$ is precisely the distribution $D_{good}$.
    
    Next, since $\mathbb{E}_{v \sim D}[\|v\|_0]=\sparsity$, we note by Markov's inequality that $\Pr_v [ v \in V_{good}] \geq 0.99$. Therefore,
    \begin{align*}
     \mathbb E_{D_{good}} \left  [ \frac{f_v}{f_0} \right ] \le 2\cdot 
    \mathbb E_{v \sim D} \left  [ \frac{f_v}{f_0} \right ] .
    \end{align*}
    
    We will now tackle our normalizing constants. Let $f_0(T)=\Pr_{x\sim N(0, I_t)}\left[ x \in T  \right]$ and $f_v(T)=\Pr_{x\sim N(v, I_t)}\left[ x \in T  \right]$. From \Cref{claim:TV-for-trunc} and \Cref{lem:closeness_general}, we know that $f_v(T) \ge 1-\frac{0.01t}{nd} \ge 0.99$, which immediately gives that $f_0(T)/ f_v(T) \le  1/ 0.99  \leq 2$. Finally,

    \begin{align*}
    \mathbb E_{v \sim D_{good}} \left  [ \frac{\mu_1^v}{\mu_0} \right ] 
    &= \mathbb E_{v \sim D_{good}} \left  [ \frac{f_v}{f_0} \cdot \frac{f_0(T)}{f_v(T)} \right ] \\
    & \leq 2\cdot \mathbb E_{v \sim D_{good}} \left [ \frac{f_v}{f_0} \right ]  \\
    &\leq 4\cdot  \mathbb E_{v \sim D} \left  [ \frac{f_v}{f_0} \right ] \\
    &\le 4C'.
    \end{align*}
    The desired result follows by taking $C=4C'$.

\end{proof}

We restate and prove \Cref{clm:gaussian-ub}.
\clmgaussianub*
\begin{proof}[Proof of \Cref{clm:gaussian-ub}]
Throughout this proof, we will let $Y_{S_1, S_2} =  \sum_{j \in S_1} \sum_{i \in S_2} x^{j}_i$ for ease of exposition. We will bound the failure probability for $\Dn$ and $\Dp$ separately.

We first bound the probability that the test fails to detect the null distribution $\Dn$.
It is straightforward to verify that when $x^1, \dots x^n \sim \Dn$, we have $Y_{S_1, S_2} \sim N(0, s_1 s_2)$ for every $s_1$-sized subset $S_1 \subseteq \{x^j\}_{j=1}^n$ and $s_2$-sized subset $S_2 \subseteq R$. . 
By a union bound, the probability that one of the test statistics exceeds $\tau$ is at most
\begin{align*}
\binom{n}{s_1} \binom{|R|}{s_2} \Pr_{Y \sim N(0, s_1 s_2)} \left [ Y \geq \tau \right ]
&\leq 
\binom{n}{s_1} \binom{|R|}{s_2}  \exp \left (   - \frac{\tau^2}{2 s_1 s_2}\right ) \\
&\leq 
\binom{n}{s_1} \binom{|R|}{s_2} \cdot  \exp \left (   - \frac{ \left ( \sqrt{2s_1 s_2 \log \left (2 \binom{n}{s_1} \binom{|R|}{s_2} / \delta \right) } \right )^2 }{2 s_1 s_2}\right ) \\
&\leq \delta / 2.
\end{align*}
Hence the failure probability for the null distribution is bounded as desired.

We now bound the probability that the test fails to detect the planted distribution $\Dp$. 
It will be convenient for us to define two events. Let $\mathcal E_1$ be the event that at least $s_1$ of the samples are drawn from the distribution $N(v, I_d)$. Let $\mathcal E_2$ be the event that $|A \cap R| \geq s_2$, where $A$ is the support of the planted vector $v$. Note that if the event $\mathcal E_1 \cap \mathcal E_2$ occurs, then there will be some pair of subsets $S_1, S_2$ that contain signal from the planted vector. If the statistic $Y_{S_1, S_2}$ exceeds the threshold, then we would correctly detect the planted distribution. It follows that the failure probability is at most 
\begin{align}\label{eq:gaussian-test-failure}
    \Pr \left [ \lnot ( \mathcal E_1 \cap \mathcal E_2 ) \right ] + 
    \Pr_{\Dp} \left [ \max_{\substack{S_1 \subseteq [n], |S_1| = s_1 \\ S_2 \subseteq R, |S_2| = s_2}}  \sum_{j \in S_1} \sum_{i \in S_2} x^{j}_i < \tau \ \Big |  \ \mathcal E_1 \cap \mathcal E_2  \right ].
\end{align}
We will bound each of these terms separately. We begin with the term $\lnot (\mathcal E_1 \cap \mathcal E_2 )$. Note that the number of samples from the distribution $N(v, I_d)$ will follow a binomial distribution $\mathrm{Bin}(n, q)$ (and have mean $nq$). Hence, we can show that 
\begin{align*}
\Pr[ \lnot \mathcal E_1 ] 
&\leq 
\Pr_{Z \sim \mathrm{Bin}(n, q)} \left [ Z \leq s_1 \right] \\
&\leq 
\Pr_{Z \sim \mathrm{Bin}(n, q)} \left [ Z \leq \frac{1}{2} nq \right] \tag{since $nq \geq  2 s_1$}\\
&\leq  \exp \left ( - \frac{(nq) (1/2)^2}{2} \right ) \tag{Chernoff bound} \\
&\leq  \exp \left ( - \frac{C_{\delta, \alpha}  \log (nd)}{4} \right ) \tag{since $nq \geq C_{\delta, \alpha} \log(nd)$  } \\
&\leq (nd)^{-2} \tag{since $C_{\delta, \alpha} \geq 8$} \\
&\leq \delta / 4  \tag{for sufficiently large $n,d$}
\end{align*}
Next, note that the number of coordinates in the intersection $|A \cap R|$ will follow the hypergeometric distribution $\mathrm{Hypergeometric}(d, \sparsity, |R|)$ (and have mean $|R| ( \sparsity / d) = 2  C_{\delta, \alpha} \log(nd/ \delta) \log (nd)$). Hence, we can show that 
\begin{align*}
    \Pr[ \lnot \mathcal E_2] 
    &\leq \Pr_{Z \sim \mathrm{Hyp}(d, \sparsity, |R|)} [ Z \leq s_2] \\
    &\leq \Pr_{Z \sim \mathrm{Bin}(|R|, \sparsity/ d)} \left [ Z \leq s_2 \right ]\\
    &= \Pr_{Z \sim \mathrm{Bin}(|R|, \sparsity/ d)} \left [ Z \leq \frac{1}{2\log(nd/ \delta)} \cdot |R|(\sparsity/ d) \right ] \\
    &\leq \Pr_{Z \sim \mathrm{Bin}(|R|, \sparsity/ d)} \left [ Z \leq \frac{1}{2} \cdot |R|(\sparsity/ d) \right ] \tag{ for sufficiently large $n, d$} \\
    &\leq \exp \left ( - \frac{|R| (\sparsity / d)(1/2)^2 }{2} \right ) \tag{Chernoff bound}\\
    &= \exp \left( - \frac{C_{\delta, \alpha} \log(nd) \log(nd/ \delta)}{4} \right) \\
    &\leq \exp \left( - \frac{C_{\delta, \alpha} \log(nd) }{4} \right) \tag{for sufficiently large $n, d$} \\
    &\leq \delta / 4
\end{align*}
In the second line we made use of the well-known fact that the binomial distribution stochastically dominates the hypergeomtric distribution. The final inequality follows from observing a similar expression in the calculation for $\Pr[\lnot \mathcal E_1]$.

Note that by a union bound we have $ \Pr \left [ \lnot ( \mathcal E_1 \cap \mathcal E_2 ) \right ] \leq \delta / 2$. It remains to show that the second term in \Cref{eq:gaussian-test-failure} is also upper bounded by $\delta / 2$. Recall that if the event $\mathcal E_1 \cap \mathcal E_2$ occurs, then then there will be some pair of subsets $S_1, S_2$ such that every coordinate contains signal from the planted vector. Note also that sum of these entries follows the Gaussian distribution $N(\alpha s_1 s_2, s_1 s_2)$.
It is not hard to see that
\begin{align*}
    \Pr_{\Dp} \left [ \max_{\substack{S_1 \subseteq [n], |S_1| = s_1 \\ S_2 \subseteq R, |S_2| = s_2}}  \sum_{j \in S_1} \sum_{i \in S_2} x^{j}_i < \tau \ \Big |  \ \mathcal E_1 \cap \mathcal E_2  \right ] 
    \leq \Pr_{Z \sim N(\alpha s_1 s_2, s_1s_2)} \left [ Z < \tau \right ] 
    =  \Pr_{Z \sim N(0, 1)} \left [ Z < \frac{\tau - s_1 s_2 \alpha}{\sqrt{s_1 s_2}}\right ]. 
\end{align*}
Let $t = \sqrt{2 \log (4 / \delta) }$. It suffices to show that $s_1 s_2 \alpha \geq \tau + t \sqrt{s_1 s_2}$. Indeed, the application of standard Gaussian tail bound would imply that the second term in \Cref{eq:gaussian-test-failure} is at most $\delta / 4$. This in turn would confirm that the total failure probability is at most $3 \delta / 4$.

We now show through a series of calculations that the desired inequality holds. We will in fact show that 
\begin{align*}
    \alpha^2 \geq  2 \left ( \frac{\tau}{s_1 s_2} \right)^2 + 2 \left ( \frac{t}{ \sqrt{s_1 s_2}} \right )^2 \geq \left (\frac{\tau + t \sqrt{s_1 s_2}}{s_1 s_2} \right )^2
\end{align*}
Note that the second inequality follows from the fact that $2a^2 + 2b^2 \geq (a + b)^2$, so we only need to establish the first inequality. 

Let $s = s_1 = s_2$ and $N = \binom{n}{s_1} \binom{|R|}{s_2}$. 
Since $t^2 = 2 \log(4/\delta)$, we can rewrite the constant $C_{\delta, \alpha} = \frac{8 + 2 t^2}{\alpha^2}$. Equivalently, we have 
$$\alpha^2 = \frac{8 + 2 t^2}{C_{\delta, \alpha}} = \frac{4}{C_{\delta, \alpha}} + \frac{4 + 2 t^2}{C_{\delta, \alpha}}.$$
We will show that this value of $\alpha$ is in fact sufficient. That is, we will show that
$$ \alpha^2 \geq \frac{2\tau^2}{s^4} + \frac{2t^2}{s^2} $$
Substituting the definition of $\tau^2 = 2s^2\log(2N/\delta)$ into the expression on the right hand side of our target inequality, we can derive that 
\begin{align*}
\frac{2\tau^2}{s^4} + \frac{2t^2}{s^2}
&= \frac{2\left(2s^2 \log(2N/\delta)\right)}{s^4} + \frac{2t^2}{s^2}  \\
&= \frac{4\log(N)}{s^2} + \frac{4\log(2/\delta) + 2t^2}{s^2} \\
&\leq \frac{4 \left (\log \binom{n}{s} + \log \binom{|R|}{s} \right )}{s^2} + \frac{4\log(2/\delta) + 2t^2}{s^2} \\
&\leq \frac{4 \log (n |R|)}{s} + \frac{4\log(2/\delta) + 2t^2}{s^2} \\
&\leq \frac{4 \log (nd) (1 + o(1))}{s} + \frac{4\log(2/\delta) + 2t^2}{s^2} \\
&=\frac{4 (1 + o(1))}{C_{\delta, \alpha}} + \frac{4\log(2/\delta) + 2t^2}{s^2} \\
&= \frac{4}{C_{\delta, \alpha}} + \left ( \frac{4\cdot o(1)}{C_\delta, \alpha} + \frac{4\log(2/\delta) + 2t^2}{s^2} \right )
\end{align*}
We note that for sufficiently large $n,d$ the second term in the above expression tends to 0 and in particular is less than the constant $\frac{4 + 2 t^2}{C_{\delta, \alpha}}$. Thus, our value of $\alpha$ is indeed sufficient, as desired.

\end{proof}

%% file: appendix-pca.tex
\section{Proofs from Section \ref{sec:pca}}
\label{sec:appendix-pca}

In what follows, we will make use of the following observation at multiple points.

\begin{claim}\label{claim:pca-Y-gaussian}
For any distribution $D \in \{ N(0, I_t), N(0, \Sigma_S) \}$, where $\Sigma_S = I_t + \alpha v v^\intercal$ and $v =  \frac{1}{\sqrt{\sparsity}} 1_S$, and any set $R \in \mathcal S = \{[1, \sparsity], [\sparsity + 1, 2 \sparsity], \dots, [t - \sparsity + 1, t]  \}$,
if $x \sim D$, the random variable $Y_R = x^\intercal 1_R$ follows a Gaussian distribution with mean $\mathbb E[Y_R] = 0$ and variance 
\begin{align*}
\sigma^2 = 1_R^\intercal \Covar(x, x) 1_R
= 
\begin{cases}
\sparsity  & x \sim N(0, I_t) \\
\sparsity & x \sim N(0, \Sigma_S) , \quad R \neq S \\
(1 + \alpha)\sparsity  & x \sim N(0, \Sigma_S), \quad R = S
\end{cases}.
\end{align*}
Moreover, $\mathbb E [Y_{R}^2]  = \sigma^2 + \mathbb E[Y_{R}]^2 = \sigma^2$.
\end{claim}

\newcommand{\BBdelta}{
d^{\epsilon/2}\sqrt{(t / \sparsity) \log(400 nd )}}

\newcommand{\BB}{\sqrt{2(1 + \alpha) \sparsity \log (400 nd)}}

We first restate and prove \Cref{lem:closeness_pca}.
\closenesspca*
\begin{proof}[Proof of \Cref{lem:closeness_pca}]
We proceed in a similar way as in the proof of  \Cref{lem:closeness_general}. For each distribution $D \in \{ N(0, I_t), N(0, \Sigma_S)  \}$, Claim \ref{claim:TV-for-trunc} tells us  that the TV distance is precisely the probability $\Pr_{x\sim D} [ x \notin T]$, where $T$ is the set defined in \Cref{eqn:pca-truncation-good-set}. We split this probability into two terms and bound each separately. In particular, we define the set
$$
T' =  \left \{ x \in \mathbb{R}^t:  |x^\intercal 1_R| \leq \BB \quad  \forall R \in \mathcal S
\right \}
$$
of values $x$ whose sums over $\sparsity$-sized blocks are bounded. Note that the variables $Y_R = x^\intercal 1_R$ are mutually independent when $x \sim D$ and that further conditioning on $T'$ preserves independence between those blocks.
Note also that 
\begin{align}
    \Pr_{x \sim D} [ x \notin T] 
    &= \Pr_{x \sim D} [ x \notin T , x \in T'] + \Pr_{x \sim D} [ x \notin T, x \notin T'] \nonumber \\  
    &\leq \Pr_{x \sim D | x \in T' }[ x \notin T] + \Pr_{x \sim D}[ x \notin T']. \label{eqn:two-terms-pca}
\end{align}
The majority of our analysis is devoted to obtaining a bound for the first term. Recalling the definition of $T$, we have 
\[
\Pr_{x \sim D | x \in T' }[ x \notin T] = \Pr_{x \sim D | x \in T'} \left [  \sum_{R}  \exp \left (  \frac{\alpha}{2(\alpha + 1)} \cdot \frac{1}{\sparsity} (x^\intercal 1_R)^2 \right )  > 
    (t / \sparsity) (1- \alpha)^{-1/2}  +
    \delta  \right ],
\]
where $\delta = \BBdelta$.
For this term, we will apply a concentration inequality for the sum of independent, bounded random variables. For the second term, we will use a standard tail bound. 
We begin by bounding the first term. 
We will make use of the following claim, which bounds the expectation of the key term in our analysis.
\begin{claim}
\label{claim:mean-small-pca}
For any distribution $D \in \{ N(0, I_t), N(0, \Sigma_S) \}$ and any set 
$R \in \mathcal S$,
\begin{align*}
     \mathbb E_{x \sim D | x \in T'} \left [\exp \left (  \frac{\alpha}{2(\alpha + 1)} \cdot \frac{1}{\sparsity} (x^\intercal 1_R)^2 \right ) \right ]
    &\leq  (1- \alpha)^{-1/2}.
\end{align*}
\end{claim}
\begin{proof}
We consider the random variable $Y_R = x^\intercal 1_R$, where $x \sim D$.
Let $c = \frac{\alpha}{2(\alpha + 1)} \cdot \frac{1}{\sparsity}$. First note that we can rewrite the left hand side of our target expression as follows
\begin{align*}
 &\mathbb E \left [ \exp( c Y_R^2 ) \ \big | \ |Y_A| \leq \BB \quad  \forall A \in \mathcal S
 \right ] \\
 ={}&
  \mathbb E \left [ \exp( c Y_R^2 ) \ \big | \ |Y_R| \leq \BB  \right ] \tag{by independence of blocks}.
\end{align*}
Next, observe that 
\begin{align*}
&\mathbb E \left  [ \exp( c Y_R^2 ) \right ] \\
={}&\mathbb E \left [ \exp( c Y_R^2 ) \ \big | \ |Y_R| > \BB  \right ] \left (  \Pr \left[ |Y_R| > \BB \right]  \right ) + \\
 {}& \mathbb E \left [ \exp( c Y_R^2 ) \ \big | \ |Y_R| \leq \BB  \right ] \left (  \Pr \left[ |Y_R| \leq \BB \right]  \right ) \\
\geq{}&\mathbb E \left [ \exp( c Y_R^2 ) \ \big | \ |Y_R| \leq \BB  \right ] \tag{since $e^{cy^2}$ is monotone}.
\end{align*}
Thus, it suffices to find an upper bound for $\mathbb E [\exp(cY_R^2)]$.
By Claim \ref{claim:pca-Y-gaussian}, we know that $Y_R$ follows a Gaussian distribution. Therefore, by standard properties of the Gaussian distribution we can show that if $c < \frac{1}{2 \sigma^2}$, then
\begin{align*}
\mathbb E_{Y_R} \left [ \exp(c Y_R^2) \right ]
&= \int_{-\infty}^\infty \frac{1}{\sqrt{2 \pi \sigma^2}} \exp\left ( -\frac{y^2}{2 \sigma^2} \right ) \cdot \exp ( cy^2) \ dy \\
&= \frac{1}{\sqrt{2 \pi \sigma^2}} \cdot  \int_{-\infty}^\infty \exp \left (  -y^2 \left ( \frac{1}{2 \sigma^2} - c \right ) \right ) \ dy \\
&=  \frac{1}{\sqrt{2 \pi \sigma^2}}  \cdot \sqrt{ \frac{\pi}{ \frac{1}{2 \sigma^2} - c } } \\
&= \sqrt{\frac{1}{1 - 2 \sigma^2 c}}.
\end{align*}
Since $c = \frac{\alpha}{2(\alpha + 1)} \cdot \frac{1}{\sparsity} $ and $\alpha < 1$, the required condition on $c$ holds and thus, we have 
\begin{align}\label{eq:pca-expectation-ub}
\mathbb E_{Y_R} \left [ \exp(c Y_R^2) \right ] = 
\begin{cases}
(1 + \alpha)^{1/2} & \sigma^2 = \sparsity \\
(1 - \alpha)^{-1/2} & \sigma^2 = (1 + \alpha) \sparsity.
\end{cases}
\end{align}
Note that since $\alpha < 1$, we have $(1 + \alpha)^{1/2} \leq (1 - \alpha)^{-1/2} $.
Our desired result immediately follows. 

\end{proof}

Next, we recall that conditioning on the set $T'$ when $x$ is drawn from a distribution $D \in \{ N(0, I_t), N(0, \Sigma_S)  \}$ preserves independence between the $\sparsity$-sized blocks of coordinates of $x$. Hence, we can apply Hoeffding's inequality to the random variables 
$$
Z_R =\exp \left (  \frac{\alpha}{2(\alpha + 1)} \cdot \frac{1}{\sparsity} (x^\intercal 1_R)^2 \right ),
$$
where $x$ is drawn from the appropriate Gaussian distribution $D$ further truncated on the set $T'$.
The random variable is clearly bounded as we show below:
\begin{align*}
0 \leq Z_R 
\leq \exp \left (  \frac{\alpha}{2(\alpha + 1)} \cdot \frac{1}{\sparsity} \left (\BB \right )^2 \right ) 
=(400nd)^{\alpha} 
\leq (400d^{11})^\alpha.
\end{align*}
The final inequality follows from the fact that $n \leq d^{10}$.
Now, notice  that the preceding \Cref{claim:mean-small-pca} implies that 
for $\delta = \BBdelta$,
\begin{align*}
\Pr \left [\sum_{R}  Z_R \ge (t / \sparsity) (1- \alpha)^{-1/2} + \delta \right ]
&\leq \Pr \left [\sum_{R}  Z_R -  \mathbb E \left [\sum_{R} Z_R \right ]  \ge  \delta \right ] \\
&\leq \exp \left ( - \frac{2   \left (\BBdelta \right )^2 }{ (t/ \sparsity)  (400d^{11})^{2\alpha}} \right ) \\
&= \exp \left ( - \frac{ d^{\epsilon - 22 \alpha}  }{400^{2\alpha}} \cdot 2\log (400 nd) \right ) \\
&\leq  (1/(400nd))^2  \tag{since $\alpha < \epsilon / 22$, and $d$ is sufficiently large} \\ 
&\leq 0.005 / (nd/t).
\end{align*}
The fourth line follows from the fact that $\alpha < \epsilon / 22$ and taking $d$ sufficiently large. Thus, we have obtained an upper bound of $0.005/ (nd/t)$ for the first term $\Pr_{x \sim D | x \in T' }[ x \notin T]$ in \eqref{eqn:two-terms-pca}. Finally, we compute an upper bound on the second term   $\Pr_{x \sim D}[ x \notin T']$ for each distribution $D \in \{ N(0, I_t), N(0, \Sigma_S)  \}$.
\begin{claim}
For each distribution $D \in \{ N(0, I_t), N(0, \Sigma_S)  \}$, 
$$\Pr_{x \sim D}[ x \notin T'] \leq  0.005/(nd/t).$$
\end{claim}
\begin{proof}
By a union bound over the sets $R \in \{ [1, \sparsity], [\sparsity + 1, 2 \sparsity], \dots, [t - \sparsity + 1, t]\}$, the left side is at most 
$
\sum_R \Pr_{x}[ |x^\intercal 1_R | \geq \BB ]
$.
We again consider the random variable $Y_R = x^\intercal 1_R$ and recall Claim \ref{claim:pca-Y-gaussian}.
Hence, it follows that for each distribution $D \in \{ N(0, I_t), N(0, \Sigma_S)  \}$, we have 
\begin{align*}
\sum_R \Pr_{x \sim D}\left [ |x^\intercal 1_R | \geq \BB \right ]
&\leq (t / \sparsity) \Pr_{Y \sim N(0,(1 + \alpha)\sparsity )} \left [ | Y | \geq \BB \right ] \\
&\leq (t / \sparsity) \exp  \left ( - \frac{\left(\BB \right)^2}{2(1 + \alpha) \sparsity} \right )  \\
&= (t / \sparsity) \exp ( - \log (400 nd) ) \\
&\leq 0.005/ ( nd (\sparsity / t)) \\
&\leq 0.005/ ( nd / t ).
\end{align*}
\end{proof}
To conclude, we have bounded the sum of the two terms  in \eqref{eqn:two-terms-pca} by $0.01/ (nd/t)$ as desired.
\end{proof}

We now restate and prove \Cref{claim:pca-density-technical-condition} below.
\pcadensity*
\begin{proof}[Proof of \Cref{claim:pca-density-technical-condition}]
Let $f_S(x)$ be the probability density function for  $N(0, \Sigma_S)$.
We first apply standard matrix identities to the matrix $\Sigma_S = I_d + \alpha  v v^\intercal$, where $v = \frac{1}{\sqrt{\sparsity}}1_S$, to derive that
\begin{align*}
f_S(x) 
&= (2 \pi)^{-t/2} |\Sigma_S|^{-1/2} \exp \left (  -\frac{1}{2} x^\intercal \Sigma_S^{-1}x  \right ) \\
&= (2 \pi)^{-t/2} |\Sigma_S|^{-1/2} \exp \left (  -\frac{1}{2} x^\intercal  \left (  I_d - \frac{\alpha v v^\intercal }{\alpha + 1}  \right ) x  \right )  \tag{Sherman-Morrison identity} \\
&= (2 \pi)^{-t/2} (1 + \alpha)^{-1/2} \exp \left (  -\frac{1}{2} x^\intercal  \left (  I_d - \frac{\alpha v v^\intercal}{\alpha + 1}  \right ) x  \right )  \tag{Matrix-determinant lemma} \\ 
&= f_0(x) (1 + \alpha)^{-1/2} \exp \left (  \frac{\alpha}{2(\alpha + 1)} (x^\intercal v)^2 \right ) \\
&= f_0(x) (1 + \alpha)^{-1/2} \exp \left (  \frac{\alpha}{2(\alpha + 1)} \cdot \frac{1}{\sparsity}(x^\intercal 1_S)^2 \right ).
\end{align*}
For every $x \in T$, we take expectation over $S \in \mathcal S$ to get 

\begin{align*}
\mathbb E_{S \sim \mathcal S} \left [  \frac{f_v(x)}{f_0(x)} \right ]
&=( 1 + \alpha ) ^{-1/2} \cdot \frac{1}{(t/\sparsity)} \sum_{S}  \exp \left (  \frac{\alpha}{2(\alpha + 1)} \cdot \frac{1}{\sparsity} (x^\intercal 1_S)^2 \right ) \\
&\leq ( 1 + \alpha)^{-1/2} \cdot 
\frac{1}{(t / \sparsity)}
\left (
(t / \sparsity) (1 - \alpha)^{-1/2}  +
    \BBdelta
\right ) \tag{since $x \in T$} \\
&= ( 1+ \alpha)^{-1/2}  ( (1- \alpha)^{-1/2} + 1 ) \tag{since $t \geq \sparsity d^{\epsilon} \log (400 nd)$}.
\end{align*}
Since $\alpha$ is a constant, the above expression is bounded above by a constant, as desired.
\end{proof}

We restate and prove Claim \ref{claim:ub-sparse-pca} below.
\pcaupperbound*
\begin{proof}[Proof of Claim \ref{claim:ub-sparse-pca}]
For each block $R \in  \{ [1, \sparsity], [\sparsity + 1, 2 \sparsity], \dots, [d - \sparsity + 1 , d]  \} $ and each sample $j \in [m]$, we define the variable 
$Y_{R, j} = \sum_{i \in R} x_i^j$. 
Since the samples are drawn independently and the covariance matrix of the underlying Gaussian distributions are such that that coordinates from different $\sparsity$-sized blocks are independent, the variables $Y_{R,j}$ are all independent. We also recall properties of $Y_{R, j}$ given by Claim \ref{claim:pca-Y-gaussian}.

Our approach is to apply Bernstein's inequality to the random variables $Z_{R, j} = Y^2_{R, j} - \mathbb E[Y^2_{R, j}]$.
\begin{proposition}[Bernstein's inequality]
Let $Z_1, \dots, Z_N$ be independent, mean-zero, sub-exponential random variables. Then for every $t \geq 0$, we have 
\begin{align*}
    \Pr \left [ \left | \sum_{i =1 }^N Z_i \right |  \geq t  \right ] \leq 2 \exp \left [  - c \min \left ( \frac{t^2}{\sum_{i= 1}^n \| Z_i \|^2} ,\frac{t}{\max_i \|Z_i\|} \right ) \right ],
\end{align*}
where $c > 0$ is an absolute constant, and $\| X\| = \inf \{ K > 0 : \mathbb E [ \exp(|X| / K ) ] \leq 2 \}$.
\end{proposition}
It is straightforward to verify that the variables $Z_{R,j}$ are independent and mean-zero. We also remark that $Z_{R, j} \sim \Var(Y_{S, j})(X - 1)$ where $X$ is a chi-squared random variable with one degree of freedom. By standard properties of the chi-squared distribution, it follows that $Z_{R, j}$ is a sub-exponential random variable and $\| Z_{R, j} \| \leq \Var(Y_{R, j}) \| X - 1 \| = C_1 \Var(Y_{R,j})$, where $C_1 > 1$ is an absolute constant that depends on the the chi-squared distribution. 

Finally, we are ready to apply Bernstein's inequality. Let $\tau = nd + n\alpha \sparsity / 2$.
The probability that the test fails to detect $\Dn$ is at most 
\begin{align*}
    \Pr_{\Dn} \left [ \sum_{j} \sum_{R} Y_{R, j}^2 \geq \tau  \right ]
    &\leq \Pr_{\Dn} \left [ \sum_{j} \sum_{R} Z_{R, j} \geq \tau - \sum_j \sum_R \mathbb E[Y_{R, j}^2] \right ] \\
    &=  \Pr_{\Dn} \left [ \sum_{j} \sum_{R} Z_{R, j} \geq \tau - n( d/ \sparsity) \sparsity \right ] \tag{by Claim \ref{claim:pca-Y-gaussian}}\\ 
    &\leq \Pr_{\Dn}\left [ \left |  \sum_{j} \sum_{R} Z_{R, j} \right |  \geq n\alpha \sparsity / 2 \right ].
\end{align*}
Similarly, the probability that the test fails to detect $\Dp$ is at most 
\begin{align*}
    \Pr_{\Dp} \left [ \sum_{j} \sum_{R} Y_{R, j}^2 < \tau  \right ] 
    &\leq \Pr_{\Dp} \left [ \sum_{j} \sum_{R} Z_{R, j} < \tau - \sum_j \sum_R \mathbb E[Y_{R, j}^2] \right ] \\
    &=  \Pr_{\Dp} \left [ \sum_{j} \sum_{R} Z_{R, j} < \tau - n \left (  ( d/ \sparsity - 1) \sparsity + (1 + \alpha)\sparsity \right ) \right ] \tag{by Claim \ref{claim:pca-Y-gaussian}}\\ 
    &\leq \Pr_{\Dp}\left [ \left |  \sum_{j} \sum_{R} Z_{R, j} \right |  \geq n\alpha \sparsity / 2 \right ].
\end{align*}
In either case, taking $n \geq \log \left ( \frac{2}{\delta} \right ) \left [ \frac{4 C_1^2 (1 + \alpha)^2}{c \alpha^2 } \cdot \frac{d}{\sparsity} \right ] $  the failure probability is at most 
\begin{align*}
&2 \exp \left [  - c \min \left ( \frac{n^2(\alpha \sparsity / 2)^2}{n(d / \sparsity) \cdot  C_1^2((1 + \alpha ) \sparsity)^2} ,\frac{ n(\alpha \sparsity / 2) }{C_1(1 + \alpha) \sparsity} \right ) \right ] \\
={}& 2 \exp \left [  - \frac{c \alpha}{2 C_1 (1 + \alpha) } n\min \left ( \frac{\alpha \sparsity }{2C_1d (1 + \alpha ) } ,1 \right ) \right ] \\
\leq{}& 2 \exp \left [  - \frac{c \alpha}{2 C_1 (1 + \alpha) } n \cdot  \left ( \frac{\alpha \sparsity }{2C_1d (1 + \alpha ) } \right ) \right ] \tag{since $\alpha < 1$, $\sparsity \leq d$ and $C_1 > 1$} \\
={}& 2 \exp \left [  - \frac{c \alpha^2 \sparsity}{4 C_1^2 (1 + \alpha)^2 d } n \right ]  \\
\leq{} & \delta 
\end{align*}
as desired. 
\end{proof}